\documentclass[letterpaper, 11pt, onecolumn]{article}
\pdfoutput=1 %

\usepackage{lmodern,booktabs,authblk,tocloft}
\usepackage[margin=1in]{geometry}
\usepackage{graphicx}
\usepackage{pgfplots}
\pgfplotsset{compat=1.18}

\usepackage[english]{babel}
\usepackage[utf8]{inputenc}
\usepackage[useregional]{datetime2}
\DTMusemodule{english}{en-US}

\usepackage{amsmath,amsthm,amssymb,amsfonts,physics,dirtytalk,enumitem, thm-restate}

\newcommand{\prvar}{\phi}
\newcommand{\prstate}[1]{\ket{\prvar_{#1}}}

\newcommand{\tshovar}{\psi} 

\newcommand{\tshostate}[1] {\ket{\Psi_{#1}}}%

\newcommand{\bratshostate}[1] {\bra{\Psi_{#1}}}%

\newcommand{\dshostate}[1] {\ket{\tshovar_{#1}}} %
\newcommand{\bradshostate}[1]{\bra{\tshovar_{#1}}} %

\newcommand{\dshermite}[1] {\ket{\psi_{#1}}}

\newcommand{\apxdshostate}[1]{\ket{\overline{\tshovar_{#1}}}} %
\newcommand{\braapxdshostate}[1]{\bra{\overline{\tshovar_{#1}}}} %

\newcommand{\sommaconstant}{\gamma}

\newcommand{\wt}{\widetilde}
\newcommand{\wh}{\widehat}

\usepackage[dvipsnames]{xcolor}
\definecolor{newcolor}{hsb}{0.6,1,0.75}
\definecolor{lowcol}{RGB}{6, 172, 211}
\definecolor{midcol}{RGB}{188, 207, 183}    
\definecolor{highcol}{RGB}{220, 83, 24}     %

\usepackage{algorithm}
\usepackage[noend]{algpseudocode}

\usepackage{amsmath,amssymb,amsthm,amsfonts,latexsym,bbm,xspace,thm-restate}
\usepackage{graphicx,float,mathtools,braket,mathdots}
\usepackage{booktabs,forest,soul,comment}
\usepackage{bold-extra}

\usepackage{enumitem}
\setlist{
	listparindent=\parindent,
	parsep=0pt,
}

\usepackage[colorlinks = true,
            urlcolor=lowcol,
            linkcolor=lowcol,
	        citecolor=lowcol,
	      unicode]{hyperref}
\usepackage[capitalise,nameinlink,noabbrev]{cleveref}
\crefname{equation}{}{}

\newtheorem{theorem}{Theorem}
\newtheorem{lemma}[theorem]{Lemma}
\newtheorem{claim}[theorem]{Claim}
\newtheorem{proposition}[theorem]{Proposition}
\newtheorem{corollary}[theorem]{Corollary}
\newtheorem{conjecture}[theorem]{Conjecture}

\newtheorem{fact}[theorem]{Fact}
\crefname{fact}{Fact}{Facts}

\theoremstyle{definition}
\newtheorem{definition}[theorem]{Definition}

\let\Pr\relax
\DeclareMathOperator*{\Pr}{Pr}
\DeclareMathOperator*{\E}{\mathbb{E}}
\DeclareMathOperator*{\Ex}{\E}

\DeclareMathOperator*{\Var}{Var}

\DeclareMathOperator{\poly}{poly}
\DeclareMathOperator{\polylog}{polylog}

\DeclareMathOperator{\sgn}{sgn}
\DeclareMathOperator{\dist}{dist}

\newcommand{\tracedistance}[1]{\dist_{\text{tr}}\left(#1\right)}

\mathchardef\mhyphen="2D %

\newcommand{\calD}{\mathcal{D}}

\newcommand{\calN}{\mathcal{N}}

\newcommand{\calS}{\mathcal{S}}

\renewcommand{\abs}[1]{\left\lvert #1 \right\rvert}
\renewcommand{\norm}[1]{\left\lVert #1 \right\rVert}

\renewcommand{\commutator}[1]{\left[ #1 \right]}
\renewcommand{\anticommutator}[1]{\left\{ #1 \right\}}

\makeatletter
\DeclareRobustCommand\bfseries{%
  \not@math@alphabet\bfseries\mathbf
  \fontseries\bfdefault\selectfont\boldmath}
\makeatother

\newcommand{\R}{\mathbb{R}}
\newcommand{\C}{\mathbb{C}}

\newcommand{\N}{\mathbb{N}}
\newcommand{\Z}{\mathbb{Z}}

\newcommand{\eps}{\varepsilon}

\newcommand{\pmone}{\{-1,+1\}}

\newcommand{\gaussianpdf}{\nu}
\newcommand{\distortion}{\kappa}

\def\rD{{\rm d}}

\def\ri{{\rm i}}

\def\cO{\mathcal{O}}

\def\cU{\mathcal{U}}

\def\cV{\mathcal{V}}

\def\Pr{\mathrm{Pr}}

\def\poly{\mathrm{poly}}

\def\one{{\mathchoice {\rm 1\mskip-4mu l} {\rm 1\mskip-4mu l} {\rm
1\mskip-4.5mu l} {\rm 1\mskip-5mu l}}}

\newcommand{\position}{\hat{x}}
\newcommand{\momentum}{\hat{p}}
\newcommand{\hamiltonian}{\hat{\mathrm{H}}}

\newcommand{\discreteposition}{\overline x}
\newcommand{\discretemomentum}{\overline p}
\newcommand{\discretehamiltonian}{\overline{\mathrm{H}}}
\newcommand{\expbound}[1]{\exp(-#1)}
\newcommand{\degree}{D}
\newcommand{\lowenergy}{N}
\newcommand{\lowenergyprojector}{\Pi_{\lowenergy}}
\newcommand{\mediumenergy}{\lowenergy^\prime}
\newcommand{\mediumenergyprojector}{\Pi_{\mediumenergy}}
\newcommand{\nestedcommutator}[3]{\commutator{#1,#2}_{#3}}
\newcommand{\poisson}[1]{\text{Poisson}_{#1}}

\title{
    \huge
    Quantum Hermite Transform and\\ Gaussian Goldreich-Levin
}

\title{Efficient Quantum Hermite Transform}

\author[1,2,*]{Siddhartha~Jain}
\author[2,*]{Vishnu~Iyer}
\author[1]{Rolando~D.~Somma}
\author[3,4]{Ning~Bao}
\author[1]{Stephen~Jordan}

\affil[1]{\footnotesize Google Quantum AI, Venice, CA 90291, United States} 
\affil[2]{The University of Texas at Austin, Austin, TX 78712, United States}
\affil[3]{Northeastern University, Boston, MA 02115, United States}
\affil[4]{Brookhaven National Laboratory, Upton, NY 11973, United States}
\affil[*]{These authors contributed equally to this work. \vspace{-24pt}}

\date{}

\renewcommand{\th}{^{\textrm{th}}}

\def\rD{{\rm d}}

\def\ri{{\rm i}}

\def\cO{\mathcal{O}}

\def\cU{\mathcal{U}}

\def\cV{\mathcal{V}}

\def\Pr{\mathrm{Pr}}

\def\poly{\mathrm{poly}}

\def\one{{\mathchoice {\rm 1\mskip-4mu l} {\rm 1\mskip-4mu l} {\rm
1\mskip-4.5mu l} {\rm 1\mskip-5mu l}}}
\def\dqc1{\textsc{DQC1}}

\newcommand{\byref}[1]{\text{(#1)}\nonumber}

\begin{document}

\maketitle

\begin{abstract}
We present a new primitive for quantum algorithms
that implements a discrete Hermite transform efficiently, in time that depends logarithmically in both the dimension and the inverse of the allowable error. 
This transform, which maps basis states to states whose amplitudes are proportional to the Hermite functions,  
can be interpreted as the Gaussian analogue of the Fourier transform.
Our  algorithm is based on 
a method to exponentially fast forward the evolution of the quantum harmonic oscillator, which significantly improves over prior art.
We apply this Hermite transform to give examples of provable quantum query advantage in property testing and learning.
In particular, we show how to efficiently test the property of being close to a low-degree in the Hermite basis when inputs are sampled from the Gaussian distribution, and how to solve a Gaussian analogue of the Goldreich-Levin learning task efficiently. We also comment on other potential uses of this transform to simulating time dynamics of quantum systems in the continuum.
\end{abstract}

\newpage
\setlength{\cftbeforesecskip}{2pt} 
\setlength{\cftbeforesubsecskip}{1pt} 
\setcounter{tocdepth}{2}
\tableofcontents

\newpage
\section{Introduction}

Hermite polynomials and their corresponding Hermite functions are central in physics, signal analysis, statistics, and beyond. For example, Hermite functions are the well-known eigenfunctions of one of the fundamental models in quantum mechanics known as the quantum harmonic oscillator (QHO)~\cite{griffiths2018introduction}.  Hermite functions can also be used to approximate arbitrary functions or probability distributions, since they form an orthonormal basis of functions weighted by a Gaussian envelope (cf.~\cite{blinnikov1998expansions}). 
Therefore, an efficient Hermite transform that performs a change of basis 
to the Hermite functions would provide a fundamental tool to quickly switch into the ``natural'' basis of many quantum systems and functions.
Moreover, the number of primitives from which most quantum algorithms are constructed, such as the quantum Fourier transform (QFT) and amplitude amplification, is relatively limited~\cite{NCbook}. This scarcity restricts the range of problems for which quantum algorithms offer speedups, thereby motivating the search of other novel quantum primitives. Given that the QFT represents a quantization of a pre-existing classical numerical transform, it is logical to attempt to quantize other classical transforms, such as the Hermite transform \cite{LRP08}, and to consider its applications in quantum computing.

This work is concerned with a quantum Hermite transform (QHT): 
constructing a quantum circuit that maps basis states into `Hermite' states. These Hermite states are superposition states with amplitudes given by the Hermite functions. While formally a Hermite transform would be defined in the continuum (i.e., a mapping $L^2(\mathbb R) \rightarrow L^2(\mathbb R)$), 
we consider a discrete and finite-dimensional version that can be implemented by a quantum circuit. 

For a QHT to be useful and provide quantum advantage, we require it to be efficient and hence implemented using a quantum circuit of complexity that is logarithmic in the dimension of the Hilbert space on which it acts. It is also desirable to 
achieve complexity scaling polylogarithmically in the inverse of the allowable error. Remarkably, we are able to achieve this: our main result is a quantum circuit $U$ of $\cO((\log N + \log 1/\eps)^3 \times \log(1/\eps))$ gates that performs the desired basis change into the Hermite states labeled by $n \in \{0,\ldots,N\}$ to within additive error $\eps$. Note that a classical discrete Hermite transform would require time polynomial in $N$ to be implemented. Our result is then comparable to the QFT, in that the QFT also transforms into another `natural' basis given by Fourier states, in time logarithmic in $N$.

Throughout this manuscript, we primarily consider the one-dimensional Hermite transform, acting on an $N$-dimensional Hilbert space. That is, we consider the position basis on a line, discretized as a one-dimensional lattice of $N$ points, and convert to the the Hermite functions of degree zero to $N-1$, which are also interpretable as the $N$ lowest-energy eigenstates of the quantum Harmonic oscillator. In some contexts, particularly property testing of oracles, it is also useful to apply the $n$-dimensional Hermite transform. This is the $n$-fold tensor power of the one-dimensional Hermite transform, much as the $n$-dimensional Fourier transform is the $n$-fold tensor power of the one-dimensional Fourier transform.

The existence of an efficient QHT is also a fundamental algorithmic question in Harmonic analysis, since it allows one to perform Hermite sampling. That is, given a quantum state $\sum_{\vec{x}} f(\vec{x}) \ket{\vec{x}}$ for some function $f:\R^n \to \R$, one can use the ($n$-dimensional) QHT to sample from the probability distribution $\mathrm{Pr}(\vec{k}) = |\wh{f}(\vec{k})|^2$, where $\wh{f}(\vec{k})$ is the coefficient of the Hermite polynomial with index $\vec{k}$ in the expansion.  
In the finite field setting of QFT, the analogous Fourier sampling subroutine features in many of the most prominent quantum algorithms~\cite{Simon97,Shor97,Regev09,Aar10,RazTal22,YZ24,jordan2025}.  Using the QHT we expect that one can design continuous analogues of problems with quantum query advantage like Forrelation~\cite{AA18,wu2020stochasticcalculusapproachoracle}. We also expect QHT to be useful for learning geometric concepts over the Gaussian distribution~\cite{KlivansOS08}, and simulating nonlinear differential equations with Gaussian noise which wecan be done by projecting onto low-degree Hermite polynomials~\cite{bravyi2025quantumsimulationnoisyclassical}. A recent work of Marwaha et al.~\cite{marwaha2025complexitydecodedquantuminterferometry} also showed that the recent DQI algorithm~\cite{jordan2025} can be thought of as performing a \emph{Kravchuk} transform, which is a particular discretization of the Hermite transform. While closely related, we do not believe our work easily gives an efficient Kravchuk transform, which would scale logarithmically in the degree.

\subsection{Problem statement}

In the continuum, a Hermite transform is a basis transformation that maps into the Hermite functions defined as
\begin{align}
    \psi_n(x):=\frac {(-1)^n} {\sqrt{2^n n! \sqrt \pi}}
    e^{-x^2/2}H_n(x) \;,
\end{align}
where $H_n(x)$ is the $n\th$ Hermite polynomial\footnote{There are two widely used normalizations for the Hermite polynomials: the physicist's Hermite polynomials and the probabilist's Hermite polynomials. We use the physicist's normalization throughout.}.
The goal of a QHT is then to construct a quantum circuit that performs a similar basis transformation but in a discrete, finite dimensional space. To this end, we introduce the Hermite states
\begin{align}
\label{eq:dshermite}
    \dshermite{n}:=\left( \frac{2\pi}M\right)^{1/4}\sum_{j=-M/2}^{M/2-1} \psi_n(x_j) \ket j \in \mathbb C^M \;,
\end{align}
where $x_j:=j \sqrt{2\pi/M}$ and $M=2^m$ is the dimension. These can be interpreted as discretizations of the Hermite functions, where the discretization size is 
$\sqrt{2\pi/M}$. We labeled the $M$ basis states of $\mathbb C^M$ from $-M/2$ to $M/2-1$ for convenience. Note that the above states both discretize and truncate the Hermite functions. The above choice of lattice spacing ensures that the truncation only cuts off a region in which the first $M$ Hermite functions all have exponentially decaying tails.

The problem of a QHT is to find a quantum circuit that maps $\ket n \mapsto \dshermite{n}$. By taking the gate-by-gate inverse one also obtains a quantum circuit for the inverse quantum Hermite transform, of the same complexity.

\paragraph{Quantum Hermite Transform Problem.}
Let $N>0$ be the dimension 
and $\eps >0$ the error. The goal of a QHT is to construct a quantum circuit $U$ that performs the map
\begin{align}
\label{eq:QHT}
    \sum_{n=0}^{N-1} \alpha_n \ket n \mapsto 
     \sum_{n=0}^{N-1} \alpha_n \dshermite{n} 
\end{align}
within additive error $\eps$,
where $\alpha_n \in \mathbb C$ are arbitrary and satisfy $\sum_{n=0}^N |\alpha_n|^2=1$.
The quantum circuit acts on a Hilbert space of dimension $M \ge N$.

\vspace{0.2cm}
In principle, such transformations cannot be exactly unitary since the states $\dshermite{n}$ are not exactly orthonormal. However, for sufficiently large $M=\Omega(N)$, they can be shown to be almost orthonormal within error that is exponentially small in $M$, due to~\cite{somma2016}. Hence, by fixing the error $\eps$ and the dimension $N$ where the QHT occurs, we can choose $M$ properly and satisfy~\cref{eq:QHT} within the precision requirements. For our specific construction, we will require $M=\poly(N,1/\eps)$. 

\subsection{Summary of results}

Our QHT relies on a novel result to fast-forward the evolution of the QHO. 
In this context, a bounded Hamiltonian $H$ is said to be fast-forwardable if one can construct a quantum circuit $W$ that approximates the time-evolution operator $e^{-\ri H t}$ for time $t$ as $\| W - e^{-\ri  H t}\| \leq \eps$, using a number of quantum gates that scales sublinearly in $\|H\| t$~\cite{AA17,GSS21}. Such result would bypass a number of no-fast-forwarding theorems that show that this is not always possible \cite{BACS07, CK10, AA17, HHK21}. 
However, the lower bound $\Omega (\|H\| t)$
applies only to the worst case and, in contrast, 
quantum systems like the QHO can indeed be fast-forwarded.  In the continuum, the QHO models a particle in a quadratic potential with Hamiltonian 
$\hamiltonian = (\position^2 + \momentum^2)/2$, where
$\position$ is the position operator and $\momentum$ is the momentum operator. 
We introduce a discretized, $M$-dimensional version of the QHO with Hamiltonian $\discretehamiltonian$, and show the following.

\begin{restatable}[Exponential fast-forwarding of the QHO]{theorem}{fastforwarding}
\label{thm:fastforwarding}
Let $N$ be the target dimension for fast-forwarding, 
 $\discretehamiltonian \in \mathbb C^{M \times M}$ be the Hamiltonian of a discrete QHO
, and $t \in [-\pi, \pi]$.
Then, we can choose $M =\Theta( N \log N)$ such that the evolution operator $e^{-\ri \discretehamiltonian t}$
   can be simulated within error $\cO(\exp(-\lowenergy/10))$  in the subspace spanned by the first $N$ eigenvectors of $\discretehamiltonian$ using ${\cO}(\log^2 \lowenergy)$ gates.
\end{restatable}

Note that we don't lose any generality due to the condition $t \in [-\pi,\pi]$ since the eigenvalues of the Harmonic oscillator are of the form $n+1/2$ with integer $n$, and hence $e^{-iH(t + 2 \pi)} = -e^{-iHt}$. Since $\|\discretehamiltonian\|=\Theta(M)$ and $N = \poly(M)$,
\cref{thm:fastforwarding} shows that exponential fast-forwarding of the QHO is possible. 
This result significantly improves over prior art~\cite{somma2016}, where a quantum algorithm for simulating  $e^{-\ri \discretehamiltonian t}$ of complexity $\sim \exp(\sqrt{\log N})$ was given, only describing a form of subexponential fast-forwarding. \cref{thm:fastforwarding} is based on a 
succinct factorization of the evolution operator 
in the continuum that can be extended to the evolution operator in the finite-dimensional case, and contrasts
the method in~\cite{somma2016} based on product (Trotter-Suzuki) formulas.
The complexity in~\cref{thm:fastforwarding}
is dominated by that of multiplication using coherent arithmetic, which is $\cO(\log^2 N)$ using schoolbook 
multiplication, and might be improved using sophisticated techniques.

The ability to fast forward a Hamiltonian also
allows one to perform tasks like quantum phase estimation (QPE) more efficiently~\cite{AA17,GSS21}.  
Our efficient QHT relies on fast QPE and hence builds upon~\cref{thm:fastforwarding}.

\begin{theorem}[Efficient QHT, informal]
\label{thm:informal qht}
 There exists a quantum circuit of complexity 
 $\cO( (\log N + \log(1/\eps))^3\times \log(1/\eps))$ 
 that can perform an $\eps$-approximate QHT of dimension $N$. 
\end{theorem}

Being equipped with an efficient QHT, we then consider some problems where it can be applied.

\paragraph{Property testing and learning.}
One main application of an efficient QHT is Hermite sampling, discussed in detail in \cref{sec:apps}. In this setting we are given black-box access to a function $f:\mathbb R^n \rightarrow \mathbb R$ and we want to sample from its Hermite spectrum. To do so, we must first use queries to the black box to construct a state whose amplitudes are proportional to $f$. Then we use the (inverse) QHT $n$ times, to map Hermite states to basis states.

Using Hermite sampling as a subroutine, we are also able to show that quantum algorithms can solve all the following tasks efficiently over the Gaussian distribution:
 
\begin{enumerate}[label=(\arabic*)]
    \item Test if $f$ is close to a product of $k$ sign functions;
    \item Test if $f$ is close to being degree $d$;
    \item Test if $f$ is close to a Hermite polynomial;
    \item Agnostically learn what we call sparse concepts.
\end{enumerate}

We describe quantum algorithms for (1--3) in \cref{sec:prop-testing} and give the learning application (4) in \cref{sec:learning}. We also prove that we have quantum advantage for tasks (1) and (2) in \cref{sec:prop-testing}. A recent independent work of Lewis et al.~\cite{lewis2025} also gave quantum advantage in learning real-valued functions, but our work is the first to show quantum advantage for property testing with real-valued functions.

We note that these quantum advantages are already attained in the limit of large $n$ with $N = \poly(n)$. Consequently, the polylogarithmic scaling in $N$ that we achieve with our QHT circuit only achieves a polynomial quantum advantage in circuit size for property testing. Now that we have shown that the quantum Hermite transform can be efficiently applied at exponentially large $N$, we invite others to use this as a building block for new quantum algorithms achieving exponential advantage.

\paragraph{Hamiltonian simulation.}
Hamiltonian simulation is the problem of simulating quantum dynamics, induced by a Hamiltonian $H$ and for time $t$, on a quantum computer. In particular, the Hermite functions represent what in physics are known as Fock states. Such states appear naturally in many quantum systems, where the QHO is only one example, and we expect the QHT to be generally applicable for Hamiltonian simulation in these systems.

Quantum systems defined in the continuum involve the position $\hat x$ and momentum $\hat p$ operators. These become sparse when represented using the Fock states or Hermite functions. Our QHT can then perform the change of basis efficiently, allowing us to work with sparse representations of the Hamiltonians. As the most efficient quantum algorithms for quantum dynamics or Hamiltonian simulation assume the sparse matrix access model (cf.~\cite{berry2015simulating,low2017optimal}), 
our efficient QHT might be used to reduce the complexity of Hamiltonian simulation in these systems. 

Similarly, the QHT is expected to bring novel examples of fast-forwarding. Consider for example the well-known Jaynes-Cummings model that models the interaction of a two-level atom with the electromagnetic field (in the continuum)~\cite{jaynes2005comparison}.
When described in the Fock basis, the Hamiltonian simply becomes a direct sum of $2 \times 2$ blocks that can be simply diagonalized. Hence,
our QHT can also exponentially fast-forward the discrete version of this model.

\vspace{0.2cm}
Like the QFT, our efficient QHT opens the path to other case studies and potential applications. To this end, in \cref{sec:open} we list open problems, charting the path for more examples of quantum advantage.

\section{Technical overview}
\label{sec:overview}

Our main result is an efficient quantum circuit for an $M$-dimensional unitary $U$, that transforms from the computational basis to the basis of Hermite states given in~\cref{eq:dshermite}, with arbitrary accuracy. 
Hermite states are also accurate approximations of the eigenvectors of the `discrete' QHO~\cite{somma2016}. That is, we approximate the transformation defined by $U \ket{n} \mapsto \dshermite{n}$, where the Hermite states $\dshermite{n}$ are a discretization of the $n\th$ eigenstate of the QHO Hamiltonian in the continuum $\hamiltonian = (\position^2 + \momentum^2)/2$.
Here, $\position$ is the position operator and $\momentum$ is the momentum operator: $\position f(x) = x f(x)$ and $\momentum f(x)=  -\ri \frac{\rD}{\rD x}f(x)$, for arbitrary $f(x) \in L^2(\mathbb R)$.

The steps in our efficient QHT are as follows. First, given index $n \in \{0,\ldots,N\}$ for some $N>0$, we prepare the Hermite state $\dshermite{n} \in \mathbb C^M$ while keeping a copy of $\ket n$ in some register. The actual dimension $M=2^m$ satisfies $M>N$ and has to be chosen properly, as we discuss. This is done via a unitary transformation conditional on $\ket n$. To reset $\ket n$ to $\ket 0^{\otimes m}$, i.e., to `uncompute' $\ket n$, we infer the value of $n$ from $\dshermite{n}$. That is, the sequence of state preparation steps to define the QHT is
\[
\begin{array}{lcll}
\ket{n} & \to & \ket{n}\ket{0}^{\otimes m} & \textrm{adjoin $m$ ancillas} \\
 & \to & \ket{n}\dshermite{n} & \textrm{Hermite state preparation} \\
 & \to & \ket{0}^{\otimes m}\dshermite{n} & \textrm{uncomputation} \\
 & \to & \dshermite{n} & \textrm{discard $m$ ancillas.}
\end{array}
\]
We give an overview of how the main steps are accomplished.

\paragraph{Hermite state preparation.} 

Given $n$, to prepare the state $\dshermite{n} \in \mathbb C^M$ we proceed in two steps: First, we prepare an approximation $\prstate{n} \in \mathbb C^M$ such that $\braket{\phi_n | \psi_n} = \Omega(1)$, and second, we use fixed-point quantum search~\cite{YLC14} to transform $\ket{ \phi_n} \mapsto \dshermite{n}$ efficiently. Amplitude amplification requires a sequence of reflections over $\dshermite{n}$ and $\prstate{n}$. The latter can be efficiently implemented
because the state $\prstate{n}$ can also be efficiently prepared, but the reflection over $\dshermite{n}$ is more intricate. To this end, we design a `filtering' algorithm that flags $\dshermite{n}$ and allows us to simulate this reflection by a reflection over a state $\ket 0$ of an ancilla qubit. This filtering is based on fast QPE; it uses our fast-forwarding result for the QHO and can be implemented efficiently.  

The states $\prstate{n}$ are chosen as follows. A useful result for this task is that of Plancherel–Rotach asymptotics~\cite{szeg1939orthogonal}, which gives approximate expressions for the Hermite functions $\psi_n(x)$ with a provable error bound. Such approximations are only valid in the domain $|x|<\sqrt{2n+1}$, where the points $\pm x_{\rm tp}=\pm \sqrt{2n+1}$ are known as the `turning' points in physics. Explicitly, these approximations read:
\begin{align}
    \psi_n(x) =  \frac {(-1)^n 2^{\frac 1 4}}{\pi^{\frac 1 2}n^{\frac 1 4}}\frac 1 {\sqrt{\sin \varphi(x)}}\left( \sin \left[\left ( \frac n 2 + \frac 1 4 \right) (\sin (2 \varphi(x))-2 \varphi(x)) + \frac{3 \pi}4\right] + \cO \left(\frac 1 n \right)\right) \;,
\end{align}
where $\varphi(x):=\arccos(x/\sqrt{2n+1})$ is such that $\pi > \varphi(x) > 0$.

We show how to use this approximation formula to construct the states $\prstate{n}$ efficiently, which are now simply superposition states
with amplitude modulated by $\frac 1 {\sqrt{\sin \varphi(x)}}$ and also a phase that depends on $\varphi(x)$. We also show that the states $\prstate{n}$ have nonzero, constant overlap with the $\dshermite{n}$.

\begin{figure}[ht!]
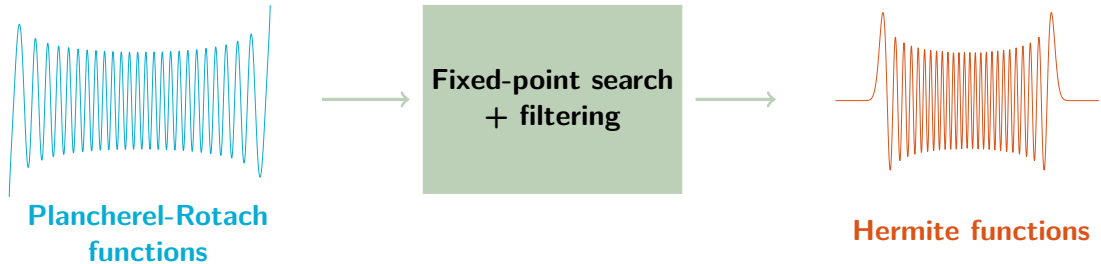

    \centering
    \begin{tikzpicture}[every node/.style={inner sep=0}]
  \node (L) at (0,0) {\resizebox{0.3\textwidth}{!}{\input{figures/hermite_asymptotic_plot.pgf}}};

  \node[draw, fill=midcol, draw=midcol, rectangle, minimum width=2.6cm, minimum height=2.5cm, align=center, inner sep=3pt, font=\bfseries\sffamily] (B) at (5.55,0) {Fixed-point search \\ + filtering};

  \node (R) at (11,0) {\resizebox{0.3\textwidth}{!}{\input{figures/hermite_function_plot.pgf}}};

  \draw[->, line width=1.2pt, color=midcol, shorten >=5pt] (L.east) -- (B.west);
  \draw[->, line width=1.2pt, color=midcol, shorten <=5pt] (B.east) -- (R.west);

\node[anchor=north west, text width=0.25\textwidth, align=center, %
      font=\sffamily\bfseries\normalsize, text=lowcol]
  at ($(L)+(-1.9,-1.4)$) {Plancherel-Rotach\\ functions};

\node[anchor=north east, text width=0.3\textwidth, align=center, %
      font=\sffamily\bfseries\normalsize, text=highcol]
  at ($(R)+(2.6,-1.6)$) {Hermite functions};
\end{tikzpicture}
    \caption{A visualization of the Quantum Hermite Transform. On the left we have the Plancherel-Rotach functions (visualized without the decaying envelope for simplicity), we prepare states corresponding to these, which is then fed into the fixed-point search algorithm using our fast-forwarding result as a subroutine. This allows us to prepare the states corresponding to the Hermite functions in superposition to arbitrary precision. Both the functions are visualized for $n=10$.}
\end{figure}

\paragraph{Uncomputation.} For the uncomputation step where we reset $\ket n \mapsto \ket 0^{\otimes m}$, we use a version of Kitaev's QPE algorithm~\cite{K95}. The Hermite functions are known eigenfunctions of the QHO Hamiltonian $\hamiltonian$ of eigenvalue $n+1/2$. Similarly, the Hermite states $\dshermite{n}$ can be shown to be accurate approximations of the eigenvectors $\dshostate{n}$ of the discrete QHO $\discretehamiltonian$, with eigenvalues that are very close to $n+1/2$~\cite{somma2016}. For a bounded Hamiltonian $H$ satisfying $\|H\|\le 1$,
QPE is able to distinguish eigenvalues separated by a difference $\Delta$
by implementing unitary time evolutions of the form $e^{-\ri H t}$ up to evolution times $t \sim 1/\Delta$.
In our case, if $H=\discretehamiltonian/M$ is the rescaled version of $\discretehamiltonian$, we have $\Delta \sim 1/M$ and hence $t \sim M$. Then, to avoid any complexity polynomial in $M$, we use~\cref{thm:fastforwarding} that allows us to fast forward this time evolution and to perform fast QPE.

\vspace{0.5cm}

Last, we give an overview for the proof of~\cref{thm:fastforwarding}, which is used in the state preparation for both filtering and uncomputation.

\paragraph{Fast-forwarding.}
Central to this result is a factorization of the evolution operator $e^{-\ri \hamiltonian t}$, where $\hamiltonian=(\momentum^2+\position^2)/2$ is the QHO Hamiltonian in the continuum. We exploit the canonical commutation relations, namely $\commutator{\position,\momentum} = \ri I$, where $I$ is the identity such that $I f(x)=f(x)$. Then~\cite{QA07}
\begin{align}
    \label{eq:QHOevolution}
\exp(- \ri \hamiltonian t) = \exp(-\frac{ \ri \tan(t/2)\momentum^2}{2}) \exp(-\frac{\ri \sin(t)\position^2}{2}) 
\exp(-\frac{\ri \tan(t/2)\momentum^2}{2})\; .
\end{align}
We then consider a discrete version of the QHO, referred to as
$\discretehamiltonian$, and analyze the time complexity and the discretization error of this factorization. To this end,
we define $\discreteposition$ and $\discretemomentum$ to be a specific choice of discretizations of $\position$ and $\momentum$, defined in \cite{somma2016}, and let $\discretehamiltonian :=(\discretemomentum^2 + \discreteposition^2)/2$.
We show that the factorization of~\Cref{eq:QHOevolution}, when replacing $\position \rightarrow \discreteposition$
and $\momentum \rightarrow \discretemomentum$,
reproduces the evolution operator $e^{-\ri \discretehamiltonian t}$
within accuracy that is exponentially small in $N$ in the low-energy subspace, where $N \le M=\poly(N)$, and $M$ is the Hilbert dimension of $\discretehamiltonian$.
The low-energy subspace is that spanned by the $N$ eigenvectors of $\discretehamiltonian$ of smallest eigenvalue.

To prove this result, we consider the commutation relations of the discrete operators $\discreteposition$ and $\discretemomentum$ and show they approximate those between $\position$ and $\momentum$, in the low-energy subspace, the space spanned by the first $\lowenergy$ eigenvectors of $\discretehamiltonian$. For example, in that subspace, $ \commutator{\discreteposition,\discretemomentum} - \ri I$ can be shown to have norm exponentially small in $N$, where $I$ is now the $N$-dimensional identity.
We also need to bound the errors of the \emph{nested} commutators.
Let $\nestedcommutator{\discreteposition^2}{\discretemomentum^2}{k} = \commutator{\discreteposition^2, \nestedcommutator{\discreteposition^2}{\discretemomentum^2}{k-1}}$ be the $k$-th nested commutator, where $\nestedcommutator{\discreteposition^2}{\discretemomentum^2}{1} = \commutator{\discreteposition^2,\discretemomentum^2}$.
We cannot directly use the fact that the norm of 
$\commutator{\discreteposition,\discretemomentum} - \ri I$ is small on the low-energy subspace, because when nesting one would naively incur a cost which scales as $\approx M^{k-1}$ due to the multiplication of $k-1$ matrices after using the relation once. We show how to control the growth of these terms within the low-energy subspace by proving that (i) polynomials in the operators with degree $t$ have norm which scale as $\lowenergy^t$ in the low-energy subspace, and (ii) certain operators which capture the discretization errors have norm $\poly(M)$ on the entire space whose growth can be controlled by $t!$ when $t$ is sufficiently greater than $N$. See \cref{sec:fastforwarding} for more details. We believe this technique would be useful for more Hamiltonian simulation results as well.

The approximated evolution operator of the discrete QHO also involves three exponentials. The exponential of $\ri \sin(t) \discreteposition^2/2$ is a diagonal unitary that can be simply implemented with $\cO(\log^2 N)$ two-qubit arbitrary gates regardless of $t$. (For simplicity throughout this paper we assume that arbitrary gates can be implemented perfectly, without any overhead in precision.) 
To this end, we use coherent arithmetic for multiplication and then use phase kickback. 
Since $\discretemomentum$ is obtained from $\discreteposition$ via conjugation with the (centered) QFT, the exponentials involving $\discretemomentum^2$ can also be implemented with $\cO(\log^2 N)$ two-qubit arbitrary gates. This is then an example of exponential fast-forwarding, since it avoids the no-fast-forwarding bound $\Omega (\|H\|t)$, which would be linear in $Nt$ in this case.

\section{Fast-forwarding of the QHO}
\label{sec:fastforwarding}
In this section we discuss~\cref{thm:fastforwarding} and present the formal results on the evolution of the QHO. First, let us discuss the algebraic factorization of the time-evolution of the quantum harmonic oscillator in the infinite-dimensional case. Since the Hamiltonian is the sum of two terms, $\hamiltonian = (\position^2 + \momentum^2)/2$, if the operators were commuting, we could simply factor by breaking up the sum. But the position and momentum operator do not commute. What we do know is the canonical commutation relation between them, 
$[\position, \momentum] = \ri I$. Combined with the Baker–Campbell–Hausdorff formula, this allows us to factor the time-evolution of the Hamiltonian. First, we notice these relations which can be checked by computation:
\begin{enumerate}
    \item $[\position^2, \momentum^2] = 2 \ri \anticommutator{\position,\momentum}$
    \item $[\position^2, \{\position,\momentum\}] = 4\ri\position^2$
    \label{item:pos}
    \item $[\momentum^2, \{\position,\momentum\}] = -4\ri\momentum^2$
    \label{item:mom}
\end{enumerate}

Further, \cref{item:pos} implies that $\commutator{\position^2, \commutator{\position^2, \anticommutator{\position,\momentum}}} = 0$. \Cref{item:mom} implies the symmetric statement about momentum. One thus finds that the Lie algebra generated by $\hat{x}^2$ and $\hat{p}^2$ is three dimensional, namely the span of $\hat{x}^2$, $\hat{p}^2$, and $\{\hat{x},\hat{p}\}$.
These facts about the nested commutators `kill off' infinite terms in the BCH expansion and give us the following factorization for the evolution operator, also discussed in \cite{QA07}.

\begin{restatable}[Factorization of the QHO evolution~\cite{QA07}]{theorem}{factorization}
The evolution operator of the QHO admits the following factorization:
\label{thm:factor}
\begin{align}
\label{eq:factor}
    \exp(-\ri \hamiltonian t) = \exp(-\frac{ \ri \tan(t/2)\momentum^2}{2}) \exp(-\frac{\ri \sin(t)\position^2}{2}) 
\exp(\frac{-\ri \tan(t/2)\momentum^2}{2})\;.
\end{align}
\end{restatable}

For completeness, we include an independent proof of this in \cref{sec:qho-factorization}.

Now, if we want to fast-forward on a digital quantum computer, then we need to pick a discretization of the infinite dimensional Hamiltonian. We use the same definition of the discretization as the previous work on fast-forwarding the QHO~\cite{somma2016}. This sets $\discreteposition$ to be a diagonal matrix of size $M \times M$, and $\discretemomentum$ as its conjugation by a centered discrete Fourier transform. Then we can set $\discretehamiltonian = (\discreteposition^2 + \discretemomentum^2)/2$. It is not at all clear a priori that this factorization is still faithful. In particular, it is no longer true that $\commutator{\discreteposition, \discretemomentum} = iI$. Remarkably, we are able to carry the
factorization to the case of the discrete QHO. Moreover, our error is \emph{doubly-exponentially} small! We recall \cref{thm:fastforwarding} below.

\fastforwarding*

We now formally define our discretized QHO. We use the same conventions as \cite{somma2016}, but change the notation slightly.

\paragraph{Discrete QHO.} 
For given dimension $M>0$, 
we discretize the space with a grid of
discretization size $\sqrt{ 2\pi/{M}}$. 
 We define the discretized position operator as
\begin{align}
\discreteposition := \sqrt{\frac{2\pi}{M}}\begin{pmatrix}
    -\frac{M}{2} & 0 & ... & 0 \\
    0 & -\frac{M}{2} + 1 & ... & 0 \\
    \vdots & \vdots & \ddots & \vdots \\
    0 & 0 & ... & \frac{M}{2}-1
\end{pmatrix} \;,
\end{align}
and the discrete momentum operator $\discretemomentum := F^{-1}\discreteposition F$, where $F$ is the centered $M$-dimensional discrete Fourier Transform. (The centered Fourier transform is equivalent to the standard discrete Fourier transform up to relabeling the indices $j$, as described in~\cref{app:cQFT}.) Then, the discrete QHO Hamiltonian is
\begin{align}
\label{eq:discretehamiltonian}
\discretehamiltonian = \frac{1}{2}\left(\discreteposition^2 + \discretemomentum^2\right) \;.
\end{align}

We let $\dshostate{n}$ be the eigenstates of $\discretehamiltonian$, for $0 \le n \le M-1$.
In~\cite{somma2016} it was shown that 
$\dshostate{n}$ are very close to the Hermite states 
\begin{align}\label{eq:hermite-state}
\dshermite{n} := \left( \frac{2\pi} M\right)^{1/4} \sum_{j=-M/2}^{M/2-1} \tshovar_n(x_j)\ket{j} \;,
\end{align}
 where $\psi_n(x)$ is the $n\th$ Hermite function.
Below, we collect some facts about this discretization. 
We introduce $\tshostate{n}$ as the bra-ket representation of the $n$th continuum QHO eigenstate.
\begin{fact}[\cite{somma2016}]\label{fact:somma-discretization-facts}
    There exist constants $\sommaconstant, c \in (0,1)$ such that for all $M$ sufficiently large and all $k,\ell \leq cM$,
    \begin{enumerate}
        \item $\abs{\langle {\tshovar_k}\! \dshostate{\ell} - \delta_{k,\ell}} \leq \exp(-\sommaconstant M)$.
        \item For $a,b \leq 4$, $\abs{\braapxdshostate{k}\discreteposition^a \discretemomentum^b\apxdshostate{\ell} - \bratshostate{k}\position^a \momentum^b\tshostate{\ell}} \leq \exp(-\sommaconstant M)$.
    \end{enumerate}
\end{fact}
In other words, this discretization is highly effective in a subspace with bounded energy. We will prove that in a subspace within this one, we can perform the fast-forwarding with provably small error, as visualized in \cref{fig:ff energy cutoff}.

\begin{figure}[ht!]
    \centering
    \resizebox{0.9\textwidth}{!}{\begin{tikzpicture}[
    x = \textwidth/16, 
    y = \textwidth/16, %
    every node/.style={font=\sffamily},
]
  
\def\barW{16.0}
\def\barH{0.9}
\def\lowW{5.0}
\def\cut{11.0}

\coordinate (barC) at (0,0);
\path (barC)++(-\barW/2,0) coordinate (L);
\path (barC)++(\barW/2,0)  coordinate (R);

\fill[rounded corners=1mm, lowcol]  (L) rectangle ++(\lowW,\barH);
\fill[rounded corners=1mm,midcol] ($(L)+(\lowW,0)$) rectangle ($(L)+(\cut,\barH)$);
\fill[rounded corners=1mm, highcol] ($(L)+(\cut,0)$) rectangle ($(R)+(0,\barH)$);

\node[anchor=north, font=\sffamily\large] at ($(L)+(0,-0.2)$) {$0$};
\node[anchor=north, font=\sffamily\large] at ($(L)+(\lowW,-0.2)$) {$N$};
\node[anchor=north, font=\sffamily\large] at ($(L)+(\cut,-0.2)$) {$cM$};
\node[anchor=north, font=\sffamily\large] at ($(R)+(0,-0.2)$) {$M$};

\node[anchor=north west, text width=0.3\textwidth, align=left, %
      font=\sffamily\bfseries\normalsize, text=lowcol]
  at ($(L)+(0,-1.2)$) {Low-energy: subspace with rigorous guarantees};

\node[anchor=north, text width=0.35\textwidth, align=center, %
      font=\sffamily\bfseries\normalsize, text=black]
  at ($(L)!0.5!(R)+(0,-1.2)$) {Medium-energy: \\discretization still works};

\node[anchor=north east, text width=0.3\textwidth, align=right, %
      font=\sffamily\bfseries\normalsize, text=highcol]
  at ($(R)+(-0.5,-1.2)$) {High-energy: \\discretization fails};
  
\end{tikzpicture}}
    \caption{Visualization of discretization error for the QHO. Here $N = \Theta(M/\log M)$, as specified in \cref{thm:discrete factoring}. In the ``low-energy'' subspace, the fast-forwarding algorithm provably works. Meanwhile, in the ``medium-energy'' subspace \cref{fact:somma-discretization-facts} still holds, although we no longer obtain rigorous guarantees on fast-forwarding.}
    \label{fig:ff energy cutoff}
\end{figure}
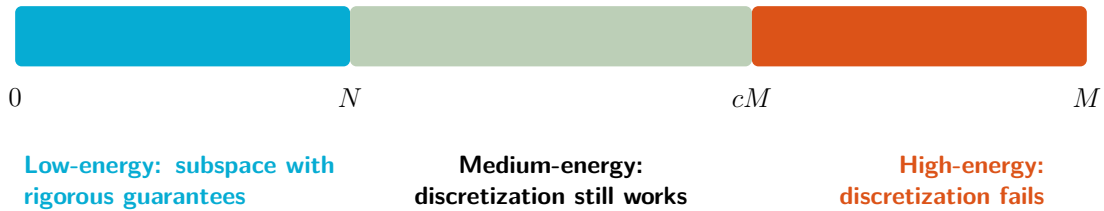

\paragraph{Remark.} You might notice that these statements do not necessarily imply that the $k^{\text{th}}$ eigenstate of the discrete Hamiltonian $\discretehamiltonian$ is close to the $k^{\text{th}}$ discrete Hermite state, in the case that eigenvalue gaps are  exponentially small.
We point out that we do not have a rigorous proof that the eigenvalue gaps are not exponentially small; in principle one could worry that eigenstates from the high energy subspace, upon discretization, could jump down to the low energy subspace and land exponentially close to a low eigen-energy. In fact, numerically one observes that the high eigen-energies, though less accurately conforming to their continuum values than the low eigen-energies, never cross all the way into the low energy subspace.
Nonetheless, since we only care about the closeness to the discrete Hermite states, our claims about fast-forwarding are true irrespective of the eigenvalue gaps.

Note that our convention for basis states is such that $j \in \{-M/2, \ldots, M/2-1\}$, which differs from the standard convention where $j \in \{0,\ldots,M-1\}$. This simplifies the exposition, since Hermite functions are defined in a domain where $x \in \mathbb R$ and we can directly associate $j$ with a position in space.
Transforming from one convention to the other is a cyclic permutation, i.e., $\ket{-M/2}\mapsto \ket 0$,  $\ket{-M/2+1}\mapsto \ket 1$, and so on. 
However, we do not need to perform this permutation in the algorithm; the relevant observation is that $F$ is not exactly the QFT, as explained in~\cref{app:cQFT}, but can still be implemented with $\cO(\log^2 M)$ gates. 

The following result shows that the factorization
of the QHO's evolution operator in the continuum of \cref{thm:factor} can be carried to this discrete case.

\begin{theorem} \label{thm:discrete factoring}
    There is a constant $\eta$ such that for $\lowenergy = \eta M / \log M$, we have the following. Let $\discretehamiltonian$ be the discretization with grid size $\sqrt{\frac{2\pi}{M}}$ and for $|t| \leq \pi/2$,
    \[
    \norm{\Pi_{\lowenergy} \left(e^{-\ri \discretehamiltonian t} - e^{- \ri a(t) \discretemomentum^2} e^{- \ri b(t) \discreteposition^2} e^{-i a(t) \discretemomentum^2} \right)\Pi_{\lowenergy} } \leq \exp(-\lowenergy/2).
    \]
    where
    \[
    a(t) = \frac{\tan(t/2)}{2} \quad \mathrm{and} \quad b(t) = \frac{\sin(t)}{2}.
    \]
\end{theorem}

We next comment on how to apply \cref{thm:discrete factoring} to decompose $e^{-i \hat{H} t}$ for arbitrary $t$. We first recall that, in the continuum, all eigenvalues of $H$ are of the form $n+1/2$ for integer $n$. Hence, the time evolution of the quantum harmonic oscillator is periodic in the sense that $e^{-i H 2 \pi} = -I$. Therefore, we can always subtract an integer multiple of $2 \pi$ from $t$ to bring $t$ into the range $-\pi$ to $\pi$. As noted in \cref{thm:factor}, the formal identity  \cref{eq:factor} holds for arbitrary $t$ but one must be careful since the tangent function has divergences, in particular at argument $t/2 = \pm \pi/2$. In the discretized case it becomes important to keep the coefficients $a(t)$ and $b(t)$ bounded in order to keep the error bounded. To achieve this, whenever $|t| > \pi/2$, we use the factorization
\begin{align}
    e^{-i \hat{H} t} &= \left( e^{-i \hat{H} (t/2)} \right)^2 \\
    &= \left( e^{- \ri a(t/2) \hat{p}^2} e^{- \ri b(t/2) \hat{x}^2} e^{-\ri a(t/2) \hat{p}^2} \right)^2 \\
    &= e^{- \ri a(t/2) \hat{p}^2} e^{- \ri b(t/2) \hat{x}^2} e^{- 2 \ri a(t/2) \hat{p}^2} e^{- \ri b(t/2) \hat{x}^2} e^{- \ri a(t/2) \hat{p}^2}
\end{align}
This factorization then translates to a quantum circuit to approximate $e^{-i \bar{H} t}$ using a total of five steps. Each step implements either the diagonal operator $e^{-i b \discreteposition^2}$ by phase kickback, or implements the operator $e^{-i a \discretemomentum^2}$ using the quantum Fourier transform to change to its eigenbasis, then applying phase kickback, and then transforming back. When $|t| \leq \pi/2$ we would apply \cref{thm:discrete factoring} directly, which requires only three phase-kickback steps.

To prove \cref{thm:discrete factoring} it would suffice show that the total contribution of the nested commutators like $\left(\sum_{t=3}^{\infty}\frac{1}{t!}\nestedcommutator{\discreteposition^2}{\discretemomentum^2}{t}\right)$ vanishes as the Hilbert space dimension goes to $\infty$. This is in fact not true, but in \cref{thm:commutator1} we show this is true when we project down to the subspace corresponding to energy levels $\lowenergy = O(M/\log M)$ for a Hilbert space of dimension $M$. We show a similar bound on the terms involving $\nestedcommutator{\discretemomentum^2}{\discreteposition^2}{t}$ and $\nestedcommutator{\discretemomentum^2}{\anticommutator{\discreteposition,\discretemomentum}}{t}$, which suffices to prove \cref{thm:discrete factoring}. Since we pay a $\log M$ cost to compute up to $M$, we can always run the discretization for an $M$ such that $\lowenergy  = O(M/\log M)$ and project back down. More precisely, we prove the following theorem about the norm of the nested commutators. Below, $\gamma$ is a constant to be declared later.

\begin{restatable}{theorem}{finalbounds}

\label{thm:final bounds}
    For any constants $c_1, c_2$ such that $\max(\abs{c_1}, \abs{c_2}) \leq 1$, and $\lowenergy = \sommaconstant M/(40 \log(2M))$, we have
    \[ \norm{\lowenergyprojector\left( \sum_{t=3}^{\infty}\frac{1}{t!} \nestedcommutator{c_1\discreteposition^2}{c_2\discretemomentum^2}{t}\right)\lowenergyprojector} \leq \expbound{\sommaconstant \lowenergy/2}.
    \]
    \[ \norm{\lowenergyprojector\left( \sum_{t=3}^{\infty}\frac{1}{t!} \nestedcommutator{c_1\discretemomentum^2}{c_2\discreteposition^2}{t}\right)\lowenergyprojector} \leq \expbound{\sommaconstant \lowenergy/2}
    \]
    \[ \norm{\lowenergyprojector\left( \sum_{t=2}^{\infty}\frac{1}{t!} \nestedcommutator{c_1\discretemomentum^2}{c_2\anticommutator{\discreteposition,\discretemomentum}}{t}\right)\lowenergyprojector} \leq \expbound{\sommaconstant \lowenergy/2}.
    \]
\end{restatable}

The main technical difficulty in proving this statement comes from the fact that when we have a nested commutator $\nestedcommutator{A}{B}{t}$ with a very large $t$, then the naive bound on this scales with an exponential in $t$, since we repeatedly use the norm of the matrices as a bound. We use a trick to split the sum into few-nestings and many-nestings to control the growth of these terms when projected to a low-energy subspace. In the next section, we discuss the explicit algorithm for the QHO and prove it is correct
using this theorem. Then in \cref{sec:nested commutator} we prove this theorem.

\subsection{The algorithm}
From the statement of \cref{thm:discrete factoring}, the algorithm for fast-forwarding is simple, as given in \cref{alg:decomposition}.

\begin{algorithm}
\caption{Fast-forwarding QHO}
\label{alg:decomposition}
\begin{algorithmic}[1]
    \State $t_2 \gets t - 2 \pi \lfloor t/2 \rceil$ \Comment{Shift $t$ into the range $[-\pi,\pi)$}
    \If{$|t_2| > \pi/2$} \Comment{If $t_2$ is too big halve it and use two repetitions.}
        \State $\alpha \gets \tan(t_2/4)/2$
        \State $\beta \gets \sin(t_2/2)/2$
        \State Compute $e^{-i \alpha \discretemomentum^2} e^{-i \beta \discreteposition^2} e^{-i 2 \alpha \discretemomentum^2} e^{-i \beta \discreteposition^2} e^{-i \alpha \discretemomentum^2}$
    \Else \Comment{Otherwise use one repetition.}
        \State $\alpha \gets \tan(t_2/2)/2$
        \State $\beta \gets \sin(t_2)/2$
        \State Compute $e^{-i \alpha \discretemomentum^2} e^{-i \beta \discreteposition^2} e^{-i \alpha \discretemomentum^2}$
    \EndIf
\end{algorithmic}
\end{algorithm}

So let us now prove \cref{thm:discrete factoring} using \cref{thm:final bounds}. We will give the proof of \cref{thm:final bounds} in the next subsection.

\begin{proof}[Proof of \cref{thm:discrete factoring}]
    Let $\discretehamiltonian = \frac{1}{2}(\discreteposition^2 + \discretemomentum^2)$ be the discretized quantum harmonic oscillator Hamiltonian. We will denote $U(t) = \exp(-\ri \discretehamiltonian t)$. Further, denote \[\widetilde{U(t)} = \exp(-\ri \frac{\tan(t/2)\discretemomentum^2}{2}) \exp(-\ri \frac{\sin(t)\discreteposition^2}{2}) 
\exp(-\ri \frac{\tan(t/2)\discretemomentum^2}{2})\]
    We show in \cref{sec:qho-factorization}, \cref{lem:discrete factor} that for $\alpha' =  \tan(t/2)/2, \beta' = \sin(t)/2$,
\[
\widetilde{U(t)}^{-1}\frac{d\widetilde{U(t)}}{dt} = -\ri \discretehamiltonian + \dot{\beta'} \sum_{t=3}^{\infty}\frac{1}{t!2^{t/2}} \nestedcommutator{\beta' \ri \discreteposition^2}{\ri \discretemomentum^2}{t}+ \dot{\alpha'} \sum_{t=3}^{\infty}\frac{1}{t!2^{t/2}} \nestedcommutator{\alpha' \ri \discretemomentum^2}{\ri \discreteposition^2}{t} + \dot{\alpha'} \sum_{t=2}^{\infty}\frac{1}{t!2^{t/2}} \nestedcommutator{\alpha' \ri \discretemomentum^2}{\ri \anticommutator{ \discreteposition,\discretemomentum}}{t} 
\]
Now we can use \cref{thm:final bounds}, but we need a bound on $\alpha', \beta'$. Assume without loss of generality that $t \in [-\pi, \pi]$. If $\abs{t} < \pi/2$ then we have that $\abs{\alpha'}, \abs{\beta'} \leq 1/2$. We also have $\dot{\alpha'}, \dot{\beta'} \leq 2$. Hence we can conclude
    \[
\|\Pi_\lowenergy  \left(\widetilde{U(t)}^{-1}\frac{d\widetilde{U(t)}}{dt} 
+ \ri \discretehamiltonian \right) \Pi_\lowenergy  \| \leq 6\exp(- \sommaconstant \lowenergy/2)
\]
Let us represent this quantity in the brackets by $\eta$. Then by the Baker-Campbell-Hausdorff theorem we know that,
\[
\exp({\eta})\widetilde{U(t)} = \exp(\eta + \ri \discretehamiltonian t + \frac{\ri t}{2}\commutator{\eta, \discretehamiltonian} + \frac{\ri t}{12}\nestedcommutator{\eta}{\discretehamiltonian}{2} + \ldots)
\]

Thus, we have that 
\[
\|\Pi_\lowenergy  \left( U(t) - \widetilde{U(t)} \right) \Pi_\lowenergy \| \leq 6\exp(- \sommaconstant \lowenergy/4)
\]

It remains to bound the case when $\abs{t}$ is in $(\pi/2, \pi)$. In this case, we can divide $t$ by 2 and repeat the factorization twice. So define $\widetilde{U(t)}_2 = \left( \exp(\frac{\ri \tan(t/4)\discretemomentum^2}{2}) \exp(\frac{\ri \sin(t/2)\discreteposition^2}{2}) 
\exp(\frac{\ri \tan(t/2)\discretemomentum^2}{2}) \right)^2$. Now we have the same bounds on the nested commutators. Hence we can derive that 

\[
\|\Pi_\lowenergy  \left( U(t) - \widetilde{U(t)}_2 \right) \Pi_\lowenergy \| \leq 12\exp(-\sommaconstant \lowenergy/4)
\]

This proves that for all $t$ the discrete Hamiltonian evolution can be factorized into 3 or 6 terms while being doubly-exponentially accurate. We use that $\sommaconstant \geq 1/2$ to conclude the quantative bound in our theorem.
\end{proof}

Now, the size of the circuit to compute these factorizations in \cref{alg:decomposition} is enough to perform fast-forwarding. Hence we are ready to prove \cref{thm:fastforwarding}.

\begin{proof}[Proof of \cref{thm:fastforwarding}]
    Computing these factorizations has error $\exp(-\lowenergy/10)$ due to \cref{thm:discrete factoring}. Thus \cref{alg:decomposition} is correct. We can compute these factorizations in $\cO(\log^2 \lowenergy)$ time because we can prepare $\exp(-\ri \discreteposition)$ of dimension $M$ in $\cO(\log^2 \lowenergy)$ time using the technique of phase kickback since $\discreteposition$ is diagonal. Further, we can perform QFT in $\widetilde{\cO}(\log \lowenergy)$ time~\cite{CW00} to prepare $\discretemomentum$ from $\discreteposition$, which is diagonal in the Fourier basis so allows us to prepare $\exp(-\ri \discretemomentum)$ in $\cO(\log^2 \lowenergy)$ time. Adding the trignometric functions in the expression costs constant time by our assumption. Hence the total runtime is $\cO(\log^2 \lowenergy)$.
\end{proof}

\subsection{Bounding nested commutators}
\label{sec:nested commutator}
In this section we prove the aforementioned bounds on the nested commutators of $\discreteposition^2$,  $\discretemomentum^2$ and $\anticommutator{\discreteposition,\discretemomentum}$.
Going forward, we define
\begin{equation}
\Delta = \nestedcommutator{\discreteposition^2}{\discretemomentum^2}{2} - 4i{\discreteposition}^2
\end{equation}
to be the \emph{discretization error}. Indeed, in the continuous case, $\nestedcommutator{\position^2}{\momentum^2}{2} - 4\ri{\position}^2 = 0$. Choose 
\begin{equation}
\label{eq:mediumenergy}
\mediumenergy = \frac{\sommaconstant M}{8 \log (2M)}
\end{equation}
for some suitably small constant $\sommaconstant > 0$ to be chosen later, and take
\begin{equation}
\label{eq:lowenergy}
\lowenergy = \mediumenergy/5.
\end{equation}
We define $\lowenergyprojector$ to be the projector onto the the subspace spanned by the discrete eigenstates $\apxdshostate{k}$ for $k \leq \lowenergy$, and $\mediumenergyprojector$ is defined analogously. Formally,
\[
\lowenergyprojector = \sum_{k=1}^{\lowenergy} \apxdshostate{k}\!\braapxdshostate{l}, \quad \mediumenergyprojector = \sum_{k=1}^{\mediumenergy} \apxdshostate{k}\!\braapxdshostate{l}
\]

We argue that, analogously to the continuous case, the discretized $\discreteposition^a$ operators do not cause significant leakage from the bottom $\lowenergy$ eigenstates to eigenstates above level $\mediumenergy$, provided that $a$ is sufficiently small. First we show that the corresponding matrix elements in the discrete and continuous cases are close.

\begin{lemma}\label{lem:x-matrix-element-bound}
    Let $a \leq \frac{\sommaconstant M}{10 \log(2M)}$. Then whenever $M$ is sufficiently large,
    \[
    \max_{\substack{k \leq \mediumenergy \\ \ell \leq \lowenergy}} \abs{\braapxdshostate{k} \discreteposition^a \apxdshostate{\ell} - \bratshostate{k} \position^a \tshostate{\ell}} \leq \exp(-\sommaconstant M/4)
    \]
\end{lemma}
\begin{proof}
    We will first bound the distance between $\bradshostate{k}\discreteposition^a \dshostate{\ell}$ and $\bratshostate{k}\position^a \tshostate{\ell}$ then later show that the former is close to $\braapxdshostate{k}\discreteposition^a \apxdshostate{\ell}$.
    We can write $\bratshostate{k} \position^a \tshostate{\ell} = \int_\R x^a\psi_k(x)\psi_\ell(x) dx$. Recall that our discretization has the limits $[-L,L]$, where $L = \sqrt{\frac{\pi M}{2}}$. Define the strip $\calS = \{z \in \mathbb{C} : \max\{\abs{\Re(z)}, \abs{\Im(z)}\} \leq L\}$.
    Within this region, applying Plancherel-Rotach asymptotic arguments (see, e.g. \cite{szeg1939orthogonal}), we have the bound 
    \begin{equation}\label{eq:hermite-strip-bound}
    \sup_{z \in \calS} \abs{\psi_k(z)} \leq K \cdot k^{-1/4} \sup_{z \in \calS} \abs{e^{-z^2}/2}  \leq K \cdot \sup_{z \in \calS} \abs{\ e^{\Im(z)^2}/2}  \leq K \cdot e^{L^2/2}.
    \end{equation}
    This readily gives us the bound
    \[
    \sup_{z \in \calS} \abs{x^a\psi_k(z)\psi_k(z)} \leq (L \sqrt{2})^{a} \cdot e^{L^2}
    \]
    for some suitably large constant $K$. Thus, employing the exponentially convergent trapezoid rule \cite{trefethen2014exponentially}, we have 
    \begin{equation}\label{eq:hermite-trapezoid-convergence}
    \abs{\bradshostate{k}\discreteposition^a \dshostate{\ell} - \int_{-L}^L x^a\psi_k(x)\psi_\ell(x) dx} \leq 2 (L\sqrt{2})^{a} \exp(L^2 - 2\pi M) = 2(L \sqrt{2})^{a}\exp(-\pi M).
    \end{equation}
    Since $a \leq M / 10 \log (2M)$ we can bound this term by $\exp(-M)$. Now we compute the error accrued from imposing the limits $\pm L$. That is,
    \[
    \abs{\int_{-L}^L x^a\psi_k(x)\psi_\ell(x) dx - \int_{\R} x^a\psi_k(x)\psi_\ell(x) dx} = 2\abs{\int_{L}^\infty  x^a\psi_k(x)\psi_\ell(x) dx}.
    \]
    We can bound $\abs{x^a \psi_k(x) \psi_{\ell}(x)} \le \pi^{-1/4}  
    \abs{x^a \psi_k(x)}$.
    Now, $L$ is much larger than the turning point of $\psi_k(x)$. By standard results (such as, say \cite[Eq. (A4)]{somma2016}, there exists a constant $c$ such that
    \[
    \abs{x^a \psi_k(x)} \leq c \abs{x^a} \cdot e^{\sqrt{2k}x - x^2/2} = c e^{\sqrt{2k}x - a \log x - x^2/2} \leq e^{-x^2/3}
    \]
    whenever $M$ is sufficiently large.
    This gives us the bound
    \[
    \abs{\int_{L}^\infty  x^a\psi_k(x)\psi_\ell(x) dx} \leq \abs{\int_{L}^\infty e^{-x^2/3} dx} \leq e^{-L^2/6} \leq e^{-\pi M/12}.
    \] 
    We now show that $\braapxdshostate{k}\discreteposition^a \apxdshostate{\ell}$ is close to $\bradshostate{k}\discreteposition^a \dshostate{\ell}$. Indeed, we have
    \begin{align}
    \abs{\braapxdshostate{k}\discreteposition^a \apxdshostate{\ell}- \bradshostate{k}\discreteposition^a \dshostate{\ell}} &= \abs{\braapxdshostate{k}\discreteposition^a \apxdshostate{\ell}- \braapxdshostate{k}\discreteposition^a \dshostate{\ell}} + \abs{\braapxdshostate{k}\discreteposition^a \dshostate{\ell}- \bradshostate{k}\discreteposition^a \dshostate{\ell}} \\
    &= \abs{\braapxdshostate{k}\discreteposition^a \left(\apxdshostate{\ell} - \dshostate{\ell}\right)} + \abs{\left(\braapxdshostate{k} - \bradshostate{k}\right)\discreteposition^a \dshostate{\ell}} \\
    &\leq \norm{\discreteposition^a \apxdshostate{k}}\cdot \norm{\apxdshostate{\ell} - \dshostate{\ell}} + \norm{\apxdshostate{k} - \dshostate{k}} \cdot \norm{\discreteposition^a \dshostate{\ell}} & \byref{Cauchy-Schwarz} \\
    &\leq \norm{\discreteposition^a}\cdot \norm{\apxdshostate{\ell} - \dshostate{\ell}} + (1 + \exp(-\sommaconstant M) \norm{\apxdshostate{k} - \dshostate{k}} \cdot \norm{\discreteposition^a} &\byref{\cref{fact:somma-discretization-facts}}\\
    &\leq \exp\left(- \frac{a\log (\pi M/2)}{2}\right)\cdot \left(\norm{\apxdshostate{\ell} - \dshostate{\ell}} + 2\norm{\apxdshostate{k} - \dshostate{k}} \right) \\
    &\leq 3 \exp(\sommaconstant M/10 -\sommaconstant M) &\byref{\cref{fact:somma-discretization-facts}} \\
    &= 3 \exp(- 9\sommaconstant M/10).
    \end{align}
    Putting everything together, the overall bound we obtain is 
    \[
    (2M)^{a/2}\exp(-\pi M / 2) + 2 e^{-\pi M/12} + 3 e^{- 9\sommaconstant M/10} \leq e^{-\sommaconstant M/4}
    \]
    for all $M$ sufficiently large.
\end{proof}

We use the above bound to show that low degree position operators result in very little leakage from energy at most $\lowenergy$ to energy above $\mediumenergy$.

\begin{lemma}\label{lem:position-leakage}
    For all $M$ sufficiently large and whenever $a \leq 4\lowenergy$,
    $\norm{(I - \mediumenergyprojector) \discreteposition^{a} \lowenergyprojector} \leq \exp(-\sommaconstant M/9)$, where $N'$ and $N$ are as defined in \cref{eq:mediumenergy} and \cref{eq:lowenergy}.
\end{lemma}
\begin{proof}
    Define 
    \[
    \delta_a = \max_{\substack{k \leq \mediumenergy \\ \ell \leq \lowenergy}} \abs{\braapxdshostate{k} \discreteposition^a \apxdshostate{\ell} - \bratshostate{k} \position^a \tshostate{\ell}}
    ,\quad
    \kappa_a = \max_{k \leq \lowenergy} \abs{\braapxdshostate{k} \discreteposition^{2a} \apxdshostate{k} - \bratshostate{k} \position^{2a} \tshostate{k}}
    \]
    We claim that
    \[
    \norm{(I - \mediumenergyprojector) \discreteposition^{a} \lowenergyprojector} \leq \left(\sqrt{\mediumenergy} \cdot \delta_{2\lowenergy} + \sqrt{\kappa_{2\lowenergy}}\right) \cdot \sqrt{M}
    \]
    Indeed, we need only show that
    \[
    \max_{\substack{\mediumenergy < k \leq M \\ \ell \leq \lowenergy}} \abs{\braapxdshostate{k} \discreteposition^a \apxdshostate{\ell}} = \max_{\substack{\mediumenergy < k \leq M \\ \ell \leq \lowenergy}} \abs{\braapxdshostate{k} \discreteposition^a \apxdshostate{\ell} - \bratshostate{k} \position^a \tshostate{\ell}} \leq \sqrt{\mediumenergy} \cdot \delta_{4\lowenergy} + \sqrt{\kappa_{4\lowenergy}}.
    \]
    We can write
    \[
    \discreteposition^a \apxdshostate{\ell} = \beta \ket{\overline{\psi}} + \sum_{j \leq \mediumenergy} \alpha_j^\prime \apxdshostate{j},\quad \position^a \tshostate{\ell} = \sum_{j \leq \mediumenergy} \alpha_j \tshostate{j}
    \]
    where we have used the fact that $a \leq 4 \lowenergy$, and doesn't push the support above energy level $\mediumenergy$. We have the bound $\abs{\braapxdshostate{k} \discreteposition^a \apxdshostate{\ell}} \leq \abs{\beta}$. We have $\abs{\alpha_j - \alpha_j^\prime} \leq \delta_a$ and
    \[
    \abs{\beta^2 + \sum_{j=1}^{\mediumenergy} (\alpha_j - \alpha_j^\prime)^2} \leq \kappa_{4\lowenergy}
    \]
    From this we have the desired bound
    \[
    \abs{\beta} \leq \sqrt{\mediumenergy} \cdot \delta_{4\lowenergy} + \sqrt{\kappa_{4\lowenergy}}
    \]

    By \cref{lem:x-matrix-element-bound}, since $a \leq 4\lowenergy \leq \sommaconstant M/(10 \log (2M))$, we can bound 
    \[
    \delta_{4\lowenergy} \leq \exp(-\sommaconstant M/4),\quad \kappa_{4\lowenergy} \leq \exp(-\sommaconstant M/4)
    \]
    Thus we have the bound
    
    \begin{align}
    \abs{\beta} &\leq \sqrt{\mediumenergy} \cdot e^{-\sommaconstant M/4} +  e^{-\sommaconstant M/8} \leq 2e^{-\sommaconstant M/8}.
    \end{align}
    Since $\norm{(I - \mediumenergyprojector)\discreteposition^a \lowenergyprojector} \leq \abs{\beta} \cdot \mediumenergy$, we have the bound
    \[
    \norm{(I - \mediumenergyprojector)\discreteposition^a \lowenergyprojector} \leq 2 \sqrt{M} e^{-\sommaconstant M/8} \leq e^{-\sommaconstant M/9}
    \] 
    for all $M$ sufficiently large.
\end{proof}

We now prove a number of desirable properties of the discretization error $\Delta$. In particular, we show that $\Delta$ is small on the subspace spanned by the first $\mediumenergy$ eigenvectors, and then bound its overall operator norm. First, we recall a fundamental fact about Hermite polynomials.
\begin{fact} \label{fact:hermite-position-recurrence}
\[
x H_k(x) = k H_{k-1}(x) + \frac{1}{2} H_{k+1}(x)
\]
\end{fact}

\begin{lemma}\label{lem:delta-bounds}
For all $M$ sufficiently large:
    \begin{enumerate}
        \item $\norm{\mediumenergyprojector \Delta \mediumenergyprojector} \leq \exp(- \sommaconstant M/3)$
        \item $\norm{\Delta} \leq 17 M^3$.
    \end{enumerate}
\end{lemma}

\begin{proof}
    We have the bound
    \[
    \norm{\mediumenergyprojector \Delta \mediumenergyprojector} \leq \sqrt{M} \cdot \max_{\substack{k \leq \mediumenergy \\ \ell \leq \mediumenergy}} \abs{\braapxdshostate{k} \Delta \apxdshostate{\ell}} = \sqrt{M} \cdot \max_{\substack{k \leq \mediumenergy \\ \ell \leq \mediumenergy}} \abs{\braapxdshostate{k} \Delta \apxdshostate{\ell} - \bratshostate{k} \nestedcommutator{\position^2}{\momentum^2}{2} - 8i\discreteposition^2 \tshostate{\ell}},
    \]
    where we recall that $\nestedcommutator{\position^2}{\momentum^2}{2} = 8i\discreteposition^2$. We expand
    \[
    \Delta = \discreteposition^4 \discretemomentum^2 - 2 \discreteposition^2 \discretemomentum^2 \discreteposition^2 + \discretemomentum^2 \discreteposition^4 - 8\ri \discreteposition^2.
    \]
    We can use the triangle inequality to match terms of $\Delta$ and $\nestedcommutator{\position^2}{\momentum^2}{2} - 8i\discreteposition^2$. In what follows, we make ample use of \cref{fact:somma-discretization-facts}. For the first such term, we bound
    \[
    \abs{\braapxdshostate{k} \discreteposition^4 \discretemomentum^2 \apxdshostate{\ell} - \bratshostate{k} \position^4 \momentum^2 \tshostate{\ell}} \leq \exp(-\sommaconstant M).
    \]
    The same argument shows that 
    \[
    \abs{\braapxdshostate{k}\discretemomentum^2 \discreteposition^4  \apxdshostate{\ell} - \bratshostate{k} \momentum^2 \position^4 \tshostate{\ell}} \leq \exp(-\sommaconstant M).
    \]
    Furthermore, we have
    \[
    \abs{\braapxdshostate{k} 8\ri \discreteposition^2  \apxdshostate{\ell} - \bratshostate{k} 8\ri \position^2 \tshostate{\ell}} \leq 8\exp(-\sommaconstant M).
    \]
    The most challenging term to bound is
    \[
    \abs{\braapxdshostate{k}\discreteposition^2 \discretemomentum^2 \discreteposition^2  \apxdshostate{\ell} - \bratshostate{k} \position^2 \momentum^2 \position^2 \tshostate{\ell}}.
    \]
    Our solution will be to write $\discreteposition^2 \apxdshostate{k}$ as a linear combination of other states. Indeed, in the continuum, we can apply \cref{fact:hermite-position-recurrence} twice to obtain $\position^2 \tshostate{k} = k(k-1) \tshostate{k-2} + k \tshostate{k} + \frac{1}{4}\tshostate{k+2}$. We write
    \begin{align}
    \label{eqn:x2 expansion}\discreteposition^2\apxdshostate{k} = \alpha_1 \apxdshostate{k-2} + \alpha_2\apxdshostate{k} + \alpha_3 \apxdshostate{k+2} + \beta\ket{\psi}.
    \end{align}
    By \cref{fact:somma-discretization-facts}, we have $\abs{\alpha_1 - k(k-1)} \leq \exp(-\sommaconstant M)$, $\abs{\alpha_2 - k} \leq \exp(-\sommaconstant M)$, and $\abs{\alpha_3 - 1/4} \leq \exp(-\sommaconstant M)$. Similarly, we have
    \[
    \abs{\braapxdshostate{k}\discreteposition^4 \apxdshostate{k} - \bratshostate{k}\position^4 \tshostate{k}} \leq \exp(-\sommaconstant M).
    \]
    This implies that
    \[
    (\alpha_1^2 - k^2(k-1)^2) + (\alpha_2^2 - k^2) + \left(\alpha_1^2 - \frac{1}{16}\right) + \beta^2 \leq \exp(-\sommaconstant M).
    \]
    We can factor the left hand side as
    \[
    (\alpha_1 + k(k-1))(\alpha_1 - k(k-1)) + (\alpha_2 + k)(\alpha_2 - k) + (\alpha_3 + 1/4)(\alpha_3 - 1/4) + \beta^2
    \]
    This factorization gives us the bound
    \begin{align*}
    \abs{\beta}^2 &\leq \exp(-\sommaconstant M) + (\alpha_1 + k(k-1))\abs{\alpha_1 - k(k-1)} + (\alpha_2 + k)\abs{\alpha_2 - k} + (\alpha_3 + 1/4)\abs{\alpha_3 - 1/4} \\
    &\leq \exp(-\sommaconstant M) + (2k(k-1) + \exp(-\sommaconstant M))\exp(-\sommaconstant M) \\
    &+ (2k + \exp(-\sommaconstant M))\exp(-\sommaconstant M) + (1/2 + \exp(-\sommaconstant M))\exp(-\sommaconstant M) \\
    &\leq 3\left(k^2 + \frac{1}{16}\right)\exp(-\sommaconstant M) + \exp(-\sommaconstant M) \leq 4k^2 \exp(-\sommaconstant M),
    \end{align*}
    so $\abs{\beta} \leq 2k\exp(-\sommaconstant M/2)$. 

    To bound the term $\abs{\bradshostate{k}\discreteposition^2 \discretemomentum^2 \discreteposition^2  \dshostate{\ell} - \bratshostate{k} \position^2 \momentum^2 \position^2 \tshostate{\ell}}$ we first use the canonical comutator relation to expand $\discreteposition^2 \discretemomentum^2 \discreteposition^2 = \discreteposition^4 \discretemomentum^2 - 4\ri \discreteposition^3 \discretemomentum - 2\discreteposition^2 + \discreteposition^2 \Delta''$ where $\Delta'' = \discretemomentum^2 \discreteposition^2 - (\discreteposition^2 \discretemomentum^2 - 4\ri \discreteposition \discretemomentum - 2I)$. After applying triangle inequality by splitting on the discrete minus the continuum operators, the same argument as the other terms implies an $O(\exp(-\sommaconstant M))$ bound on the other 3 terms. The only non-trivial term here is $\discreteposition^2 \Delta''$.

    To bound this term, we expand $\abs{\braapxdshostate{k}\discreteposition^2 \discretemomentum^2 \discreteposition^2  \apxdshostate{\ell}}$ using \cref{eqn:x2 expansion}. We get

    \[
    \abs{\left( \alpha_1 \braapxdshostate{k-2} + \alpha_2\braapxdshostate{k} + \alpha_3 \braapxdshostate{k+2} + \beta\bra{\psi} \right)\Delta'' \apxdshostate{l}} \leq \exp(-\sommaconstant M) + \abs{\beta} \|\Delta''\|
    \]

    But we know $\abs{\beta} \leq 2k \exp(-\sommaconstant M/2)$ and $\|\Delta''\| \leq 7M^2$. Overall, we have that
    \[
    \abs{\braapxdshostate{k}\discreteposition^2 \discretemomentum^2 \discreteposition^2  \apxdshostate{\ell} - \bratshostate{k} \position^2 \momentum^2 \position^2 \tshostate{\ell}} \leq 40k M^2 \exp(-\sommaconstant M/2).
    \]

    Putting everything together, $\norm{\mediumenergyprojector \Delta \mediumenergyprojector} \leq \exp(-\sommaconstant M/3)$ for all $M$ sufficiently large.

    To see the second point, we simply write $\Delta = \discreteposition^4\discretemomentum^2 - 2 \discreteposition^2 \discretemomentum^2 \discreteposition^2 + \discretemomentum^2 \discreteposition^4 - 8 \ri \discreteposition^2$. Now, $\norm{\discreteposition} = \norm{\discretemomentum} = \sqrt{\frac{\pi M}{2}}$. Thus we have
    \begin{align}
        \norm{\Delta} &\leq \norm{\discreteposition^4\discretemomentum^2} + 2 \norm{\discreteposition^2 \discretemomentum^2 \discreteposition^2} + \norm{\discretemomentum^2 \discreteposition^4} + 8 \norm{\discreteposition^2} \\
        &= 4 \cdot \left(\frac{\pi M}{2}\right)^3 + 8 \left(\frac{\pi M}{2}\right) \\
        &\leq 16 M^3 + 4\pi M \leq 17M^3. \qedhere
    \end{align}
\end{proof}

Our strategy to bound each level of the nested commutator tail will be to break the $t$-th order commutator into a noncommutative polynomial of $\discreteposition^2$ and $\Delta$. First, we provide a bound on each monomial appearing in this expansion.

\begin{lemma}\label{lem:single-monomial-bound}
    Let $a + b = t$. For all $M$ sufficiently large,
    \[
    \norm{\lowenergyprojector \discreteposition^{2a} \Delta \discreteposition^{2b} \lowenergyprojector} \leq \begin{cases}
        \left(51 M^3 + 1\right) \cdot (2\lowenergy)^t \cdot \exp(-\sommaconstant M/9) & t \leq (\mediumenergy - \lowenergy)/2 \\
        \left(51 M^3 + 1\right) \cdot (2\lowenergy)^t & t > (\mediumenergy - \lowenergy)/2
    \end{cases}
    \]
\end{lemma}

\begin{proof}
    Our strategy will be to introduce the projector $\mediumenergyprojector$ onto the \say{medium energy} subspace spanned by the first $\mediumenergy$ eigenstates of $\discretehamiltonian$. 
    We write
    \[
    \norm{\lowenergyprojector \discreteposition^{2a} \Delta \discreteposition^{2b} \lowenergyprojector} = \norm{\lowenergyprojector \discreteposition^{2a} (\mediumenergyprojector + (I - \mediumenergyprojector)) \Delta (\mediumenergyprojector + (I - \mediumenergyprojector)) \discreteposition^{2b}  \lowenergyprojector}
    \]
    We can break this up into four terms using the triangle inequality:
    \begin{align}
        T_1 &= \norm{\lowenergyprojector \discreteposition^{2a} \mediumenergyprojector} \cdot \norm{\mediumenergyprojector \Delta \mediumenergyprojector} \cdot \norm{\mediumenergyprojector \discreteposition^{2b} \lowenergyprojector} \\
        T_2 &= \norm{\lowenergyprojector \discreteposition^{2a} (I-\mediumenergyprojector)} \cdot \norm{(I -\mediumenergyprojector) \Delta \mediumenergyprojector} \cdot \norm{\mediumenergyprojector \discreteposition^{2b} \lowenergyprojector} \\
        T_3 &= \norm{\lowenergyprojector \discreteposition^{2a} \mediumenergyprojector} \cdot \norm{\mediumenergyprojector \Delta (I-\mediumenergyprojector)} \cdot \norm{(I-\mediumenergyprojector) \discreteposition^{2b} \lowenergyprojector} \\
        T_4 &= \norm{\lowenergyprojector \discreteposition^{2a} (I-\mediumenergyprojector)} \cdot \norm{(I-\mediumenergyprojector) \Delta (I-\mediumenergyprojector)} \cdot \norm{(I-\mediumenergyprojector) \discreteposition^{2b} \lowenergyprojector}
    \end{align}
    The overall norm is upper bounded by $T_1 + T_2 + T_3 + T_4$. We first consider $T_1$. Applying \cref{lem:delta-bounds}, we have
    \[
    T_1 \leq (2\lowenergy)^{a} \cdot \exp(-\sommaconstant M/3) \cdot (2 \lowenergy)^{b} \leq (2\lowenergy)^{t} \cdot \exp(-\sommaconstant M / 3).
    \]
    For $T_2,T_3,T_4$, we will use the trivial bound on the norm $\Delta$ -- that is, we won't use the projector. Applying \cref{lem:position-leakage} and \cref{lem:delta-bounds}, we have
    \[
    T_2 \leq \exp(-\sommaconstant M/9) \cdot 17 M^3 \cdot (2\lowenergy)^{b} \leq 17 M^3 \cdot (2 \lowenergy)^t \cdot \exp(-\sommaconstant M/9)
    \]
    whenever $t \leq (\mediumenergy - \lowenergy)/2$. When $t > (\mediumenergy - \lowenergy)/2$ we recover the bound 
    \[
    T_2 \leq (2\lowenergy)^{a} \cdot 17M^3 \cdot (2\lowenergy)^{b} \leq 17M^3 \cdot (2 \lowenergy)^t
    \]
    By symmetry, we recover the same bounds on $T_3$. Finally, we use similar techniques to obtain the bounds 
    \[
    T_4 \leq \exp(-\sommaconstant M/9) \cdot 17M^3 \cdot \exp(-\sommaconstant M/9) \leq 17M^3 \cdot \exp(-2\sommaconstant M)
    \]
    when $t \leq (\mediumenergy - \lowenergy)/2$ and 
    \[
    T_2 \leq (2\lowenergy)^{a} \cdot 17M^3 \cdot (2\lowenergy)^{b} \leq 17M^3 \cdot (2 \lowenergy)^t
    \]
    when $t > (\mediumenergy - \lowenergy)/2$.
\end{proof}

Next we bound the $t$-th order commutator by applying the triangle inequality along with the monomial bounds above.

\begin{lemma} \label{lem:single-commutator-bound}
For all $M$ sufficiently large,
    \[
    \norm{\lowenergyprojector \nestedcommutator{\discreteposition^2}{\Delta}{t} \lowenergyprojector} \leq \begin{cases}
         \left(51 M^3 + 1\right) \cdot (4\lowenergy)^t \cdot \exp(-\sommaconstant M/9) & t \leq (\mediumenergy - \lowenergy)/2 \\
        \left(51 M^3 + 1\right) \cdot (4\lowenergy)^t & t > (\mediumenergy - \lowenergy)/2
    \end{cases}
    \]
\end{lemma}

\begin{proof}
    We can expand
    \[
    \nestedcommutator{\discreteposition^2}{\Delta}{t} = \sum_{s = 0}^{t} \binom{t}{s} (-1)^s \discreteposition^{2s} \Delta \discreteposition^{2(t-s)}.
    \]
    Applying the triangle inequality, we have
    \[
    \norm{\lowenergyprojector \nestedcommutator{\discreteposition^2}{\Delta}{t} \lowenergyprojector} \leq \sum_{s = 0}^{t} \binom{t}{s} \norm{\lowenergyprojector \discreteposition^{2s} \Delta \discreteposition^{2(t-s)}\lowenergyprojector}.
    \]
    We use the identity $\sum_{s=0}^t \binom{t}{s} = 2^t$ and \cref{lem:single-monomial-bound} to obtain the result.
\end{proof}

Before proving our bounds on the nested commutator tail of $\discreteposition^2$ and $\discretemomentum$, we recall a useful fact about the tail of the power series of the exponential function.
\begin{proposition} \label{prop:exponential-sum-tail}
    \[
    \sum_{k=3a}^{\infty} \frac{a^k}{k!} \leq \exp(-a/4). 
    \]
\end{proposition}

\begin{proof}
    The sum is exactly $e^{a}$ times the tail of the Poisson distribution with parameter $a$. Thus, from standard formulae for Poisson tail bounds, we have
    \[
    \sum_{k=3a}^{\infty} \frac{a^k}{k!} \leq e^a \Pr[\poisson{a} \geq 3a] \leq e^{a} \exp(-3a \ln 3 + 2a) \leq \exp(-3 (\ln 3 - 1) a) \leq \exp(-a/4). \qedhere
    \]
\end{proof}

We are now ready to prove the key technical ingredient of our fast-forwarding algorithm: that the higher-order nested commutators of $\discreteposition$ and $\discretemomentum$ fall off sufficiently fast.

\begin{theorem}\label{thm:commutator1}
    Choose $\lowenergy = \sommaconstant M/(40 \log (2M))$. Then
    \[ \norm{\lowenergyprojector\left( \sum_{t=3}^{\infty}\frac{1}{t!} \nestedcommutator{\discreteposition^2}{\discretemomentum^2}{t}\right)\lowenergyprojector} \leq \expbound{\sommaconstant \lowenergy/2}.
    \]
\end{theorem}

\begin{proof}
    We observe that $\nestedcommutator{\discreteposition^2}{\discretemomentum^2}{2} = \Delta + 8i\discreteposition^2$. Thus, for $t \geq 3$, we can write
    \[
    \nestedcommutator{\discreteposition^2}{\discretemomentum^2}{t} = \nestedcommutator{\discreteposition^2}{\Delta}{t-2}.
    \]
    Thus,
    \begin{align}
        \norm{\lowenergyprojector \sum_{t=3}^{\infty} \frac{\nestedcommutator{\discreteposition^2}{\discretemomentum^2}{t}}{t!} \lowenergyprojector} &= \norm{\lowenergyprojector \sum_{t=3}^{\infty} \frac{\nestedcommutator{\discreteposition^2}{\discretemomentum^2}{t-2}}{t!} \lowenergyprojector} \\
        &\leq \sum_{t=1}^{\infty}  \frac{ \norm{\lowenergyprojector\nestedcommutator{\discreteposition^2}{\discretemomentum^2}{t}\lowenergyprojector}}{(t+2)!} \\
        &\leq (51 M^3 + 1)\left(e^{-\sommaconstant M/9} \cdot \sum_{t=1}^{(\mediumenergy - \lowenergy)/2} \frac{(4\lowenergy)^t}{t!} + \sum_{t = (\mediumenergy - \lowenergy)/2 + 1}^{\infty} (4\lowenergy)^t \right) & \byref{\cref{lem:single-commutator-bound}} \\
        &\leq (51 M^3 + 1)\left(e^{-\sommaconstant M/9} \cdot \sum_{t=0}^{\infty} \frac{(4\lowenergy)^t}{t!} + \sum_{t = 12\lowenergy}^{\infty} (4\lowenergy)^t \right) \\
        &\leq (51 M^3 + 1)\left(e^{4 \lowenergy -\sommaconstant M/9} + e^{-\lowenergy}\right) & \byref{\cref{prop:exponential-sum-tail}} \\
        &\leq 2(51 M^3 + 1)e^{-\sommaconstant \lowenergy } \leq e^{-\sommaconstant \lowenergy /2}. & & \qedhere
    \end{align}
\end{proof}

As a consequence of our proof technique, the analogous statement holds if we swap position and momentum in the nested commutator.

\begin{corollary}
\label{cor:commutator2}
    Choose $\lowenergy = \sommaconstant M/(40 \log (2M))$. Then
    \[ \norm{\lowenergyprojector\left( \sum_{t=3}^{\infty}\frac{1}{t!} \nestedcommutator{\discretemomentum^2}{\discreteposition^2}{t}\right)\lowenergyprojector} \leq \expbound{\sommaconstant \lowenergy /2}.
    \]
\end{corollary}

\begin{proof}
    The bounds in the proofs of \cref{lem:position-leakage,lem:delta-bounds,lem:single-commutator-bound} remain unchanged when swapping position and momentum. Thus, following the proof of \cref{thm:commutator1}, we obtain the exact same bound on the projected higher-order nested commutators of $\discretemomentum^2$ and $\discreteposition^2$.
\end{proof}

\begin{theorem}
\label{thm:commutator3}
    Choose $\lowenergy = \sommaconstant M/(40 \log(2M))$. Then
    \[ \norm{\lowenergyprojector\left( \sum_{t=2}^{\infty}\frac{1}{t!} \nestedcommutator{\discretemomentum^2}{\anticommutator{\discreteposition,\discretemomentum}}{t}\right)\lowenergyprojector} \leq \expbound{\sommaconstant \lowenergy /2}.
    \]
\end{theorem}

Just as in the proof of \cref{cor:commutator2}, the proofs of \cref{lem:position-leakage,lem:single-monomial-bound,lem:single-commutator-bound} hold with the same scaling if $\discretemomentum^2$ is nested instead of $\discreteposition^2$. What remains is to bound the new defect $\Delta^\prime = \commutator{\discretemomentum^2, \anticommutator{\discreteposition, \discretemomentum}} - 4i \discretemomentum^2 = \discretemomentum^2\discreteposition\discretemomentum + \discretemomentum^3\discreteposition - 4i \discretemomentum^2$. In the continuum, the corresponding operator is $0$, and we now prove an analogous statement to \cref{lem:delta-bounds} for this defect.

\begin{lemma}\label{lem:delta-prime-bounds}
Let $M$ be sufficiently large.
    \begin{enumerate}
        \item $\norm{\lowenergyprojector \Delta^\prime \lowenergyprojector} \leq \exp(-\sommaconstant M/2)$.
        \item $\norm{\Delta^\prime} \leq 5 M^2$
    \end{enumerate}
\end{lemma}

\begin{proof}
    The proof will be very similar to that of \cref{lem:delta-bounds}.
    First we show the the second point:
    \[
    \norm{\Delta^\prime} \leq \norm{\discretemomentum^2\discreteposition\discretemomentum} + \norm{\discretemomentum^3 \discreteposition} + 4\norm{\discretemomentum^2} \leq 2\left(\frac{\pi M}{2}\right)^2 + 4\left(\frac{\pi M}{2}\right) \leq 5 M^2
    \]
    whenever $M$ is sufficiently large.
\end{proof}

To bound this quantity we will need the following bound.
\begin{lemma}
    $\abs{\braapxdshostate{k} \discretemomentum^2 \discreteposition \discretemomentum \apxdshostate{l} - \bratshostate{k} \momentum^2 \position \momentum \tshostate{l}} \leq 10 \exp(-\sommaconstant M)$
\end{lemma}
\begin{proof}
    
    We use the canonical commutation relation to derive $\discretemomentum^2 \discreteposition \discretemomentum = \discretemomentum^3 \discreteposition - \ri \discretemomentum^2 + \discretemomentum^2 \Delta'$ where $\Delta' = \discretemomentum \discreteposition - \ri I - \discreteposition \discretemomentum$. Now, using the \cref{fact:somma-discretization-facts} we can bound the first two terms after applying triangle inequality. Thus, the only remaining task is to bound $\abs{\bradshostate{k} \discretemomentum^2 \Delta' \dshostate{l}}$ is small.

    To do this we use the expansion,
 \begin{align}
    \discretemomentum^2\apxdshostate{k} = \alpha_1 \apxdshostate{k-2} + \alpha_2\apxdshostate{k} + \alpha_3 \apxdshostate{k+2} + \beta\ket{\psi}
    \end{align}

where $\abs{\beta} \leq 2 \exp(-\sommaconstant M)$, as in the proof of \cref{lem:delta-bounds}. Using this we can rewrite the above as

\begin{align}
    \abs{\left(\alpha_1 \braapxdshostate{k-2} + \alpha_2\braapxdshostate{k} + \alpha_3 \braapxdshostate{k+2} + \beta\bra{\psi} \right) \Delta' \apxdshostate{l}} \leq 5\exp(-\sommaconstant M) + 7M^2 \exp(-\sommaconstant M)
\end{align}

where the second term is just bounded using $\abs{\beta} \norm{\Delta'}$.
\end{proof}

Now we are ready to prove \cref{thm:final bounds}. We recall it below.

\finalbounds*

\begin{proof}[Proof of \cref{thm:final bounds}]
    Notice that the bound in \cref{lem:single-commutator-bound} holds when replacing the operator $\discreteposition$ (or $\discretemomentum$) with $c \discreteposition$ where $|c| \leq 1$. Nothing else changes in the proofs.
\end{proof}

\section{Quantum Hermite transform}

Building on our fast-forwarding result, we can implement an efficient QHT if we can implement the state preparation, filtering, and QPE algorithms efficiently. We now discuss these steps and then present the algorithm. The main result of this section is the formal version of \cref{thm:informal qht}.

\begin{theorem}[Quantum Hermite transform, formal]
\label{thm:quantumhermitetransform}
  Let $\lowenergy$ be the dimension for the quantum Hermite transform and $\eps>0$ be the error. 
    Then, there exists a quantum circuit that performs the transformation 
    \begin{align}
    \sum_{n=0}^{\lowenergy-1} \alpha_n \ket n \mapsto 
    \sum_{n=0}^{\lowenergy-1} \alpha_n \dshermite{n}
\end{align}
within additive error $\eps$, where the coefficients $\alpha_n$ are arbitrary and satisfy $\sum_{n=0}^{\lowenergy-1}| \alpha_n|^2 =1$.
The quantum circuit acts on a Hilbert space of dimension $M= \poly(N, 1/\eps)$ and,
if $N > \log(1/\eps)$, the complexity of the quantum circuit is
\begin{align}
    \cO \left( (\log \lowenergy + \log(1/\eps))^3 \times \log(1/\eps)\right) \;.
 \end{align}
\end{theorem}

This result is based on the state preparation steps outlined in~\cref{sec:overview}.
More generally, 
 we can choose any dimension $M \geq c\lowenergy^{9/4}/\eps^{13/4}$,
where $c>0$ is some constant. The cost of the QHT
is then $\cO(\log^3 M \times \log (1/\eps))$, and choosing 
$M= \poly(N, 1/\eps)$ gives~\cref{thm:quantumhermitetransform}.

\subsection{State preparation}\label{sec:stateprep}

We first show how to efficiently prepare a set of quantum states that have constant overlap with the Hermite states $\dshermite{n}$ of~\cref{eq:dshermite} in a low-energy subspace of interest. We will then used fixed-point amplitude amplification to increase the overlap with $\dshermite{n}$ arbitrarily. Our strategy is based on the  Plancherel-Rotach approximation, which approximates the Hermite functions
$\psi_n(x)$ in the `oscillatory' region, 
specified by the domain $|x|< x_{\rm tp}:=\sqrt{2n+1}$.

\begin{lemma}[Plancherel-Rotach asymptotics for Hermite functions, Thm. 8.22.9~\cite{szeg1939orthogonal}]
\label{lem:Plancherel}
    Let $\varphi(x) = \mathrm{arccos}(x/\sqrt{2n+1})$. Let $c>0$ be any fixed positive constant and define the domain
    \begin{align}
        \mathcal{D}_c = \{x \in \mathbb{R} : |x| < \sqrt{2n+1} \ \mathrm{ and } \ c \leq \varphi(x) \leq \pi - c \}.
    \end{align}
    Then, for each $n > 0$ and all $x \in \mathcal{D}_c$, we have
    \begin{align}
        \psi_n(x) =  \frac {(-1)^n 2^{\frac 1 4}}{\pi^{\frac 1 2}n^{\frac 1 4}}\frac 1 {\sqrt{\sin \varphi(x)}}\left( \sin \left[\left ( \frac n 2 + \frac 1 4 \right) (\sin (2 \varphi(x))-2 \varphi(x)) + \frac{3 \pi}4\right] + \mathcal{O} \left(\frac 1 n \right)\right).
    \end{align}
\end{lemma}

In~\cref{app:PRproperties} we describe other properties of these approximations.
For the case $n=0$, it will suffice to approximate $\psi_0(x)$ by a constant like  $\psi_0(0)\approx .75$, assuring constant overlap.

Hence, the Hermite functions can be approximated by an oscillatory term and an amplitude that depends on $x$ via $1/\sqrt{\sin (\varphi(x))}$, in the oscillatory region. This motivates the definition of a set of quantum states that have constant overlap with the finite-dimensional Hermite states $\dshermite{n}$ by considering, for example, a domain included in the oscillatory region and far from the `turning points' $\pm x_{\rm tp}$ where, for example,  $|x| \le \sqrt{(3/4)(2n+1)}$. In the following we disregard the phase $(-1)^n$ in the definition of $\psi_n(x)$ to ease the exposition.

\begin{lemma}[Plancherel-Rotach states]
\label{lem:PlancherelStates}
Let $N>0$ be the dimension for the QHT and
and $\eps>0$ be the error.
For all $0 \le n \le N-1$, let
\begin{align}
      \prvar_n(x) : =  \frac {2^{\frac 1 4}}{\pi^{\frac 1 2}n^{\frac 1 4}}\frac 1 {\sqrt{\sin \varphi(x)}}\left( \sin \left[\left ( \frac n 2 + \frac 1 4 \right) (\sin (2 \varphi(x))-2 \varphi(x)) + \frac{3 \pi}4\right] \right) \times g_n(x) \;,
\end{align}
where $\varphi(x) := \arccos(x/\sqrt{2n+1})$,
be the Plancherel-Rotach approximation of the $n^{\rm th}$ Hermite function. 
The function $g_n(x)$ is some smooth approximation to the indicator function and satisfies
\begin{align}
    g_n(x):= \left \{ \begin{matrix}
        0 & {\rm if} \ |x| \ge \sqrt{(3/4)(2n+1)}+ 1/ (10\sqrt{2n+1})\;, \\
        1 & {\rm if} \ |x| \le \sqrt{(3/4)(2n+1)} \;,\\
        \in (0,1) & {\rm if} \ \sqrt{(3/4)(2n+1)}+ 1/(10\sqrt{2n+1}) >|x| >\sqrt{(3/4)(2n+1)} \;.
    \end{matrix}\right .
\end{align}
For any $M >N$, define the $M$-dimensional `Plancherel-Rotach' quantum states 
\begin{align}
\label{eq:plancherelrotachstates}
    \prstate{n}:=\left( \frac{2\pi}M\right)^{1/4}\sum_{j=-J(n)}^{J(n)-1} \prvar_n(x_j)  \ket j \;, \; 0 \le n \le \lowenergy-1 \;,
\end{align}
where $x_j:=j \sqrt{2\pi/M}$ denotes the discretized space coordinate, and
\begin{align}
    J(n):= \left \lceil \sqrt{  \frac 3 4\frac{(2 n+1)M}{2 \pi}}\right \rceil 
\end{align}
is such that $J(n)\sqrt{2\pi/M} \approx  \sqrt{(3/4)(2n +1)}$ and $J(n)<M/2$. Let $\dshermite{n}$ be Hermite states of~\cref{eq:dshermite}, that is,
\begin{align}
    \dshermite{n}:=\left( \frac{2\pi}M\right)^{1/4}\sum_{j=-M/2}^{M/2-1} \psi_n(x_j)  \ket j \; , \; 0 \le n \le \lowenergy-1 \;.
\end{align}
Then, there exist constants $c>0$ and $c'>0$
such that, for all $M \ge c (\lowenergy)^{9/4}/\eps^{13/4}$ and all 
$0 \le n \le \lowenergy -1$, the overlap is
\begin{align}
\label{eq:plancherelrotachoverlap}
    \bra{\psi_n}\psi_n\rangle \ge c'  \;.
\end{align}
For every such $M$ there exists $N_{\rm high}=\cO(\lowenergy/\eps)$ satisfying $N < N_{\rm high}<M$
and,
for all $n \le \lowenergy-1$,
\begin{align}
\label{eq:lemmaleakage}
   \| \Pi_{> N_{\rm high}}  \prstate{n}\|^2\le \eps \;,
\end{align}
 where $\Pi_{> N_{\rm high}}$ is 
a `high-energy' projector onto the subspace
orthogonal to the subspace spanned by $\{\dshermite{n}\}_{0 \le n \le N_{\rm high}}$. (Note that 
$N_{\rm high} \ll M$ asymptotically.)

\end{lemma}

\begin{proof}
    The quantum states $\prstate{n}$
    are suggested by the Plancherel-Rotach approximation of \cref{lem:Plancherel}. Asymptotically, they can be shown to be subnormalized for $n \le \lowenergy-1$ because we are cutting off the domain by choosing $J(n)<M/2$.
    Our goal is then to carry the Plancherel-Rotach approximation to the discrete space, which involves approximating integrals by finite sums. We also note that in the subspace of interest, where $n \le \lowenergy-1$, we will satisfy $n \ll M$ asymptotically.

For each $n \le \lowenergy-1$, we are selecting as the domain of interest one where $|x| \le x_{\max}:=\sqrt{(3/4)(2n+1)}$.
In the definition of $\prstate{n}$, we let $x_j$
run up to $\pm J(n) \sqrt{2\pi/M}$, and our choice of $J(n)$ is such that $|x_j|$ is upper bounded by $x_{\max}$ up to an asymptotically small correction. This domain is purposely far from, and does not include, the turning points $\pm x_{\rm tp} = \pm \sqrt{(2n+1)}$ where the Plancherel-Rotach approximations are known to fail (e.g., $\varphi(x)=0$ at $x=x_{\rm tp}$). 

We will start by obtaining some properties of these approximations in the domain of interest. Note that rather than considering~\cref{lem:Plancherel} directly in this domain, we are modifying the approximations slightly to avoid issues when considering the momentum operator and the Fourier transform. For example, rather than assuming the functions to be exactly $0$ for all $x$ such that $|x| > x_{\max}$ and oscillating otherwise, we introduced smooth envelope functions $g_n(x)$ in the definition with the following properties (for each $n \ge 0$):
\begin{align}
\label{eq:windowfunction}
    g_n(x):= \left \{ \begin{matrix}
        0 & {\rm if} \ |x| \ge x_{\max}+ 1/ (10\sqrt{2n+1})\;, \\
        1 & {\rm if} \ |x| \le x_{\max} \;,\\
        \in (0,1) & {\rm if} \ x_{\max}+ 1/(10\sqrt{2n+1}) >|x| >x_{\max} \;.
    \end{matrix}\right .
\end{align}
There is nothing special about the term 
$1/ (10\sqrt{2n+1})$, other that it still guarantees the relevant $x$ to be far from the turning points, and that the magnitudes of the derivatives of $g_n(x)$ can be properly bounded. A good choice for these functions is given in Ref.~\cite{somma2019quantum} and obtained by convolving the indicator function with the bump function. Formally, if $\delta:=1/(20\sqrt{2n+1})$, these are the convolutions
\begin{align}
 \label{eq: bump function}
    g_n(x) :=\frac {2a}{\delta} \int_{-x_{\max}-\delta}^{x_{\max}+\delta} \rD x' \exp(-\frac{1}{1-4(x'-x)^2/\delta^2})   \; 
\end{align}
The constant is $a \approx 2.25$.
We will use these to prove \cref{lem:PlancherelStates}, but note that  other choices can also work. Besides the properties in~\cref{eq:windowfunction}, they also satisfy
$\frac{\rD}{\rD x}g_n(x)=\frac{\rD^2}{\rD x^2}g_n(x)=0$ for $|x|>x_{\max}+ 1/ (10\sqrt{2n+1})$ and $|x|<x_{\max}$.
Also, $|\frac{\rD}{\rD x}g_n(x)| =\cO( \sqrt{2n+1})$
and $|\frac{\rD^2}{\rD x^2}g_n(x)| =\cO(2n+1)$. 
From now on the domain where $|x|\le x_{\max}+1/(10\sqrt{2n+1})$ will be referred to as the `domain of interest'.

We will first establish the desired properties by working in the continuum, and use the following properties of the relevant functions.
In the domain of interest,
the Hermite functions satisfy
$|\psi_n(x)| =\cO(1/n^{1/4}) $ and we can use the property
$\frac{\rD}{\rD x}  \psi_n(x)=\frac 1 {\sqrt 2}(\sqrt n \psi_{n-1}(x)-\sqrt{n+1}\psi_{n+1}(x))
$
to show that  $|\frac{\rD}{\rD x}  \psi_n(x)|=\cO(n^{1/4})$, and Schr\"odinger equation 
    $\frac{\rD^2}{\rD x^2}  \psi_n(x) = (x^2-(2n+1)) \psi_n(x)$
to show  that
 $|\frac{\rD^2}{\rD x^2}  \psi_n(x)|=\cO(n^{3/4})$.
According to~\cref{lem:Plancherel},  we have $\phi_n(x) = \psi_n(x) + \cO(1/n^{5/4})$ in this domain.
Also, we can use the bounds in~\cref{app:PRproperties} and the chain rule for $\phi_n(x)=\tilde \phi_n(x) \times g_n(x)$ to determine
$|\phi_n(x)|=\cO(1/n^{1/4})$, 
 $|\frac{\rD}{\rD x}  \phi_n(x)|=\cO(n^{1/4})$, and
  $|\frac{\rD^2}{\rD x^2}  \phi_n(x)|=\cO(n^{3/4})$
in the domain of interest.

These properties allow us to establish a lower bound on the overlap:
\begin{align}
\label{eq:Hermiteoverlap}
    \int_{-\infty}^\infty \rD x \; \psi_n(x) \phi_n(x)& =\int_{-x_{\max}-1/(10\sqrt{2n+1})}^{x_{\max}+1/(10\sqrt{2n+1})} \rD x \;\psi_n(x) \phi_n(x) \\
    & = \int_{-x_{\max}-1/(10\sqrt{2n+1})}^{x_{\max}+1/(10\sqrt{2n+1})} \rD x \;\psi_n(x) (\psi_n(x) +\cO(1/n^{5/4}))\\
    &= \int_{-x_{\max}}^{x_{\max}} \rD x \;|\psi_n(x)|^2 + \cO \left( 1/(n+1/2)\right) \;.
\end{align}    
 In addition,
 \begin{align}
 \int_{-x_{\max}}^{x_{\max}} \rD x \;|\psi_n(x)|^2   &= 1  -2  \int_{x_{\max}}^{\infty} \rD x \;  |\psi_n(x)|^2 \\
        & \ge 1 -   \frac{2}{(x_{\max})^2}
         \int_{x_{\max}}^{\infty} \rD x \;  x^2 |\psi_n(x)|^2 \\
         & \ge 1 -   \frac{2}{(x_{\max})^2}
         \int_{0}^{\infty} \rD x \;  x^2 |\psi_n(x)|^2 \\
         & = 1   -\frac{1}{2(x_{\max})^2}
        (2n+1) \\
        & = 1/3 \;.
\end{align}
Hence, the overlap is lower bounded by $1/3 + \cO(1/n)$.
Numerical calculations show that this overlap actually approximates $2/3$ and remains close to $2/3$ for all $n \ge 0$. 

Next, we show that the  functions $ \phi_n(x)$ have negligible overlap with the high-energy sector, defined by some $N_{\rm high}$. This proof is the one that will use the smoothness property of $g_n(x)$ (i.e., a bounded second derivative); otherwise, a sharp cutoff could result in high energies. Recall that in the continuum, the QHO Hamiltonian is 
$\frac 1 2 (-\frac{\rD^2}{\rD x^2}+x^2)$.
The expected value of $x^2$ on $\phi_n(x)$, corresponding to the potential term, satisfies
\begin{align}
   \int_{-\infty}^\infty \rD x \; x^2 | \phi_n(x)|^2  
   & \le   \int_{-x_{\max}}^{x_{\max}} \rD x \; x^2 | \phi_n(x)|^2 \\
    & = \int_{-x_{\max}}^{x_{\max}} x^2 |\psi_n(x) + \cO(1/n^{5/4})|^2
    \\
   & = \int_{-x_{\max}}^{x_{\max}} \rD x \; x^2 | \psi_n(x)|^2 + \cO((x_{\max})^3 /n^{3/2}) 
   \\
   & = \int_{-x_{\max}}^{x_{\max}} \rD x \; x^2 | \psi_n(x)|^2 + \cO(1) \;,
\end{align}
and we also know for the Hermite functions
\begin{align}
    \int_{-x_{\max}}^{x_{\max}} \rD x \; x^2 | \psi_n(x)|^2  \le   \int_{-\infty}^{\infty} \rD x \; x^2 | \psi_n(x)|^2 = n+1/2 \;.
\end{align}
It follows that 
\begin{align}
   \int_{-\infty}^\infty \rD x \; x^2 | \phi_n(x)|^2 =\cO(n+1/2)\;.
   \end{align}
The expected value of the kinetic term satisfies
 \begin{align}
   -\int_{-\infty}^\infty \rD x \;  \phi_n(x)  \frac{\rD^2}{\rD x^2} \phi_n(x)  & =
   -\int_{-x_{\max}-1/(10\sqrt{2n+1})}^{x_{\max}+1/(10\sqrt{2n+1})} \rD x \;  \phi_n(x)  \frac{\rD^2}{\rD x^2} \phi_n(x) \\
   & = \cO(n +1/2)\;,
\end{align}
where we used the above properties of $\phi_n(x)$.
It follows that the energy given by these functions, which is the expectation of the QHO Hamiltonian on $\phi_n(x)$, satisfies
\begin{align}
    {\rm E}\left[\frac 1 2 (-\frac{\rD^2}{\rD x^2}+x^2) \right] \le C (n+1/2) \;,
\end{align}
for some $C>0$ that can be determined.
 We can then use Markov's inequality
 to bound the support of $\phi_n(x)$ in the high-energy space, spanned by Hermite functions with $n' > N_{\rm high}$; that is
    \begin{align}
    \label{eq:continuumleakage}
    \sum_{n'=N_{\rm high}+1}^\infty \left(
     \int_{-\infty}^\infty \rD x  \;       \phi_n(x)   \psi_{n'}(x) \right)^2 \le  \frac{C(n+1/2)}{N_{\rm high}+1/2}   \;.
    \end{align}

We readily proved some desired features of the approximated Hermite functions $\phi_n(x)$ in the oscillatory domain of interest. Our next goal is to prove that the finite-dimensional Plancherel-Rotach states of~\cref{eq:plancherelrotachstates} satisfy similar properties: their overlaps with the discrete Hermite states $\dshermite{n}$ is bounded by a constant, and their support on the high-energy subspace can be arbitrarily bounded. 
We will obtain these results via approximations of integrals by finite sums. We will make repeatedly use of a standard trapezoidal rule:
\begin{align}
\left|    \int_a^b \rD x \; f(x) -  h \sum_{k=1}^{M} \frac{f(x_{k-1})+f(x_k)} 2  \right| \le \frac{(b-a) h^2}{12} \max_{x \in (a,b)}\left| \frac{\rD^2}{\rD x^2}f(x)\right| \;.
\end{align}
Here $h$ is the size of the discretization, $M=(b-a)/h$, and $x_k = a + kh$.
While this rule will suffice because the scaling of the algorithm is only logarithmic in the dimension, we note that improved results could be obtained with a detailed analysis that uses exponentially convergent trapezoidal rules, since the functions are smooth. Nevertheless, we follow the standard trapezoidal rule to simplify the proof.

Consider the overlap
\begin{align}
    \bra{\psi_{n}} \phi_n \rangle= \left(\frac {2\pi}M \right)^{1/2}\sum_{j=-M/2}^{M/2} \psi_n(x_j)   \phi_n(x_j)\;,
\end{align}
which approximates~\cref{eq:Hermiteoverlap}
 using a discretization of size $\sqrt{2\pi/M}$. Indeed, using the trapezoidal rule above we can show that
$ \bra{\psi_n} \phi_n \rangle \approx \int_{-\infty}^\infty \rD x \; \psi_n(x)  \phi_n(x)$ within additive error
\begin{align}
    \cO \left(\frac{n+1/2} M \right) 
\end{align}
since $   \phi_n(x)$  is zero for $|x|\ge x_{\max}+1/(10\sqrt{2n+1})$. For this we used
the properties of the functions to show
\begin{align}
      \left|\frac{\rD^2}{\rD x^2} (\psi_n(x) \phi_n(x))\right| 
      & \le 
      \left|\frac{\rD^2}{\rD x^2} \psi_n(x) \phi_n(x)\right| + 2\left| \frac{\rD}{\rD x} \psi_n(x)   \frac{\rD}{\rD x} \phi_n(x)\right| 
      + \left| \psi_n(x) \frac{\rD^2}{\rD x^2} \phi_n(x)\right| \\
      & = \cO(n^{1/2})\;
\end{align}
Then, a lower bound to $\bra{\psi_{n}} \phi_n \rangle$
is
\begin{align}
    1/3 +  \cO \left(\frac{n+1/2} M \right) +
     \cO \left(\frac{1}{n+1/2} \right) \;.
\end{align}
Note that we obtained a constant lower bound in the asymptotic regime, which suffices for our goal of the efficient Hermite transform (i.e., we will always choose $M$ large enough to make $(n+1/2)/M$ very small), but a suitable value of $c>0$ in~\cref{eq:plancherelrotachoverlap} 
for all $n \ge 0$ can be obtained with a detailed analysis. Indeed, numerical simulations show that this overlap is approximately $2/3$ for many $n$'s. See~\cref{fig:overlap}.

\begin{figure}
\centering
\includegraphics[width=10cm]{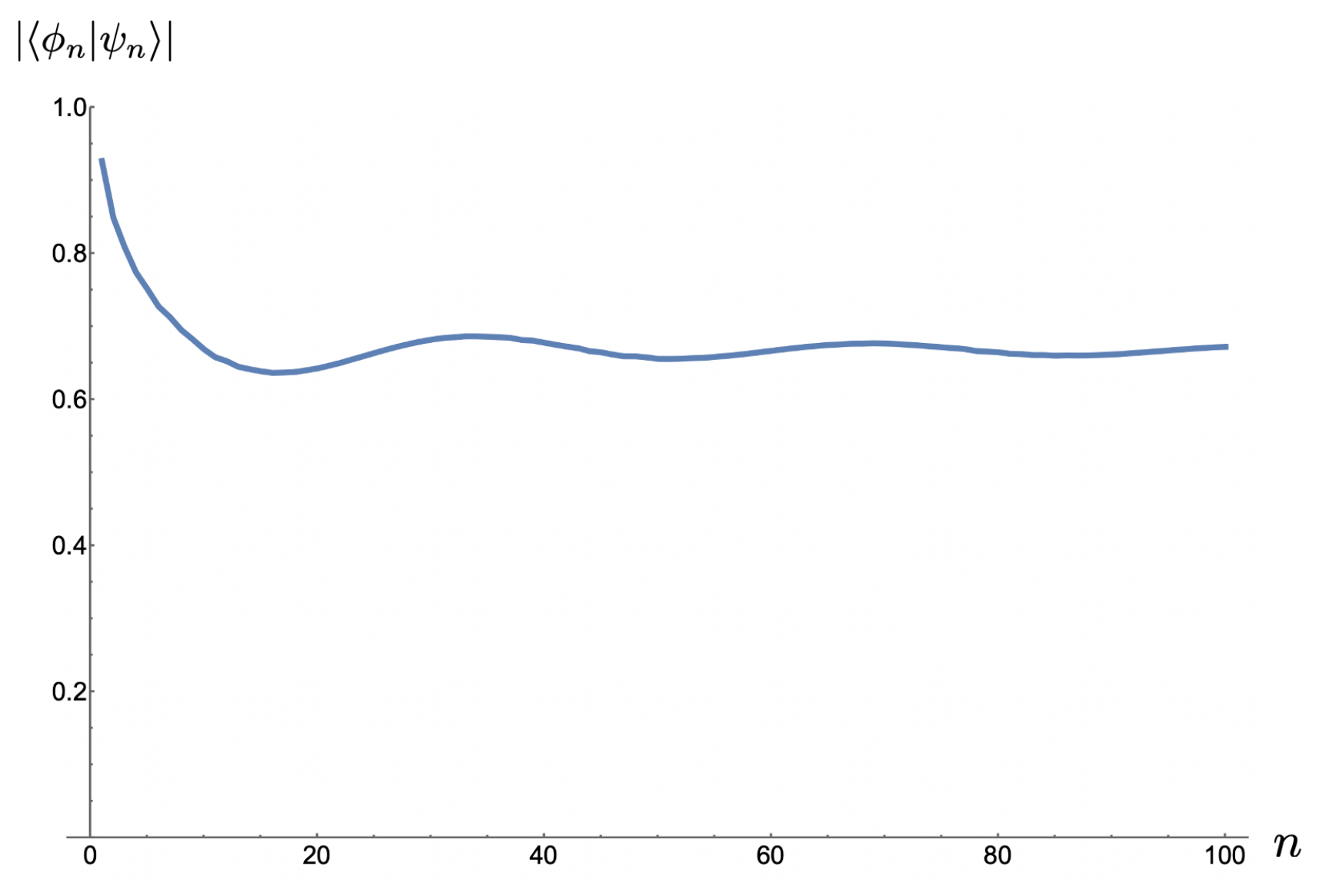}
\caption{The overlap between the Plancherel-Rotach states and the Hermite states for dimension $M=10^5$ and $0 \le n \le 100$. The overlap approaches $2/3$, which corresponds to the probability mass of the Hermite functions in the domain $|x| \le \sqrt{(3/4)(2n+1)}$.}
\label{fig:overlap}
\end{figure}

Our next goal is to prove that the support of the Plancherel-Rotach states $\prstate{n}$, where $0 \le n \le \lowenergy-1$, is arbitrarily small in the subspace of the discrete QHO states specified by an $N_{\rm high}>\lowenergy$. In the continuum we already proved that, expressing $\phi_n(x) = \sum_{n'=0}^\infty \alpha_{n,n'} \psi_{n'}(x)$, then 
\begin{align}
    \sum_{n' > N_{\rm high}}|\alpha_{n,n'} |^2 \le C \frac{n+1/2}{N_{\rm high}+1/2} \;.
\end{align}
This is equivalent to~\cref{eq:continuumleakage} since Hermite functions are orthogonal.
Consider now $\ket{\phi_n}$ and $\ket{ \psi_{n'}}$ for $0 \le n \le \lowenergy-1$ and $0 \le n' \le N_{\rm high}$. Note that, in general for all $ x \in \mathbb R$,
\begin{align}
 \left |\psi_{n'}(x)\right|\lesssim 1.086 \; ,  \;   \left | \frac{\rD}{\rD x}\psi_{n'}(x)\right|  \le c (2n'+1)^{1/4} \; ,    \;  \left | \frac{\rD^2}{\rD x^2}\psi_{n'}(x)\right|  \le c' (2n'+1)^{3/4} \;.
\end{align}
for some constants $c>0$ and $c'>0$.
Since $n'$ can be larger than $n$, then we cannot use improved bounds for these functions. It follows that
\begin{align}
    \left| \frac{\rD^2}{\rD x^2} \left(\psi_{n'}(x) \phi_n(x)\right)\right| = \cO \left( ( N_{\rm high}+1/2)^{3/4}\right) \;.
\end{align}
Then, using the trapezoidal rule we have 
$\bra{\psi_{n'}}  \phi_n \rangle \approx \int_{-\infty}^\infty \rD x \; \psi_n'(x)  \phi_n(x)$ within additive error
\begin{align}
    \cO \left(\frac{(n+1/2)^{1/2}( N_{\rm high}+1/2)^{3/4}} M \right) = \cO\left(\frac{( N_{\rm high}+1/2)^{5/4}} M \right)\;.
\end{align}
since we only need to integrate in the domain of interest.
Then,
\begin{align}
   \bra{\psi_{n'}} \phi_n \rangle =  \alpha_{n,n'}+\cO\left(\frac{(N_{\rm high}+1/2)^{5/4}} M \right) \;.
\end{align}
The trapezoidal rule also implies
\begin{align}
    \|\ket {\phi_n} \|^2& = \int_{-x_{\max}-1/(10\sqrt{2n+1})}^{x_{\max}+1/(10\sqrt{2n+1})} \rD x \; | \phi_n(x)|^2 + \cO\left( \frac{(n+1/2)^{5/4}}{M}\right) \\
    & = \sum_{n'=0}^\infty| \alpha_{n,n'}|^2 + \cO\left( \frac{(n+1/2)^{5/4}}{M}\right)  \;,
\end{align}
since the prior bounds imply $|\frac{\rD^2}{\rD x^2}   \phi_n(x)| = \cO(n^{3/4})$ in the domain of interest.

Let $M > N_{\rm high}>N$.
Consider the subspace spanned by 
$\{\dshermite{n}\}_{0 \le n \le N_{\rm high}}$
and its orthogonal complement; both span $\mathbb C^M$.
Let $\Pi_{>N_{\rm high}}$ be the projector onto the latter. Without loss of generality we can write
\begin{align}
    \ket{\phi_n}=\sum_{n'=0}^{N_{\rm high}}\beta_{n,n'}\dshermite{n'} + \ket{\phi_n^\perp}
\end{align}
where $\ket{\phi_n^\perp}=\Pi_{>N_{\rm high}}  \ket{\phi_n}$ is subnormalized and the amplitudes are $\sum_{n'=0}^{N_{\rm high}}|\beta_{n,n'}|^2 \le 1$.
Note that $\bra{\psi_{n'}}  \phi_n\rangle = \beta_{n,n'}+ \cO(N_{\rm high}\exp(-\Omega(M)))$ for $n' \le N_{\rm high}$.
Our goal is to bound $\|\ket{\phi_n^\perp}\|$. 
Combining the prior results we have
\begin{align}
\|\ket{\phi_n^\perp}\|^2  &=      \|\ket {\phi_n} \|^2 - \| \sum_{n'=0}^{N_{\rm high}}\beta_{n,n'}\dshermite{n}\|^2 \\
& = \sum_{n'=0}^\infty| \alpha_{n,n'}|^2 + \cO\left( \frac{(n+1/2)^{5/4}}{M}\right) - \sum_{n'=0}^{N_{\rm high}}|\beta_{n,n'}|^2 + \cO(N_{\rm high} \exp (-\Omega(M))) \\
& = \sum_{n'=0}^\infty| \alpha_{n,n'}|^2 + \cO\left( \frac{(n+1/2)^{5/4}}{M}\right) - \sum_{n'=0}^{N_{\rm high}} 
|\bra{\psi_{n'}} \phi_n\rangle|^2 + \cO((N_{\rm high})^2 \exp (-\Omega(M))) \\
& =  \sum_{n'=0}^\infty| \alpha_{n,n'}|^2 + \cO\left( \frac{(n+1/2)^{5/4}}{M}\right) - \sum_{n'=0}^{N_{\rm high}} 
|\alpha_{n,n'}|^2 + \cO\left(N_{\rm high}\frac{(N_{\rm high}+1/2)^{5/4}} M \right) \\
& = \sum_{n'>N_{\rm high}} | \alpha_{n,n'}|^2
+ \cO\left(\frac{ N_{\rm high}^{9/4}} M \right) \\
& = \cO\left(\frac{ N+1/2} {N_{\rm high}+1/2} \right) + \cO\left(\frac{ (N_{\rm high}+1/2)^{9/4}} M \right) \;.
\end{align}
We dropped the term exponentially small in $M$ since it is asymptotically subdominant.
To establish $\|\ket{\phi_n^\perp}\|^2=\cO(\eps)$ it suffices to choose $N_{\rm high}=\Omega( N/\eps)$ and
$M = \Omega((N_{\rm high})^{9/4}/\eps)$.
This shows~\cref{eq:lemmaleakage} and completes the proof.

\end{proof}
    
\paragraph{Efficient preparation of the Plancherel-Rotach states.}

Thus far we considered the Plancherel-Rotach  states defined in~\cref{eq:plancherelrotachstates} as good initial states to prepare the corresponding $\dshermite{n}$, since they have overlap bounded by a positive constant. For the quantum Hermite transform to be efficient, we need to show that these states can be efficiently prepared.
Note that it suffices to produce $\cO(\eps)$-approximations to the $\ket{\phi_n}$ since this will imply a QHT with error $\cO(\eps)$. Unfortunately, we cannot allow for a constant-error approximation to $\ket{\phi_n}$ because we need to make sure that the support of the state in the high-energy space is still bounded by $\cO(\eps)$.

\begin{lemma}[Efficient state preparation]
\label{lem:state prep}
Let $0 \le n \le N-1$ and $\eps$ the error. Then there is a quantum circuit that performs the map
\begin{align}
    \ket n \mapsto \ket n \frac{{\ket{\phi_n}}}{\|{\ket{\phi_n}}\|}
\end{align}
    within error $\eps$ using $\cO((\log N + \log(1/\eps))^3)$ gates.
\end{lemma}

\begin{proof}
The states $\ket{\phi_n}$  admit a simple form, being a linear combination of basis states with amplitudes proportional to
\begin{align}
  A_n(x):=  \frac 1 {\sqrt{\sin \varphi(x)}} \times g_n(x) =
    \frac 1 {\sqrt{1-\frac{x^2}{2n+1}}} \times g_n(x)
\end{align}
and phases
\begin{align}
\pm  \Theta_n(x):=  \pm \left[\left ( \frac n 2 + \frac 1 4 \right) \left( \sin(2\varphi(x) - 2 \varphi(x) )+\frac{3\pi}4\right)\right] \;,
\end{align}
where $\varphi(x)=\arccos(x/\sqrt{2n+1})$.
That is, $\phi_n(x) \propto A_n(x) (e^{+\ri \Theta_n(x)}-e^{-\ri \Theta_n(x)})$.
We will describe a method to prepare
a state with one of these phases, say $+\Theta_n(x)$, and the linear combination can be obtained via a simple Hadamard gate. This is because, while $\ket{\phi_n}$ are subnormalized, they have norm bounded by a positive constant. 
There are several known state preparation methods that can be useful for our context. For example, we could repurpose the algorithm of conditional rotations~\cite{Z98, GR02} and, to restrict the domain to within the oscillatory region we can use quantum rejection sampling~\cite{ORR12}. A more recent approach is that based on inequality testing in~\cite{sanders2019black} 
and is the one we will follow here.

For each $0 \le n \le N-1$, the states $\ket{\phi_n}$
are superpositions over basis states $\ket j$ where, in our convention, $-J(n) \le j \le J(n)-1$, and $J(n)$ is described in~\cref{lem:PlancherelStates}. They are 
in the space $\mathbb C^M$, where the number of qubits is
$m=\log_2 M$.
To simplify the details of state preparation, we are going to use the standard convention where basis states are labeled as $\{\ket 0 , \ket 1, \ldots, \ket {M-1}\}$. As mentioned, we can go from one convention to the other via a cyclic permutation (see below). Initially, we will consider setting the  amplitudes proportional to $A_n(x)$ and then we set the  phases $e^{\ri \Theta_n(x)}$. For the amplitudes we assume access to 
an oracle $\mathfrak{amp}$ that computes $A(x)$
with $r$ bits of precision in an additional register. 
Note that $A(x)$ is bounded by a constant in the oscillatory region of interest. We are going to introduce ancillary registers in many steps, but the relevant
parts are flagged by the states $\ket 0$.

\begin{itemize}
    \item Step 1: Given $n$, compute the smallest integer $q$ that satisfies $2^q \ge 2 J(n)$.
    Apply $q$ Hadamard gates to map $\ket 0^{\otimes m} \mapsto \frac 1 {\sqrt{2^q}}\sum_{j=0}^{2^q-1}\ket j$.
    Apply an inequality test to obtain $\frac 1 {\sqrt{2^q}}(\sum_{j=0}^{2J(n)-1}\ket j \ket 0+ \sum_{j=2J(n)}^{2^q-1} \ket j \ket 1$, which flags
    the relevant support. 
    \item Step 2: Use $\mathfrak{amp}$ and inequality test
    to perform the map $\ket j \mapsto \ket j \ket{\bar A_n(x_j)} \frac 1 {\sqrt {2^r}}(\sum_{z=0}^{\bar A_n(x_j)-1}\ket z \ket 0 + \sum_{\bar A_n(x_j)}^{2^r} \ket z \ket 1)$, where $\bar A_n(x_j)$ is an $r$-bit approximation of $A_n(x_j)$ (rescaled by a constant).
    \item{Step 3:} Apply $r$ Hadamard gates and the inverse of $\mathfrak{amp}$ to map the prior state to 
    $\ket j \frac {\bar A_n(x_j)} {2^r} \ket 0^{\otimes r+1} + \ket{\omega_j^\perp}$, where $\ket{\omega_j^\perp}$ is orthogonal to $\ket 0^{\otimes r+1}$ and subnormalized.
\end{itemize}

Following this steps we have essentially prepared
the state
\begin{align}
   \frac 1 {\sqrt {2^q}} \sum_{j=0}^{2J(n)-1}  
   \frac {\bar A_n(x_j)} {2^r} \ket j \ket 0^{\otimes r+2 }
   + \ket{\omega_n^\perp}
\end{align}
where again $\ket{\omega_n^\perp}$ is subnormalized
and orthogonal to $\ket 0^{\otimes r+2 }$.
We can use fixed point amplitude amplification to select the desired part of this state and obtain an $\cO(\eps)$ approximation to 
\begin{align}
    \frac{ \sum_{j=0}^{2J(n)-1}  
   \frac {\bar A_n(x_j)} {2^r} \ket j } {\|\sum_{j=0}^{2J(n)-1}  
   \frac {\bar A_n(x_j)} {2^r} \ket j\|}\;.
\end{align}
Amplitude amplification introduces a multiplicative cost that is $\cO(\log(1/\eps))$, since the relevant overlaps are constant.

We discuss the relevant complexities of this approach. 
First note that it suffices to prepare an $\cO(\eps)$
approximation to $\ket{\phi_n}$, since the QHT is performed within this precision and $\ket{\phi_n}$ has overlap with $\dshermite{n}$ bounded by a constant.
It then suffices to compute the amplitudes $A(x)$
within multiplicative error $\cO(\eps)$, but since these are bounded from above by a constant, we can compute them within additive error $\cO(\eps)$. It follows that the number of bits of precision is $r =\cO(\log(1/\eps))$.
We can then construct $\mathfrak{amp}$  using coherent arithmetics; one procedure would involve first computing $\varphi(x)$ within error $\cO(\eps)$, which follows if $x/\sqrt{2n+1}$ is obtained within that complexity, and then computing 
$1/\sqrt{\sin \varphi(x)}$ within this error.
For this step it suffices to first compute 
$\sqrt{2\pi/M}/\sqrt{2n+1}$ within error $\cO(\eps/M)$,
which can be accomplished with cost $\cO((\log M + \log(1/\eps))^2)$. Then we multiply with $j$, and since $j$ is expressed with $\log M$ bits, this is subdominant.   We also need to compute $g_n(x)$ within error $\cO(\eps)$. Since $g_n(x)$ can be obtained by rescaling some $g(x)$~\cite{somma2019quantum}, this translates to computing $g(x)$
within error $\cO(\eps/\sqrt n)$. 
This function involves the bump function and for this precision the cost is $\tilde \cO((\log n + \log(1/\eps))^2)$, hiding a doubly logarithmic cost.
This function is smooth and also involves an integral, and this increases the cost by a factor of $\cO(\log(1/\eps))$. So the total cost of computing $g_n(x)$ is $\cO((\log n + \log(1/\eps))^2 \log (1/\eps))$.

There are additional gates used in the procedure, the Hadamard gates and those for inequality tests, which are also subdominant.

The next step involves including the phases. 
Again, it suffices to compute these phases within $r'=\cO(\log(1/\eps))$ bits of precision.
We can construct an oracle $\mathfrak{phase}$
that performs the map $\ket j \mapsto \ket j \ket{\bar \Theta(x_j)}$ and then use phase kickback and the inverse of  $\mathfrak{phase}$ and obtain $e^{\ri \bar \Theta(x_j)} \ket j$. To this end it suffices to compute $\varphi(x)$ within additive error $\cO(\eps/N)$. 
This follows if $x/\sqrt{2n+1}$ is computed with that error which again can be done with cost $\cO((\log M + \log(1/\eps))^2)$, as explained in the construction of $\mathfrak{amp}$.

The last step is to perform the cyclic permutation, since in the standard convention the states $\ket{\phi_n}$
involve a superposition of basis states from $-J(n)+M/2$ to $J(n)-1+M/2$. This can be accomplished 
with gate cost that is $\log^2 M$ using the standard QFT,
as explained in~\cite{somma2016}.

According to~\cref{lem:PlancherelStates}, $M = \poly(N,1/\eps)$ and including the cost of amplitude amplification to obtain the normalized state, the result follows.

\end{proof}

\paragraph{Fast eigenstate filtering.}
The next step in the QHT is that of state `filtering', which is a transformation that ideally would flag the Hermite state and perform the map
\begin{align}
\label{eq:eigenstate filtering}
    \prstate{n} \mapsto \beta_n \dshermite{n} \ket 0^{\otimes m} + \ket{0_n^\perp} \;,
\end{align}
where $\beta_n=\bra{\psi_n}\phi_n \rangle$ is the (bounded) overlap, and $\ket 0^{\otimes q}$ and $\ket {0_n^\perp}$
are orthogonal states of the ancilla qubits. 
Since we would like to implement this transformation for an exponentially large set of values of $0 \le n \le \lowenergy -1$, we will combine the fast-forwarding results for the discrete QHO in~\cref{sec:fastforwarding} with the conventional QPE algorithm~\cite{K95}. 
\begin{lemma}
\label{lem:eigenstate filtering}
Let $0 \le n \le N-1$ and $\eps$ the error, where $N > \log(1/\eps)$. Then there is a quantum circuit that performs the map in 
\cref{eq:eigenstate filtering}  within error $\eps$ using $\cO((\log N + \log(1/\eps))^3)$ gates.
\end{lemma}

\begin{proof}
For given $N$ and $\eps$, we set $M$ and $N_{\rm high}$
according to~\cref{lem:PlancherelStates}.
We also write $M=2^m$ and let $U(t)=e^{-\ri \discretehamiltonian t}$ be the time-evolution unitary for the discrete QHO $\discretehamiltonian \in \mathbb C^{M \times M}$, where $|t| \le 2\pi$.
Plancherel-Rotach states are mostly supported in the subspace where $n \le N_{\rm high}$. For all $0 \le n \le N_{\rm high}$ define 
$U_n:=U(2 \pi/M)e^{\ri (2\pi/M)(n+1/2)}$. In the low-energy subspace of interest, 
where $n \le N_{\rm high}$,
we already proved $U(t) \dshermite{n}= e^{-\ri (n+1/2)t}\dshermite{n} + \cO(\exp(-\Omega(M)))$ and hence,
\begin{align}
 \frac 1 M (\one +   U_n^{2^{m-1}}) \ldots (\one +   U_n) \dshermite{{n'}} &  = \cO(\exp(-\Omega(M)))\;,
\end{align}
for all $0 \le n,n' \le N_{\rm high}$ and $n \ne n'$.
It follows that if we were to run the standard QPE algorithm with $m$ ancillas and controlled-$U_n^{2^j-1}$ unitaries, we would already be performing the desired eigenstate filtering within exponentially small error
$\cO(\exp(-\Omega(M)))$, which is subdominant. This approach, however, requires implementing $U$ a number of times that is exponential and scales with $N$ or $M$.

To fast forward eigenstate filtering, we are going to replace 
\begin{align}
\label{eq:Wforfiltering}
    U_n^{2^j} \rightarrow W_{n,j}:=V(2^j (2\pi/M)) 
    e^{i 2^j (2\pi/M)(n+1/2)}\;,
\end{align}
where $V(t)$
    is the approximation to $U(t)=e^{-\ri \discretehamiltonian t}$ using the sequence of three exponentials as in \cref{thm:discrete factoring}.  
    The new QPE circuit would then have controlled-$W_{n,j}$ operations rather than the controlled-$U_n^{2^j-1}$. Our fast-forwarding results also imply
    \begin{align}
 \frac 1 M (\one +   W_{n,m-1}) \ldots (\one +   W_{n,1}) \dshermite{n} &  = \dshermite{n}  +\cO(\exp(-\Omega(N))) \\
\frac 1 M(\one +   W_{n,m-1}) \ldots (\one +   W_{n,1}) \dshermite{n'} &  = \cO(\exp(-\Omega(N)))\;,
\end{align}
again for $0 \le n,n' \le N_{\rm high}$ and $n \ne n'$.

Consider now an arbitrary input Plancherel-Rotach state $\prstate{n}$.
We have shown $\prstate{n}=\sum_{n=0}^{N_{\rm high}}
\beta_n \dshermite{n}+\ket{\phi_n^\perp}$,
where $\|\ket{\phi_n^\perp}\|=\cO(\eps)$, for our choice of parameters in~\cref{lem:PlancherelStates}.
Hence, for this state our QPE algorithm performs
\begin{align}
  \prstate{n} \mapsto \beta_n \dshermite{n}\ket 0^{\otimes m} + \ket{0_n^\perp}
\end{align}
within error that is $\cO(\eps + \exp(-\Omega(N)))$.
Assuming $N> \log(1/\eps)$,
the $\cO(\eps)$ term dominates, and this is the desired eigenstate filtering operation of~\cref{eq:eigenstate filtering}.

The number of operations for eigenstate filtering 
can be estimated as follows. Each $W_j$ involves three exponentials, and there are $m-1$ calls to them. We already established that each exponential requires $\cO((\log M)^2)$ gates (using schoolbook arithmetics).
Then, the gate cost is $\cO((\log M)^3)=\cO((\log N + \log(1/\eps))^3)$.

\end{proof}

\paragraph{Fixed-point amplitude amplification.}

As a next step we need to show that amplitude amplification enables
the efficient preparation of the Hermite states from the Plancherel-Rotach states, even when the overlaps for different $n \le \lowenergy$ are different, but still bounded from below by a positive constant. After eigenstate filtering, our goal is to implement the map
\begin{align}
\label{eq:AAstateprep}
    \sum_{n=0}^{N-1} \frac{\alpha_n}{\|\ket{\phi_n}\|} \ket n (\beta_n \dshermite{n}\ket{0}^{\otimes m}+\ket{0^\perp_n}) \mapsto 
     \sum_{n=0}^{N-1} \alpha_n \ket n \dshermite{n} \;,
\end{align}
where the $\alpha_n$ are arbitrary and $\beta_n=\bra{\psi_n}\phi_n\rangle$.
To this end we start from results on fixed-point amplitude amplification in Refs.~\cite{YLC14,gilyen2019quantum}.
These will ultimately give the desired mapping even when the overlaps $\beta_n$ are different.

\begin{lemma}
    \label{lem:stateprepviafixedpointAA}
    Let $N>0$ and $\eps$ the error.
    Let $R_0 = (\one - 2 \ketbra 0^{\otimes m})$
    be the reflection over the ancilla state $\ket 0 ^{\otimes m}$ and $\cU:= \sum_{n=0}^{N-1} \ketbra n \otimes \cU_n + \sum_{n=N}^{M-1} \ketbra n \otimes \one$ be a conditional unitary, where $\cU_n$ is the unitary that prepares the Plancherel-Rotach states $\ket{\phi_n}/\|\ket{\phi_n}\|$, and let $\cV:= \sum_{n=0}^{N-1} \ketbra n \otimes \cV_n + \sum_{n=N}^{M-1} \ketbra n \otimes \one$ be another conditional unitary, where $\cV_n$ is the unitary that performs eigenstate filtering.
    Then there is a quantum circuit that performs the map in~\cref{eq:AAstateprep} within error $\eps$ using $\cU$, $\cV$, $R_0$, their inverses, and arbitrary one qubit gates $\cO(\log(1/\eps))$ times.
\end{lemma}

Before providing the proof, we note that our result relies
on, and generalizes the following fixed-point amplitude amplification.

\begin{lemma}[Fixed-point amplitude amplification, Thm. 27 of Ref.~\cite{gilyen2019quantum}]
\label{lem:fixed-pointAA}
    Let $U$ be a unitary and $\Pi$ be an orthogonal
projector such that $a \ket{\Psi_G} = \Pi U \ket{\Psi_0}$, and $a > \delta > 0$. There is a unitary circuit $\tilde U$ such that
$\| \ket{\Psi_G}-\tilde U \ket{\Psi_0} \| \le \eps$, which uses a single ancilla qubit and consists of $\cO (\log(1/\eps)/\delta)$ $U$, $U^\dagger$, $C_\Pi NOT$, $C_{\ketbra{\Psi_0}}NOT$, and $e^{\ri \phi \sigma_Z}$ one qubit gates.
\end{lemma}
Here, $a$ is unknown but the lower bound $\delta$ is known. The operation $C_\Pi NOT$ is simply an operation that coherently checks whether the state is supported in $\Pi$ or not, and flips the state of another qubit based on the outcome. Similarly,   $C_{\ketbra{\Psi_0}}NOT$ performs a flip on the qubit depending on whether the state is $\ket{\Psi_0}$ or not.
Also, fixed point amplitude amplification uses an ancillary qubit, which starts in $\ket 0$ and ends in $\ket 0$ within error $\eps$. This is important for~\cref{eq:AAstateprep}
to avoid changing the amplitudes $\alpha_n$.
A generalization of this lemma, which we give below as it can be of independent interest, is used in the proof of~\cref{lem:stateprepviafixedpointAA}.
\begin{lemma}[Fixed-point amplitude amplification in subspaces]
\label{lem:subspacefixed-pointAA}
Let $N >0$ be the dimension. 
For all $0 \le n \le N-1$, let $U_n$ be a unitary,
$\Pi_n$ be an orthogonal projector such that $a_n \ket{\Psi_G^n}=\Pi_n U_n \ket{\Psi_0^n}$, and $a_n>\delta>0$. 
There is a unitary quantum circuit $\tilde \cU$ such that
$\|\sum_{n=0}^{N-1}\alpha_n \ket n \ket{\Psi_G^n}- \tilde \cU \sum_{n=0}^{N-1}\alpha_n \ket n \ket{\Psi_0^n}\|\le \eps$,
where the amplitudes $\alpha_n$ are arbitrary and $\sum_n |\alpha_n|^2=1$, which uses a single ancilla qubit and consists of $\cO (\log(1/\eps)/\delta)$ $\cU$, $\cU^\dagger$, $C_\Pi NOT$, $C_{\Pi_0}NOT$, and $e^{\ri \phi \sigma_Z}$ one qubit gates, where $\cU=\sum_n \ketbra n \otimes U_n$, $\Pi =\sum_n \ketbra n \otimes \Pi_n$, and $\Pi_0=\sum_n \ketbra n \otimes \ketbra{\Psi_0^n}$.
\end{lemma}

\begin{proof}
    Fixed-point amplitude amplification is based on Chebyshev approximations; in the case of~\cref{lem:fixed-pointAA}, 
    we can assume these to be, for example, linear combinations of Chebyshev polynomials of
\begin{align}
    B=\begin{pmatrix}
        0 & A \cr A^\dagger & 0 
    \end{pmatrix} \; , \; A :=\ketbra{\Psi_0} U^\dagger \Pi =
   \bar a \ket{\Psi_0}\bra{\Psi_G}\; ,
\end{align}
where $\bar a$ is the complex conjugate.
The linear combination can be implemented with techniques like QSVT and implements the map $\ket{\Psi_0}\mapsto \ket{\Psi_G}$.
The complexity is determined by the degree of the Chebyshev polynomials in the approximation, and this is $\cO (\log(1/\eps)/\delta)$; even when $a$ is unknown (but $\delta$ is known), the polynomial allows for the correct mapping.  

To work in subspaces, we simply replace $B$ by the direct sum: $B:=\bigoplus_n B_n$, where 
\begin{align}
    B_n:=\begin{pmatrix}
        0 & A_n \cr A_n^\dagger & 0 
    \end{pmatrix} \; , \; A_n :=\ketbra{\Psi_0^n} U_n^\dagger \Pi_n =
   \bar a_n \ket{\Psi_0^n}\bra{\Psi_G^n}\; .
\end{align}
The Chebyshev polynomial approximation in each subspace will perform the mapping $\ket{\Psi_0^n}\mapsto \ket{\Psi_G^n}$, within error $\eps$, and the transformation carries to an arbitrary state $\sum_{n=0}^{N-1}\alpha_n \ket n \ket{\Psi_0^n}$ due to linearity. Note that while fixed-point amplitude amplification uses an ancilla qubit, the ancilla state stays in $\ket 0$ at the end, within error $\cO(\eps)$, for all $n$. 
Since we are using the same approximation, the degree is unchanged and still 
$\cO (\log(1/\eps)/\delta)$. 

We also note that $B\equiv \sum_n \ketbra n \otimes B_n$,
and accordingly we can define
\begin{align}
    A&:= \left(\sum_n \ketbra n \otimes \ketbra{\Psi_0^n}\right)
    \left(\sum_n \ketbra n \otimes U_n^\dagger \right) 
     \left(\sum_n \ketbra n \otimes \Pi_n \right)  
     = \sum_n  \ketbra n \otimes A_n \;.
\end{align}
It follows that we can build the walk operator to implement the Chebyshev polynomials using $\cU$, $C_\Pi NOT$, $C_{\Pi_0}NOT$, their inverses and arbitrary gates, a constant number of times. The result follows from multiplying these complexities with the degree.
\end{proof}

{\bf Proof of~\cref{lem:stateprepviafixedpointAA}:}
Given $N$ and $\eps$, the dimension $M$ is set according to~\cref{lem:PlancherelStates}.
As mentioned, the proof relies on~\cref{lem:subspacefixed-pointAA} and we need to define the corresponding operations. 
For all $0 \le n \le N-1$,
we let $\ket{\Psi_0^n}=\ket 0^{\otimes 2 m}$
and $\ket{\Psi_G^n}=\dshermite{n}\ket0 ^{\otimes m}$.
The unitaries $U_n$ are a sequence of two unitaries, the unitary $\cU_n$ to prepare the Plancherel-Rotach state $\ket{\phi_n}/\|\ket{\phi_n}\|$ and another unitary $\cV_n$ that performs eigenstate filtering as in~\cref{lem:eigenstate filtering}:
\begin{align}
   \frac{\beta_n}{\|\ket{\phi_n}\|} \dshermite{n}\ket 0^{\otimes m}= (\one \otimes \ketbra 0^{\otimes m})   U_n \ket{0}^{\otimes 2m} \;.
\end{align}
The projectors are then $\Pi_n = \ketbra 0^{\otimes m}$ for all $n$. 
The subspace where $M-1 \ge n \ge N$ is irrelevant since we only need to perform the right transform when $n \le N-1$. Hence, we can freely define $U_n$ and $\Pi_n$ in that subspace; for example we can set $U_n=\one$ and $\Pi_n=\ketbra 0^{\otimes m}$ for $n\ge N$.

With these definitions we fit the framework of~\cref{lem:subspacefixed-pointAA}.
Note that, since we showed $\beta_n /\|\ket{\phi_n}\|$ is bounded from below by a positive constant, the number of amplitude amplification rounds is $\cO(\log(1/\eps))$.
This is the largest degree of the Chebyshev polynomial in the approximation to implement fixed-point amplitude amplification. 
The overall complexity, as determined by the uses of $\cU=\sum_n \ketbra n \otimes \cU_n$, $\cV=\sum_n \ketbra n \otimes \cV_n$, $R_0=\one - 2 \ketbra 0^{\otimes m}$, and their inverses, is then $\cO(\log(1/\eps))$. The number of one-qubit gates scales similarly due to~\cref{lem:fixed-pointAA}.
\qed

\vspace{0.2cm}

Last we comment on the complexity of implementing the conditional unitaries.
The unitary $\cU_n$ is discussed in~\cref{lem:state prep}.
It requires performing a sequence of computations that depend on $n$, including $J(n)$, $x_j/\sqrt{2n+1}$, and it uses ${\mathfrak amp}$. We can run these computations coherently without adding a relevant factor to the gate cost and produce a version of ${\mathfrak amp}$ that now is not only conditional on $\ket j$ but also on $\ket n$. Hence the gate complexity of the conditional gate $\cU$ is dominated by the largest gate complexity of the $\cU_n$'s, which is the one of~\cref{lem:state prep}. 
The unitary $\cV_n$ is discussed in~\cref{lem:eigenstate filtering}. It requires implementing a version of QPE with unitaries $W_{n,j}$ that depend on $n$ in that there is a phase
$e^{\ri 2^j (2\pi/M)(n+1/2)}$, but all other operations are independent of $n$; see~\cref{eq:Wforfiltering}.
This is a simple phase gate on $\ket n$ that needs to be performed for $0 \le j \le m-1$, but does not add a relevant factor to the overall gate complexity either. Hence, implementing
$\cU$ and $\cV$ can also be done with complexity $\cO((\log N + \log(1/\eps))^3)$.
For more general results on constructing conditional unitaries efficiently, Cf.~\cite{bera2010efficient}.

\paragraph{Uncomputation via quantum phase estimation.}

The last step is to `uncompute' the register that contains $\ket n$, that is, to perform the map
\begin{align}
\label{eq:uncomputation}
    \ket n  \dshermite{n} \mapsto \ket 0^{\otimes m}    \dshermite{n} \;.
\end{align}
Similar to eigenstate filtering, we can leverage the fast-forwarding property to perform exponentially-precise QPE that would allows us to obtain $\ket n$ on input $\dshermite{n}$. QPE can then be used to uncompute.

\begin{lemma}
\label{lem:uncomputation}
   Let $0 \le n \le N-1$ and $\eps$ the error, where $N > \log(1/\eps)$. Then there is a quantum circuit that performs the map in 
\cref{eq:uncomputation}  within error $\eps$ using $\cO((\log N + \log(1/\eps))^3)$ gates.
\end{lemma}

\begin{proof}
As before, for given $N$ and $\eps$, we set the dimension $M=2^m$ and $N_{\rm high}$ as in~\cref{lem:PlancherelStates}.
Let $V_j:= V(2^{j}(2\pi/M)) e^{\ri 2^{j}(2\pi/M)/2}$ for $0 \le j \le m-1$, where $V(t)$ is the approximation to $U(t)=e^{-\ri \discretehamiltonian t}$ using the three exponentials in \cref{thm:discrete factoring}. Consider the action of the QPE algorithm using $m$ ancilla qubits initially in uniform superposition and using 
controlled-$V_{j}$, on the Hermite state.
Prior to the action of the QFT, the state is transformed as
\begin{align}
    \frac 1 {\sqrt M} \sum_{l_0,\ldots,l_{m-1} \in \{0,1\}}  \ket {l_0 \ldots l_{m-1}}  \dshermite{n} \mapsto \frac 1 {\sqrt M} \sum_{l_0,\ldots,l_{m-1} \in \{0,1\}}  \ket {l_0 \ldots l_{m-1}}   (V_{m-1})^{l_{m-1}} \ldots (V_0)^{l_0}\dshermite{n} \;,
\end{align}
and using \cref{thm:discrete factoring} this is
\begin{align}
     \frac 1 {\sqrt M}  \sum_{l_0,\ldots,l_{m-1} \in \{0,1\}} e^{-\ri (2\pi n/M)(l_0+ \ldots + 2^{m-1}l_{m-1})} \ket {l_0 \ldots l_{m-1}}   \dshermite{n} + \cO(\exp(-\Omega(N))) \;.
\end{align}
We now apply the QFT to complete the QPE. Since the ancillary system is a superposition over basis states $\ket j$ with phases $e^{-\ri 2 \pi n j/M}$,
then the QFT maps it to
     $\ket n   \dshermite{n}$
within error $\cO(\exp(-\Omega(N)))$ that is subdominant under the assumption  $N> \log(1/\eps)$.
QPE applied then the inverse of~\cref{eq:uncomputation}, so the quantum circuit for uncomputation is the inverse of this QPE procedure.

Like in eigenstate filtering,
the gate complexity of this approach 
is given by that of all the $V_j's$. 
This 
is also $\cO(\log^3(M))$, since each $V_j$
 involves three exponentials of complexity $\cO(\log^2(M))$ each. Since $M=\poly(N,1/\eps)$,
 the result follows.

\end{proof}

\subsection{The algorithm}
We are now ready to provide a quantum algorithm that implements the QHT. This essentially follows the steps outlined in~\cref{sec:overview}.

\begin{algorithm}
    \caption{Quantum Hermite Transform} \label{alg:apx-hermite-xform}
    \begin{algorithmic}[1]
    \State Let $N>0$ be the dimension of the QHT and $\epsilon>0$ be the error. Set $M=\cO(N^{9/4}/\eps^{13/4})$ according to~\cref{lem:PlancherelStates} and let $m=\log_2 M$.
        \State For an $m$-qubit state  $\sum_{n=0}^{N-1} \alpha_n \ket{n}$, use~\cref{lem:state prep} to prepare
        $\sum_{n=0}^{N-1} \alpha_n \ket{n}\ket{\phi_n}/\|\ket{\phi_n}\|$.
        \State Use eigenstate filtering and fixed-point amplitude amplification in~\cref{lem:fixed-pointAA} to prepare
         $\sum_{n=0}^{N-1} \alpha_n \ket{n}\ket{\psi_n}$.
        \State Use uncomputation in~\cref{lem:uncomputation} to reset
        the state $\ket{n}$ and prepare 
        $\sum_{n=0}^{N-1}\alpha_n \dshermite{n}$.
    \end{algorithmic}
\end{algorithm}

\begin{proof}[Proof of \cref{thm:quantumhermitetransform}]
The combination of the above steps and results allows us to perform the map 
\begin{align}
    \sum_{n=0}^{\lowenergy-1} \alpha_n \ket n \mapsto 
    \sum_{n=0}^{\lowenergy-1} \alpha_n \dshermite{n} \;,
\end{align}
within additive error $\eps$, and hence imply the QHT.
Before using amplitude amplification, for both the preparation of Plancherel-Rotach states and eigenstate filtering, the size of the circuit is $\cO \left( (\log \lowenergy + \log(1/\eps))^3 \right)$. Combining this with the complexity of amplitude amplification itself, we have an overall complexity of
\begin{align}
    \cO \left( (\log \lowenergy + \log(1/\eps))^3 \times \log(1/\eps)\right) \;.
\end{align}
\end{proof}

In~\cref{thm:quantumhermitetransform}
we assumed $\lowenergy > \log(1/\eps)$. This
implies that all error terms $\cO(\exp(-\Omega(\lowenergy)))$
are subdominant.
As $\lowenergy$ is exponentially large in some problem size, this assumes $\eps$ not to be double-exponentially small. Nevertheless, it is also possible to avoid this assumption and follow the same steps of our proofs, while keeping track of all errors that are exponentially small in $\lowenergy$.

\section{Hermite sampling}
\label{sec:apps}
In this section we discuss an application of the Hermite transform in theoretical computer science. We solve problems in property testing and learning using Hermite sampling as a subroutine. The goal here is to return a sample $v \in \N^n$ with probability $|\wh{f_v}|^2$, where $\wh{f_v}$ is the coefficient of the Hermite polynomial corresponding to $v$ in the representation of $f$. This gives a simple algorithm to solve the Gaussian Goldreich-Levin problem, an analogue of the well-known Goldreich-Levin problem in boolean functions and cryptography, in the Hermite function basis. We also show some property testing results which also hold for boolean functions using QFT as a subroutine, although we do not elaborate on that here. 

Suppose we have a function $f: \R^n \to [-1,1]$ and a discrete set of input points $S \subseteq \R$. We fix $P \in \N$ assume that $f$ is constant over subcubes of $\R^n$ whose endpoints are integer multiples of $2^{-P_1}$. In other words, $f$ acts with $P_1$ bits of precision. Furthermore, we assume the output of $f$ is recorded with $P_2$ bits of precision. Finally, we assume that $f$ is \say{well-conditioned} in the sense that there exists a \emph{distortion} $\distortion > 0$ such that $\norm{f} \geq \distortion^{-1}$. This distortion parameter captures the ``spikiness" of $f$, it is the max value of $f$ (one can imagine the tails are cut off) divided by the average $\ell_2$-norm of $f$ under the Gaussian distribution. The dependence on this parameter comes from the complexity of preparing a state whose amplitudes are proportional to the values of $f$ using phase kickback, conditional rotations, and postselection.

We note that technically we do not need the full power of the Hermite transform in this section. In particular, we do not need to perform the uncomputation step (\cref{lem:uncomputation}) to be able to sample proportional to $|\wh{f_v}|^2$, since the absolute value is agnostic to the relative phase.

In the following section, we give an algorithm that takes an oracle to $f: \R^n \rightarrow \{-1,1\}$ and implements Hermite sampling, for a cleaner exposition. Later in \cref{sec:gen func}, we generalize this to general functions $f: \R^n \rightarrow \R$. Finally, in \cref{sec:prop-testing,sec:learning} we give applications in property testing and learning, respectively.

\subsection{Boolean-output functions}
In this section we restrict our attention to Boolean-output functions -- that is, functions whose outputs are in $\pmone$. Since there is no output precision, we write $P$, as opposed to $P_1$, to be the input precision of $f$, for simplicity.

For a \emph{resolution parameter} $h$ and bound $L > 0$ define 
\[
\calS_h[-L,L] = \left\{kh : k \in \Z, -L \leq kh < L\right\}.
\]
This constitutes a set of points with $\log(1/h)$ bits of precision over which we will discretize the QHO eigenstates. Define $M = \abs{\calS_h[-L,L]} = \lfloor2L/h\rfloor$. This is the number of evaluation points. We will use the state preparation algorithm from \cref{sec:stateprep} to construct approximate hermite states of dimension 
\[
M = \max\{2C^{-1}L \cdot 2^{P}P \cdot (n + \log n + \log(8/\eps)), 40 \sommaconstant \degree \log (2 \degree)\}.
\]

For simplicity, we first give an algorithm for the case where the output of $f$ is $\pmone$.

\begin{algorithm}
    \caption{Approximate Boolean Hermite Sampling} \label{alg:apx-hermite-sampling}
    \begin{algorithmic}[1]
        \State Let $P$ be the number of bits of precision that $f$ acts with.
        \State Choose $M = \max\{2\sommaconstant^{-1}L \cdot 2^{P}P \cdot (n + \log n + \log(8/\eps)), 40 \sommaconstant \degree \log (2 \degree)\}$, let $m = \log_2 M$.
        \State Let $V$ that implements the Hermite transform up to degree $\degree$ and ambient dimension $M$.
        \State Prepare the state $\apxdshostate{0}^{\otimes n} = (V \ket{0^m})^{\otimes n}$ %
        \State Apply the oracle $U_f$ to $\apxdshostate{0}^{\otimes n}$.
        
        \State Apply ${V^\dagger}^{\otimes n}$ to the state $U_f \apxdshostate{0}^{\otimes n}$
        \State Measure in the computational basis and return the result.
    \end{algorithmic}
\end{algorithm}

In the remainder of the section, we will prove the correctness of \cref{alg:apx-hermite-sampling}. For succinctness we write for some $x\in \{0,1\}^n$, $\apxdshostate{x_1} \otimes \ldots \apxdshostate{x_n}$ as $\apxdshostate{x}$ in the rest of this section. We begin by employing a standard result in approximation theory which states that Riemann sums converge exponentially fast for analytic functions such as the hermite functions.

\begin{fact}[\cite{trefethen2014exponentially}]\label{fact:exponential-riemann-convergence}
    Let $f: \C \to \C$ be analytic on a strip $\calS = \{z \in \C: \Re(z) \in [a,b], \abs{\Im(z)} \leq a\}$ and suppose further that $|f(z)| \leq Q$ for $z$ in this strip. Then we have
    \[
    \abs{\frac{1}{M}\sum_{\ell=0}^{M-1} f\left(a + \frac{(b-a) \ell}{M}\right) - \int_a^b f(t) dt} \leq \frac{2\pi Q}{e^{2\pi a M} - 1} \leq 4\pi Q e^{-2\pi a M}.
    \]
\end{fact}

We show that the product of convergent univariate Riemann sums converges to the product of the corresponding integrals.

\begin{proposition} \label{prop:riemann-hybrid}
    Let $\{a_i^{(k)}\}_{i=1}^M$ be finite sequences for $k \in [n]$. For each $k$ choose a $b_k$ such that 
    \[
    \abs{\sum_{i=1}^M a_{i}^{(k)} - b_k} \leq \eps.
    \]
    Furthermore, suppose that there exists an $Q$ such that for all $k$
    \[
    \abs{\sum_{i=1}^M a_{i}^{(k)}} \leq L, \quad \abs{b_k} \leq Q.
    \]
    Then we have
    \[
    \abs{\prod_{k=1}^n \sum_{i=1}^M a_{i}^{(k)} - \prod_{k=1}^n b_k} \leq n Q^{n-1} \eps.
    \]
\end{proposition}
\begin{proof}
    The proof proceeds via a simple hybrid argument. Let 
    \[
    H_0 = \prod_{k=1}^n \sum_{i=1}^M a_{i}^{(k)},\ H_1 = b_1 \cdot \prod_{k=2}^n \sum_{i=1}^M a_{i}^{(k)},\ ...,\ H_n = \prod_{k=1}^n b_k.
    \]
    For $\ell \in [n]$ we bound the difference between the hybrids $H_\ell$ and $H_{\ell-1}$. We have
    \begin{align*}
        \abs{H_\ell - H_{\ell-1}} &= \abs{\prod_{k=1}^\ell b_k \cdot \prod_{k=\ell+1}^n \sum_{i=1}^M a_{i}^{(k)} - \prod_{k=1}^{\ell-1} b_k \cdot \prod_{k=\ell}^n \sum_{i=1}^M a_{i}^{(k)}} \\
        &= \abs{\prod_{k=1}^{\ell-1} b_k \cdot \prod_{k=\ell+1}^n \sum_{i=1}^M a_{i}^{(k)}} \cdot \abs{\sum_{i=1}^M a_{i}^{(\ell)} - b_{\ell}} \\
        &= Q^{\ell - 1} \cdot Q^{n-\ell} \cdot \abs{\sum_{i=1}^M a_{i}^{(\ell)} - b_{\ell}} \leq Q^{n-1} \eps.
    \end{align*}
    Since this holds for any choice of $\ell \in [n]$, we have
    \[
    \abs{\prod_{k=1}^n \sum_{i=1}^M a_{i}^{(k)} - \prod_{k=1}^n b_k} = \abs{H_0 - H_n} \leq \sum_{\ell=1}^n \abs{H_{\ell+1} - H_{\ell}} \leq n Q^{n-1} \eps. \qedhere
    \]
\end{proof}

Next, we bound the probability mass of hermite functions outside the region $L$. In the sections that follow, we denote by $\gaussianpdf$ the Gaussian probability density function $e^{-x^2}/\sqrt{2\pi}$. Recall that $\int_{\R} H_k(x) H_{\ell}(x) \gaussianpdf(x) dx = \delta_{k\ell}$.

\begin{proposition}\label{prop:gaussian-tail-bound}
    Fix $L > 1$ and let $R = [-L,L]$. We have
    \[
    \int_{\R^n \setminus R^n} f(x) H_v(x) \gaussianpdf(x) \leq e^{-nL^2/4}
    \]
\end{proposition}

\begin{proof}
Define $R = [-L,L]^n$. We have
    \begin{align*}
        \abs{\int_{\R^n \setminus R} f(x) H_v(x) \gaussianpdf(x) dx} &= \abs{\int_{\R^n \setminus R} f(x) \gaussianpdf^{1/2}(x) H_v(x) \gaussianpdf^{1/2}(x) dx} \\
        &\leq \sqrt{\int_{\R^n \setminus R} f(x)^2 \gaussianpdf(x)} \sqrt{\int_{\R^n \setminus R} H_v(x)^2 \gaussianpdf(x)} &\text{(Cauchy-Schwarz)} \\
        &\leq \sqrt{\int_{\R^n \setminus R} \gaussianpdf(x)} \sqrt{\int_{\R^n} H_v(x)^2 \gaussianpdf(x)} \\\
        &\leq \left(2\int_{L}^{\infty} \gaussianpdf(x)\right)^{n/2} \leq e^{-nL^2/4}. \qedhere
    \end{align*}
\end{proof}

As our first major step towards proving the correctness of \cref{alg:apx-hermite-sampling}, we show that Riemann sums that are approximately induced by the discretized QHO eigenstates  converges to the correct quantity.

\begin{proposition} \label{prop:riemann-sum-convergence}
    Let $f:\R^n \to [-1,1]$ be constant on subsets of the form $[2^{-P}k,2^{-P}(k+1)]$. For a sufficiently large $M \in \N$, take $h = \sqrt{\frac{2 \pi}{M}}$, $L = \sqrt{\frac{\pi M}{2}}$, and $S_1 = \calS_{h}[-L,L]$. Then for any $v \in \N^n$,
    \[
    \abs{\sum_{x \in S_1^n} M^{-n} f(x) H_v(x) \gaussianpdf(x) - \int_{\R^n} f(x) H_v(x) \gaussianpdf(x)} \leq 2n \cdot (2L\cdot 2^P)^ne^{-2^{-P} \cdot (M/2L)^2} + e^{-n L^2/4}
    \]
\end{proposition}
\begin{proof}
For each $x \in S$, define $S_2(x) = \calS_{h}[x - 2^{-P}, x]$ and $T_2(x) = \{[x, x+h] : x \in S_2(x)\}$. We remark that the $T_2(x)^n$ are cubes over which $f$ is constant, by assumption.
    We denote
    \[
    R = \sum_{x \in S_1^n} M^{-n} f(x) H_v(x) \gaussianpdf(x).
    \]
    This is a Riemann sum that approximates the integral $f(x) H_v(x) \gaussianpdf(x)$ within the region $[-L,L]$. First we bound the error of this approximation and later show using \cref{prop:gaussian-tail-bound} that restricting to this region accrues very little error.
    \begin{align*}
        \abs{R - \int_{[-L,L]^n} f(x) H_v(x) \gaussianpdf(x) dx} &= \abs{\sum_{x \in S_2^n} \left(M^{-n}\sum_{y \in S_2(x)^n} f(y) H_v(y) \gaussianpdf(y) - \int_{T_2(x)} f(y) H_v(y) \gaussianpdf(y) dy\right)} \\
        &= \abs{\sum_{x \in S_2^n} f(x) \left(M^{-n}\sum_{y \in S_2(x)^n}  H_v(y) \gaussianpdf(y) - \int_{T_2(x)} H_v(y) \gaussianpdf(y) dy\right)} \\
        &\leq \sum_{x \in S_2^n} \abs{M^{-n}\sum_{y \in S_2(x)^n}  H_v(y) \gaussianpdf(y) - \int_{T_2(x)} H_v(y) \gaussianpdf(y) dy}
    \end{align*}
    We will bound the inner term of the above sum. First, we observe that the both the sum and the integral can be expressed as products of individual sums and integrals. Thus, applying \cref{prop:riemann-hybrid}, we have
    \begin{align*}
        \abs{M^{-n}\sum_{y \in S_2(x)^n}  H_v(y) \gaussianpdf(y) - \int_{T_2(x)} H_v(y) \gaussianpdf(y) dy} &\leq n \cdot \max_{i\in[n]} \abs{\frac{1}{M}\sum_{z \in S_2(x_i)}  H_{v_i}(z) \gaussianpdf(z) - \int_{T_2(x_i)} H_{v_i}(z) \gaussianpdf(z) dz}
    \end{align*}
    Note that the variable $z$ above is univariate. We can now use \cref{fact:exponential-riemann-convergence} to bound the inner term. Indeed, invoking the application of the trapezoid rule in \cref{eq:hermite-trapezoid-convergence}, we have for any $i \in [n]$
    \[
    \abs{\frac{1}{M}\sum_{z \in S_2(x_i)}  H_{v_i}(z) \gaussianpdf(z) - \int_{T_2(x_i)} H_{v_i}(z) \gaussianpdf(z) dz} \leq e^{-2^{-P} \cdot (M/2L)^2}
    \]
    whenever $M > 3\cdot 2^{P}$.
    The sum over $S_2^n$ results in an additional multiplicative factor of $(2L2^P)^n$. As such, we have
    \[
    \abs{R - \int_{[-L,L]^n} f(x) H_v(x) \gaussianpdf(x) dx} \leq 2n \cdot (2L\cdot 2^P)^ne^{-2^{-P} \cdot (M/2L)^2}
    \]

    We now conclude the proof by showing that the exclusion of the region $\R^n \setminus[-L,L]^n$ results in a very small additive error. Indeed, applying \cref{prop:gaussian-tail-bound} we have
    \begin{align*}
         \abs{R - \int_{\R^n}f(x) H_v(x) \gaussianpdf(x)} &\leq \abs{R - \int_{[-L,L]^n} f(x) H_v(x) \gaussianpdf(x) dx}  +\abs{\int_{\R^n \setminus[-L,L]^n} f(x) H_v(x) \gaussianpdf(x) dx} \\
         &\leq 2n \cdot (2L\cdot 2^P)^ne^{-2^{-P} \cdot (M/2L)^2}  +\abs{\int_{\R^n \setminus[-L,L]^n} f(x) H_v(x) \gaussianpdf(x) dx} \\
         &\leq 2n \cdot (2L\cdot 2^P)^ne^{-2^{-P} \cdot (M/2L)^2} + e^{-n L^2/4}.  \qedhere
    \end{align*}
\end{proof}

Now we show that applying the oracle to the discrete, unnormalized ground state $\dshostate{0^n}$ and measuring in the $\dshostate{v}$ basis corresponds to a Riemann sum that converges to $\wh{f}(v)$.

\begin{lemma} \label{lem:dsho-riemann-convergence}
For each $v \in [D]^n$, let $\dshostate{v} = \otimes_{i=1}^n \dshostate{v_i}$, where each $\dshostate{v_i}$ is formed with
resolution parameter $h = M^{-1}$ and limit $L$. We have, for every $v \in [D]^n$
\[
\abs{\bradshostate{v} U_f \dshostate{0^n} - \wh{f}(v)} \leq \eps/4.
\]
\end{lemma}

\begin{proof}
    We can write
    \[
    \bradshostate{v} U_f \dshostate{0^n} = M^{-n}\sum_{x \in S_1^n} f(x) H_v(x) \gaussianpdf(x).
    \]
    By \cref{prop:riemann-sum-convergence}, we have the bound
    \[
    \abs{\abs{\bradshostate{v} U_f \dshostate{0^n} - \wh{f}(v)}} \leq 2n \cdot (2L\cdot 2^P)^ne^{-2^{-P} \cdot (M/2L)^2} + e^{-n L^2/4}.
    \]
    For our choice of $M$ we have that the right hand side of this is bounded by $\eps/8 + \eps/8 = \eps/4$.
\end{proof}

Since the (orthonormal) approximations $\apxdshostate{}$ of the discrete QHO eigenstates are very close in trace distance to the $\dshostate{}$ above, we can show that they satisfy a similar property.

\begin{corollary} \label{cor:apxdsho-riemann-convergence}
    For each $v \in [D]^n$, let $\apxdshostate{v} = \otimes_{i=1}^n \apxdshostate{v_i}$, where each $\apxdshostate{v_i}$ is formed with resolution parameter
    $h = M^{-1}$. We have, for every $v \in [D]^n$
    \[
    \abs{\braapxdshostate{v} U_f \apxdshostate{0^n} - \wh{f}(v)} \leq \eps/2.
    \]
\end{corollary}

\begin{proof}
    Now we argue that replacing $\dshostate{0^n}$ with $\apxdshostate{0^n}$ doesn't result in too much error. Recall that we assume the existence of a constant $\sommaconstant$ such that
    \[
    \tracedistance{\dshostate{0^n}, \apxdshostate{0^n}}\leq n\tracedistance{\dshostate{0}, \apxdshostate{0}} \leq n\exp(-\sommaconstant M).
    \]
    Thus we have
    \begin{align*}
        \abs{\braapxdshostate{v} U_f \apxdshostate{0^n} - \wh{f}(v)} &\leq \abs{\braapxdshostate{v} U_f \apxdshostate{0^n} - \braapxdshostate{v} U_f \dshostate{0^n}} + \abs{\braapxdshostate{v} U_f \dshostate{0^n} - \wh{f}(v)} \\
        &= \abs{\braapxdshostate{v} U_f \left(\apxdshostate{0^n} - \dshostate{0^n}\right)} + \abs{\braapxdshostate{v} U_f \dshostate{0^n} - \wh{f}(v)} \\
        &\leq n\exp(- \sommaconstant M) + \abs{\braapxdshostate{v} U_f \dshostate{0^n} - \wh{f}(v)} \\
        &\leq \eps/12 + \abs{\braapxdshostate{v} U_f \dshostate{0^n} - \bradshostate{v} U_f \dshostate{0^n}} + \abs{\bradshostate{v} U_f \dshostate{0^n} - \wh{f}(v)} \\
        &\leq \abs{\braapxdshostate{v} U_f \dshostate{0^n} - \bradshostate{v} U_f \dshostate{0^n}} + \eps/3  \\
        &= \abs{\left(\braapxdshostate{v} - \bradshostate{v}\right)U_f \dshostate{0^n}} + \eps/3 \\
        &\leq n\exp(-\sommaconstant M) + \eps/3 \leq \eps/2. \qedhere
    \end{align*}
\end{proof}

Now, we are ready to prove the correctness of \cref{alg:apx-hermite-sampling}.

\begin{restatable}[label={thm:apx-hermite}]{theorem}{apxHermite}
    For each $v \in [D]^n,$ \Cref{alg:apx-hermite-sampling} outputs $v$ with a probability $p_v$ such that $\abs{p_v - \wh{f}(v)^2} \leq \eps$. Furthermore, it runs in time $\cO(n \polylog(n,D,1/\eps))$.
\end{restatable}
\begin{proof}
    The correctness of the algorithm follows from \cref{cor:apxdsho-riemann-convergence}. Indeed, for every $v \in [D]^n$ we have
    \begin{align*}
    \abs{p_v - \wh{f}(v)^2} &= \abs{\abs{\braapxdshostate{S}U_f\apxdshostate{0^n}}^2 - \wh{f}(v)^2} \\
    &= \abs{\abs{\braapxdshostate{S}U_f\apxdshostate{0^n}} + \wh{f}(v)} \cdot \abs{\abs{\braapxdshostate{S}U_f\apxdshostate{0^n}}^2 - \wh{f}(v)^2} \\
    & \leq 2 \abs{\abs{\braapxdshostate{S}U_f\apxdshostate{0^n}}^2 - \wh{f}(v)^2} \leq \eps. & \text{\cref{cor:apxdsho-riemann-convergence}}
    \end{align*}
    
    We now analyze the time complexity. We will prepare the QHO ground states $\apxdshostate{0^n}$ with precision $M$ and approximation error $\eps/3n$.
    Using \cref{thm:quantumhermitetransform} this can be done in time $\cO((\log n + \log D + \log(1/\eps))^3 \log n \log(1/\eps))$. Then, after applying the oracle, we apply the unitary that maps $\apxdshostate{k}$ to $\ket{k}$ for $k$ in the low-energy subspace for each qubit, with approximation error $\eps/6n$. Since these states lie in a Hilbert space of dimension $M$ this also takes $\cO((\log n + \log D + \log(1/\eps))^3 \log n \log(1/\eps))$ time for each $i \in [n]$. Overall, the runtime is $\cO(n \polylog(n, D, 1/\eps))$, and the additional error from implementing the states $\apxdshostate{0^n}$ and inverse Hermite transform unitary is $\eps/2$, for an overall error of $\eps$.
\end{proof}

Furthermore, with a slight overhead in runtime, we can guarantee that our Hermite sampling procedure approximates the true distribution in total variation distance.
\begin{corollary}
\label{cor:apx-hermite-sampling}
    Let $f$ be $\upsilon$-concentrated on Hermite coefficients with univariate degree at most $D$.
    With a runtime overhead of $\cO(n \log D)$,
    \cref{alg:apx-hermite-sampling} returns a distribution $\calD = \{p_v\}_{v \in [D]^n}$ such that the total variation distribution between $\calD$ and the true Hermite distribution is at most $\eps + \upsilon$.
\end{corollary}
\begin{proof}
    This follows easily from \cref{thm:apx-hermite} by choosing an error $\eps^\prime = \eps \cdot (D+1)^{-n}$ and applying the union bound over all elements in $\{0,1,...,D\}^n$. The overhead in incurred runtime is $n \log D$.
\end{proof}

\subsection{General functions}
\label{sec:gen func}

In this section, we describe how to generalize our algorithm to functions whose output is now in $[-1,1]$. We remark that the algorithm is effectively the same as in the Boolean case, with the addition of a quantum multiplication step that multiplies the amplitude of each $\ket{x}$ by $f(x)$. This involves controlled rotations and postselection. These steps require a more careful analysis and increase the query and time complexity by a factor depending on the \say{well-behavedness} of $f$. This property is what we call distortion, as defined below.

\begin{definition}[Distortion]
    \[\kappa(f) = \frac{\norm{f(x) \sqrt{\gaussianpdf(x)}}_\infty}{ \norm{f(x)\sqrt{\gaussianpdf(x)}}_2}\]

We denote this by $\kappa$ when $f$ is clear from context. Note that $\kappa = 1$ for all functions $f:\R^n \rightarrow \pmone$.
\end{definition}\textbf{}

\begin{algorithm}
    \caption{Approximate Hermite Sampling} \label{alg:apx-hermite-non-boolean}
    \begin{algorithmic}[1]
        \State Let $P_1$ and $P_2$ be the number of input and output bits of precision, respectively, for $f$. Furthermore, let $\kappa$ be the distortion parameter for $f$.
        \State Choose $M = \max\{2^{P_1} \cdot (n + P_1 + \log(24\kappa /\eps)) \log(1/\eps), 40\sommaconstant \degree \log \degree\}$, let $m = \log_2 M$.
        \State Let $V$ that implements the Hermite transform up to degree $\degree$ and ambient dimension $M$.
        \Repeat
        \State Prepare the state $\apxdshostate{0}^{\otimes n}\ket{0^{P_2}}\ket{0}$ %
        \State Apply the oracle $U_f$ controlled on register $1$ to register $2$.
        \State Define the controlled rotation $R \ket{f(x)} \ket{0} \to \sqrt{1 - f(x)^2}\ket{f(x)} \ket{0} + f(x)\ket{f(x)}\ket{1}$.
        \State Apply $R$ controlled by register $2$ onto register $3$.
        \State Measure register $3$ in the computational basis.
        \Until{$\ket{1}$ is measured} \label{step:postselect}
        \State Uncompute by applying $U_f$ controlled on register $1$ to register $2$. \label{step:uncompute}
        \State Measure register $2$ in the computational basis and let $\ket{\wt\psi}$ be the state of register $1$. \label{step:reduce}
        \State Prepare a unitary $V$ that takes $\apxdshostate{k}$ to $\ket{k}$ for all $k \in [D]$. \\
        \Return The result of applying $V^{\otimes n}$ to $\ket{\psi}$, and measuring in the computational basis.
    \end{algorithmic}
\end{algorithm}

We first show that the probability distribution generated by \cref{alg:apx-hermite-non-boolean} is close to the true Hermite distribution pointwise. Recall that, when the output of $f$ is non-Boolean, the Hermite distribution $q_v$ is defined to be $\wh{f}(v)^2/\norm{f}^2$ for all $v \in \N^n$.

\begin{lemma}\label{lem:non-boolean-correctness}
    Define $q_v$ to be the distribution induced on $\N^n$ by the Hermite coefficients of $f$ and let $p_v$ be the probability of sampling $v \in [D]^n$ in \cref{alg:apx-hermite-non-boolean}.
    \[
    \abs{p_v - q_v} \leq \eps.
    \]
\end{lemma}

\begin{proof}
    Let $\ket{\overline{\psi}}$ represents the state of the algorithm after the postselection step, assuming the initial state was instead $\dshostate{0^n}\otimes \ket{0^{P_2}} \ket{0}$. We have that
    \[
    \ket{\overline{\psi}} = \frac{\sum_{x \in S_1^n} M^{-n/2} f(x) \gaussianpdf^{1/2}(x) \ket{x}}{\sqrt{\sum_{x \in S_1^n} M^{-n} f(x)^2 \gaussianpdf(x) }}.
    \]
    The denominator (call it $\sqrt{P}$) is the normalization factor, and we show that it is very close to $\norm{f}$. Applying \cref{prop:riemann-sum-convergence}, we have
    \begin{align*}
    \abs{\sum_{x \in S_1^n} M^{-n} f(x)^2 \gaussianpdf(x) - \norm{f}^2} &= \abs{\sum_{x \in S_1^n} M^{-n} f(x)^2 \gaussianpdf(x) - \int_{\R^n} f(x)^2 dx} \leq \frac{\eps \cdot \norm{f}}{12}
    \end{align*} %
    Thus, we can bound
    \[
    \abs{\sqrt{\sum_{x \in S_1^n} M^{-n} f(x)^2 \gaussianpdf(x)} - \norm{f}} = \frac{\abs{\sum_{x \in S_1^n} M^{-n} f(x)^2 \gaussianpdf(x) - \int_{\R^n} f(x)^2 dx}}{\sqrt{\sum_{x \in S_1^n} M^{-n} f(x)^2 \gaussianpdf(x)} + \norm{f}} \leq \frac{\eps \cdot \norm{f}}{12 \cdot \norm{f}} \leq \frac{\eps}{12}.
    \]
    We can write:
    \[\bradshostate{v}\overline{\psi}\rangle = \frac{1}{\sqrt{P}} \sum_{x \in S_1^n} M^{-n} f(x) H_v(x) \gaussianpdf(x).
    \]
    We show using \cref{prop:riemann-sum-convergence} that this is a Riemann sum which well-approximates $\frac{\wh{f}(v)}{\norm{f}}$.
    \begin{align*}
        \abs{\bradshostate{v}\overline{\psi}\rangle - \frac{\wh{f}(v)}{\norm{f}}} &=\abs{\frac{1}{\sqrt{P}} \sum_{x \in S_1^n} M^{-n} f(x) H_v(x) \gaussianpdf(x) - \frac{\wh{f}(v)}{\norm{f}}} \\
        &\leq \frac{1}{\sqrt{P}}\abs{\sum_{x \in S_1^n} M^{-n} f(x) H_v(x) - \wh{f}(v)} - \abs{\wh{f}(v)}\abs{\frac{1}{\sqrt{P}} - \frac{1}{\norm{f}}} \\
        &\leq \frac{\eps}{12} + \abs{\frac{1}{\sqrt{P}} - \frac{1}{\norm{f}}} \\
        &\leq \frac{\eps}{12} + \frac{\eps}{12} \leq \frac{\eps}{6}.
    \end{align*}
    Now, let $\ket{\wt\psi}$ be the actual state of the algorithm after \cref{step:reduce}. We compute $\norm{\ket{\overline\psi} - \ket{\wt\psi}}$. Recall that $\norm{\dshostate{0^n} - \apxdshostate{0^n}} \leq n\sqrt{2}\exp(- \sommaconstant M)$. It is a standard fact that postselection which occurs successfully with probability $p$ blows up this distance by a factor of $p^{-1}$. The postselection in \cref{step:postselect} of the algorithm succeeds with probability at least $\frac{1}{2\norm{f}}$, so we have that  %
    \[
    \norm{\ket{\overline\psi} - \ket{\wt\psi}} \leq 2\sqrt{2} n \norm{f} \cdot \exp(- \sommaconstant M) \leq \eps/12
    \]
    
    We replace $\dshostate{v}$ with $\apxdshostate{v}$, which gives us
    \begin{align*}
    \abs{\braapxdshostate{v}\overline{\psi}\rangle - \frac{\wh{f}(v)}{\norm{f}}} &\leq \abs{\braapxdshostate{v}\overline{\psi}\rangle - \bradshostate{v}\overline{\psi}\rangle} + \abs{\bradshostate{v}\overline{\psi}\rangle - \frac{\wh{f}(v)}{\norm{f}}} \leq \eps/4
\end{align*}
Finally, we replace $\dshostate{0^n}$ with $\apxdshostate{0^n}$, which swaps and bound:
\begin{align*}
\abs{\braapxdshostate{v}\wt{\psi}\rangle - \frac{\wh{f}(v)}{\norm{f}}} \leq \abs{\braapxdshostate{v}\wt{\psi}\rangle - \braapxdshostate{v}\overline{\psi}\rangle} + \abs{\braapxdshostate{v}\overline{\psi}\rangle - \frac{\wh{f}(v)}{\norm{f}}} \leq \frac{\eps}{3}
\end{align*}
Recalling that the probability of sampling $v \in [D]^n$ in \cref{alg:apx-hermite-non-boolean} is exactly $\abs{\braapxdshostate{v}\wt{\psi}\rangle}^2$ and that $q_v = \frac{\wh{f}(v)^2}{\norm{f}^2}$, we conclude by bounding
\begin{align*}
    \abs{p_v - q_v} &= \abs{\abs{\braapxdshostate{v}\wt{\psi}\rangle}^2 - \frac{\wh{f}(v)^2}{\norm{f}^2}} \\
    &= \abs{\abs{\braapxdshostate{v}\wt{\psi}\rangle} + \frac{\wh{f}(v)}{\norm{f}}} \cdot \abs{\abs{\braapxdshostate{v}\wt{\psi}\rangle} - \frac{\wh{f}(v)}{\norm{f}}} \\
    &\leq 3\abs{\abs{\braapxdshostate{v}\wt{\psi}\rangle} - \frac{\wh{f}(v)}{\norm{f}}} \leq \eps. \qedhere
\end{align*}
    
\end{proof}

The proof of the above lemma also gives us the following stronger result:
\begin{lemma}
    Let $\ket{\wt{\psi}}$ be the state of \cref{alg:apx-hermite-non-boolean} after \cref{step:uncompute}. Then
    \[
    \abs{\braapxdshostate{v}\!\psi\rangle - \frac{\wh{f}(v)}{\norm{f}}} \leq \frac{\eps}{6}.
    \]
\end{lemma}

We are now ready to prove the correctness of \cref{alg:apx-hermite-non-boolean}.

\begin{theorem}
    Let $p_v$ be the distribution induced by the Hermite spectrum of a function $f$ with distortion $\kappa$.
    \Cref{alg:apx-hermite-non-boolean} succeeds in sampling $v \in [D^n]$ with probability $p_v$ such that $\abs{p_v - q_v} \leq \eps$. Moreover it runs in time $\cO(\kappa n \polylog(n,D,1/\eps))$ and makes $\cO(\kappa)$ queries in expectation to $f$.
\end{theorem}

\begin{proof}
    By \cref{lem:non-boolean-correctness}, we recover the bound $\abs{p_v - q_v} \leq \eps$. It remains to analyze the runtime and query complexity of the algorithm. The postselection step in \cref{step:postselect} succeeds with probability at least $\frac{\norm{f}}{2} \geq \frac{1}{2\kappa}$. Each attempts makes a single query to $f$ and the only other time a query is made to $f$ is in the uncomputation in \cref{step:uncompute}. Thus the query complexity is $\cO(\kappa)$.

    Now we analyze the computational complexity. Preparing each initial state using \cref{thm:quantumhermitetransform} this can be done in time $\cO((\log n + \log(1/\eps))^3 \log n \log(1/\eps))$. So, overall, this step takes time $\cO(\kappa n \polylog(n,1/\eps))$. The remainder of the complexity is dominated by the preparation of $V$, which takes circuit complexity at most $\cO(n\log M) = O(n\polylog(n,D,1/\eps))$. Thus the overall runtime is $\cO(\kappa n \polylog(n,D,1/\eps))$.
\end{proof}

Again, we can boost the guarantee to total variation distance with an additional factor of $\cO(n \log D)$.

\begin{corollary}
\label{cor:apx-hermite-sampling-non-boolean}
    Let $f$ be $\upsilon$-concentrated on Hermite coefficients with univariate degree at most $D$.
    With a runtime overhead of $\cO(n \log D)$,
    \cref{alg:apx-hermite-non-boolean} returns a distribution $\calD = \{p_v\}_{v \in [D]^n}$ such that the total variation distribution between $\calD$ and the true Hermite distribution is at most $\eps + \upsilon$.
\end{corollary}
The proof is identical to that of \cref{cor:apx-hermite-sampling}.

\paragraph{Applications.}

Inspired by the work of Klivans et al.~\cite{KlivansOS08}, we give applications of our algorithm in the form of learning and property testing algorithms. In PAC learning, we are given samples of the form $f(x)$ where $f$ belongs to some known concept class $\cal C$ and $x$ is drawn from a known distribution $\cal D$. For us $\cal D$ will always be the Gaussian distribution. The goal is to learn some $g$ which agrees with $f$ with high probability on a random input from $\cal D$. Sometimes learning tasks are also studied assuming query access to $f$ instead of samples. In property testing, we have an easier task. We are given query access to $f$, and we simply want to test if $f$ is in $\cal C$ or disagrees with all functions in $\cal C$ with high probability. For more background, we refer readers to surveys on learning theory~\cite{MR1331838,AdW17} and property testing~\cite{CIT-114,MdW16}. We go over some preliminaries needed for this section below.

We will use a seminal result in the theory of Sobolev spaces, specialized for the Gaussian distribution. This is known as the Gaussian Poincar\'e inequality.

\begin{lemma}[Gaussian Poincar\'e inequality~\cite{poincare90,BLM13}]
\label{lem:poincare}
    For $\mathbf{x} \sim \calN(0,I)$, and a differentiable function $f:\R^n \rightarrow \R$ we have that
    \[\Var(f(\mathbf{x})) \leq \E \norm{\nabla f(\mathbf{x})}_2^2\]
\end{lemma}

The quantity on the right is known as the Hardy-Krause variation of $f$ ($\sigma_{HK}(f)$) in the literature on Quasi Monte-Carlo methods~\cite{fishmanMonteCarlo}. This inequality tells us that for the Gaussian distribution vanilla MC suffices to get a good sample complexity. We will also need a result about the approximation of smooth functions by Hermite polynomials. We state it below.

\begin{lemma}[Polynomial approximation of smooth functions (\cite{DKKTZ24,KTZ19})]
\label{lem:hermite-concentration}
    Denote by $P_{>m}f$ the Hermite expansion of $f$ truncated to only include terms with total degree greater than $m$. Then for all almost-everywhere differentiable functions, we have that 

    \[
    \Ex_{x\sim \mathcal{N}}[P_{>m}f(x)^2] \leq O(\frac{1}{m}) \Ex_{x\sim \mathcal{N}} \|\nabla f(x)\|_2^2
    \]
\end{lemma}

For this section we will also need some lemmas about the classical complexity of computing for Hermite coefficients of restricted functions, which we prove below.

\begin{definition}
    For a function $f:\R^n \rightarrow \R$, set $J\subseteq [n]$ and $z\in \R^{\Bar{J}}$. Then define $f_{J | z}(y) : \R^J \rightarrow \R$ to be the function where the scalars in $\Bar{J}$ is restricted to $z$.
\end{definition}

\begin{definition}
    For a differentiable function $f:\R^n \rightarrow \R$, define $\gamma_f := \E_{x\sim \calN} \norm{\nabla f}^2_2$. We will denote this by $\gamma$ when $f$ is clear from context.
\end{definition}

\begin{lemma}
\label{lem:prop321}
    Let $f: \R^n \rightarrow \R$ and $S\in \N^J$, then we say $J\subseteq [n]$ is the set of indices where $S_i > 0$. For $z\in \R^{\Bar{J}}$, write $F_{S} {f}(z) = \wh{f_{J | z}}(S)$. Then we have that
    \[
    \wh{F_S f}(T) = \wh{f}(S \cup T)
    \]
\end{lemma}

\begin{proof}
For every $U\in \mathbb{N}^n$ we write $U = S\cup T$ where $S\in \N^J$and $T\in \N^{\Bar{J}}$. We also use the notation $h_U = h_Sh_T$ where $h_U$ is the Hermite polynomial with coefficient $U$ and $h_S$, $h_T$ are defined so that they only contain variables in $S$ and $T$ respectively. Then we have that

    \[f(x) = \sum_{U\in \N^n} \wh{f}(U) h_U = \sum_{S\in \N^J, T\in \N^{\Bar{J}}} \wh{f}(S\cup T) h_Sh_T = \sum_{S\in \N^J} \left(\sum_{T\in \N^{\Bar{J}}} \wh{f}(S\cup T) h_T\right) h_S
    \]

    If the variables in $T$ are fixed to $z$ then as a function of the variables in $S$ the coefficient on the polynomial $h_S$ is $\sum_{T\in \N^{\Bar{J}}} \wh{f}(S\cup T) h_T(z)$. Thus, we have that $\wh{F_S f}(T) = \wh{f}(S \cup T)$ as desired.
\end{proof}

We can use \cref{lem:prop321} to derive a quantity which we can estimate efficiently.

\begin{definition}
    Let
    \[
    \mathbf{W}^{S}(f) = \sum_{T \in \N^{\Bar{J}}} \wh{f}(S\cup T)^2
    \]
\end{definition}
\begin{corollary}[Corollary of \cref{lem:prop321}]
    \[
    \mathbf{W}^{S}(f) = \int_{\R^{\Bar{J}}} \wh{f_{J | z}}(S)^2 \gaussianpdf(z) dz
    \]
\end{corollary}
\begin{proof}
    \begin{align*}
    \int_{\R^{\Bar{J}}} \wh{f_{J | z}}(S)^2 \gaussianpdf(z) dz &= \int_{\R^{\Bar{J}}} F_{S}f(z)^2 \gaussianpdf(z) dz\\
            &= \sum_{T\in \N^{\Bar{J}}} \wh{F_S}(T)^2 \tag*{(Parseval's)}\\
            &= \sum_{T\in \N^{\Bar{J}}} f(S\cup T)^2
    \end{align*}
\end{proof}

Using our corollary, we can estimate $\mathbf{W}^{S}(f)$ efficiently.

\begin{lemma}
\label{lem:weight-estm}
    For any $S \in \N^n$ an algorithm with query access to a differentiable function $f : \R^n \rightarrow \R$ can compute an estimate of $\mathbf{W}^{S}(f)$ that is accurate to within $\pm \varepsilon$ with probability $1-\delta$ using $\poly(\gamma/\varepsilon) \log(1/\delta)$ queries.
\end{lemma}
\begin{proof}
    We rewrite $\mathbf{W}^S(f)$ in a way which can be estimated using Monte-Carlo integration.

    \begin{align*}
        \mathbf{W}^{S}(f) &= \int_{\R^{\Bar{J}}} \wh{f_{J | z}}(S)^2 \gaussianpdf(z) dz = \int_{\R^{\Bar{J}}}\left(\int_{\R^J} f(y,z) h_S(y) \gaussianpdf(y) dy \right)^2 \gaussianpdf(z) dz\\
        &= \int_{\R^{\Bar{J}}}\left(\int_{\R^J} f(y,z)f(y',z) h_S(y) h_S(y')\gaussianpdf(y)\gaussianpdf(y')dy dy'\right) \gaussianpdf(z) dz\\
        &= \E_{z\sim \calN^{\Bar{J}}} \E_{y,y' \sim \calN^{J}} [f(y,z)f(y',z)h_S(y)h_S(y')]
    \end{align*}

This quantity can be estimated by a Monte-Carlo method, and due to \cref{lem:poincare} we know that $\Var(f(x)) \leq \gamma^2$, so we can get an estimate of this quantity which is accurate up to $\pm \varepsilon$ with probability $1-\delta$ using $\frac{\gamma^2}{\varepsilon^2} log(1/\delta)$ samples.
\end{proof}

\subsection{Property testing}
\label{sec:prop-testing}

In this section, we give quantum algorithms for property testing using our Approximate Hermite Sampling subroutine. In our first two examples, we give provable quantum advantage. Throughout this section, our notion of distance will be Gaussian inner product.

\begin{definition}[Distance]
    We say that $f, g:\R^n \rightarrow \R$ are $\eps$-close if we have that 

    \[\int_{\R^n} f(x) g(x) \gaussianpdf(x) \geq 1 - \eps\]

    Conversely, we say that $f,g$ are $\eps$-far if the inequality does not hold.
\end{definition}

\subsubsection{Provable quantum advantage for sign functions}

The problem we consider here is to test if a given function $f$ is close to a product of $k$ sign function or far from every product of $k$ sign functions.

\begin{definition}[Product sign functions]
    \[\chi_S(x) = \Pi_{i\in S} \sgn(x_i)\]
\end{definition}
\begin{definition}[Testing $k$-product sign functions]
    Given black-box access to $f: \R^n \rightarrow \R$ and given $\eps_2 > \eps_1 > 0$ determine whether
    \begin{itemize}
        \item $f(x)$ is $\eps_1$-close to  $\chi_S(x)$ for some $S\subseteq [n]$ such that $|S| = k$, or
        \item $f(x)$ is $\eps_2$-far from all $k$-product sign functions.
    \end{itemize}
\end{definition}

Let $\eps = \eps_2 - \eps_1$.

\begin{lemma}
    Product sign functions can be tolerantly tested with $O(1/\eps^2)$ quantum queries.
\end{lemma}
\begin{proof}
    Note that if $f(x) = \chi_S(x)$ for some $S\subseteq [n]$ then the only non-zero Hermite coefficients of $f$ are the ones which non-zero on indices in $S$.
    We can use Approximate Hermite Sampling (\cref{alg:apx-hermite-sampling}) setting the TV distance to be less than $\eps/2$ (\cref{cor:apx-hermite-sampling}) and observe a sample $T$, then let $S$ be its support.
    Set $S^*(i) = *$ whenever $i\in S$ and 0 otherwise. Then we can estimate $\mathbf{W}^{S^*}(f)$ upto $\pm \eps/4$ using $O(1/\eps^2)$ queries by \cref{lem:weight-estm} and accept if it is at least $1 - \frac{\eps_1+\eps_2}{2}$.
\end{proof}

We prove a classical lower bound of $\Omega(k)$ for this problem. We use the framework of Blais et al.~\cite{BlaisBM11}, and use a reduction in the communication setting from Disjointness, for which we know tight lower bounds.

\begin{theorem}
\label{thm:sgn-lb}
    Any classical algorithm for testing $k$-product sign functions requires $\Omega(k)$ queries.
\end{theorem}

\begin{proof}

It proceeds by starting from a hard distribution for Disjointness where we sample two sets $A, B$ such that $|A| = |B| = \frac{k}{2}+1$ and with probability 1/2 $A \cap B = \phi$ and with probability 1/2 we have $|A \cap B| = 1$.

For our reduction, given such sets $A, B$ Alice produces $\chi_A = \Pi_{i\in A} \sgn(x_i)$ and Bob produces $\chi_B = \Pi_{i\in B} \sgn(x_i)$ and they want to test $\chi_{AB} = \chi_A \cdot \chi_B$.

We have

\[
\deg(\chi_{AB}) =
\begin{cases}
k+2, & \text{if } A\cap B = \phi\\
k, & \text{if } |A\cap B| = 1
\end{cases}
\]

{\noindent because if $i \in A \cap B$ then $\chi_A \cdot \chi_B$ does not depend on $x_i$, since $\sgn(x_i)^2 = 1$.} To prove our property testing lower bound, it sufffices to prove the following claim.

\begin{lemma}
    For any set $S: |S| = k+2$ and $T: |T| = k$, the functions $\chi_S$ and $\chi_T$ are $\Omega(1)$-far  wrt to the Gaussian distribution.
\end{lemma}
\begin{proof}
    \begin{align*}
    d(\chi_S, \chi_T) &= 1 - \int_{\R^n} \chi_S(x) \chi_T(x) \gaussianpdf(x) dx\\
        &= 1 - \int_{\R^n} \chi_{S\Delta T}(x) \gaussianpdf(x) dx\\
        &= 1
    \end{align*}

    Here we used that $S\Delta T$ is non-empty since $S$ is strictly larger than $T$.
\end{proof}

Therefore, deciding whether the input is close to a product sign function on $k$ variables or not allows us to decide the output of the Disjointness function. This must take $\Omega(k)$ queries.
\end{proof}

Using the lower bound for testing sign functions, we can infer a $\Omega(n)$ lower bound on the classical query complexity of Gaussian Goldreich-Levin. We recall that it is known that one can solve the Goldreich-Levin problem with constant success probability with $O(n/\tau^2)$ classical queries \cite{Goldreich_2001}. The factor $1/\tau^2$ is optimal since Parseval's only gaurantees that the list $L$ of $\tau$-correlated functions is of size $O(1/\tau^2)$. In the quantum case, we can do much better and get $O(1/\tau^2)$ queries.
Here, we ask a similar question but for functions over the domain $\mathbb{R}^n$.

\begin{definition}[Gaussian Goldreich-Levin]
    Given query access to a function $f: \mathbb{R}^n \rightarrow \mathbb{R}$, return a list $L$ of indices of Hermite polynomials such that
    \begin{itemize}
        \item $|\wh{f}(S)| \geq \tau \implies S\in L$
        \item $S \in L \impliedby |\wh{f}(S)| \geq \tau/2$
    \end{itemize}
\end{definition} To derive this we will need a observation about the Hermite expansion of sign functions.

\begin{lemma}
\label{lem:sgn-hermite}
     For $f(x) = \Pi_{i\in T} \sgn(x_i)$ we have that for any $S\in \N^n$, $\wh{f}(S) \leq 1/\exp(k)$ where $k$ is the max degree of a variable in $S$ from $T$.
\end{lemma}
\begin{proof}
Let $f(x) = \sgn(x)$, and let $\wh{f}(k)$ be the coefficient of $f$ for the degree $k$ univariate Hermite polynomial. Then we first show that $\wh{f}(k) \leq 1/\exp(k)$. Note that for even $k$, this coefficient is always 0 since $h_k(x)$ is an even function. For odd $k$, $h_k(x)$ is an odd function and thus we have the following.
    \begin{align*}
        \wh{f}(k) = \int_{\R} \sgn(x) h_k(x) \gaussianpdf(x) dx &= \int_{0}^{\infty} h_k(x) \gaussianpdf(x) dx - \int_{-\infty}^0 h_k(x) \gaussianpdf(x) dx \\
        &= 2\int_{0}^{\infty} h_k(x) \gaussianpdf(x) dx\\
        &\leq 1/\exp(k)
    \end{align*}

In general, if $f(x) = \Pi_{i\in T} \sgn(x_i)$ then for any $S\in \N^n$ with non-zero number in any index outside of $T$, the coefficient $\wh{f}(S) = 0$. For any $S$ supported entirely on $T$, we can factor the integral for computing the Hermite coefficient based on each variable in $T$ and infer a lower bound from the univariate case.
\end{proof}

\begin{corollary}
    Any classical algorithm for Gaussian Goldreich-Levin requires $\Omega_{\eps}(n)$ queries.
\end{corollary}
\begin{proof}
    Using \cref{lem:sgn-hermite}, we can see that any algorithm for Gaussian Goldreich-Levin with parameter $\tau =\frac{1}{\log (1/\eps)}$ can be used to test if a function $f$ is $\eps$-close to a product of $k$ sign functions. But we know that for $k = \Omega(n)$ this problem requires $\Omega(n)$ queries classically.
\end{proof}

We also give a classical algorithm that has a matching  linear dependence on $n$ for the number of queries, but we defer this to \cref{sec:ggl}. Using the Quantum Hermite Transform, we can get rid of the dependence on $n$ entirely which gives quantum query advantage for this problem. Bel

\subsubsection{Tolerant low-degree testing}
We are now concerned with testing if a function $f:\R^n \to \pmone$ is well-approximated by a low-degree Hermite expansion. 
\begin{definition}
We say that $f$ is $(\eps,d)$-low-degree if
\[
\sum_{\substack{S \in \N^n \\ |S| \leq d}} |\wh{f}(S)|^2 \geq 1 - \eps.
\]
\end{definition}

\begin{definition}[Testing degree-$d$]
    Given black-box access to $f: \R^n \rightarrow \pmone$ and given $\eps_2 > \eps_1 > 0$ determine whether
    \begin{itemize}
        \item $f$ is $(\eps_1,d)$-low-degree, or
        \item $f$ fails to be $(\eps_2,d)$-low-degree
    \end{itemize}
\end{definition}
With Hermite sampling, the idea is once again simple: use samples from the Hermite distribution and compute the empirical fraction of samples below degree $d$. This will be a low-bias estimate of the low-degree mass, using which we can distinguish the two cases. Denoting $\eps = \eps_2 - \eps_1$, this approach yields an algorithm with query complexity $\poly(1/\eps)$. Moreover, by importing the lower bound for sign functions (\cref{thm:sgn-lb}) we can prove that any classical algorithm requires $\Omega(\frac{d}{\log 1/\eps})$ queries, giving provable quantum advantage in the regime where $d$ is large.

\begin{algorithm}
    \caption{Tolerant Low-Degree Testing} \label{alg:tolerant-degree-testing}
    \begin{algorithmic}[1]
        \State Define $X = 0$ and $\eps = \eps_2 - \eps_1$.
        \State Set $m = O(\log(1/\delta)/\eps^2)$.
        \For{$i = 1$ \textbf{to} $m$}
            \State Perform Hermite sampling on $f$ to obtain $S = (v_1,...,v_n) \in \N^n$.
            \If{$|S| \leq d$}
                \State Update $X = X + 1/m$.
            \EndIf
        \EndFor
        \If{$X > (\eps_1 + \eps_2)/2$}
            \State \textbf{accept}
        \Else
            \State \textbf{reject}
        \EndIf
    \end{algorithmic}
\end{algorithm}

\begin{theorem}
    \Cref{alg:tolerant-degree-testing} makes 
    $O(\log(1/\delta)/\eps^2)$ quantum queries to $f$ and
    succeeds with probability at least $1 - \delta$.
\end{theorem}
\begin{proof}
    We will run Approximate Hermite Sampling (\cref{alg:apx-hermite-sampling}) with TV distance $\eps/10$ which requires a $\poly(n,\log d,\log (1/\eps))$ runtime (\cref{cor:apx-hermite-sampling}) and constant query complexity. Since we need to perform Hermite sampling $m$ times and each sampling takes $O(1)$ queries, the total number of queries is $O(m)$.
    
    To see that the algorithm outputs the correct answer with probability $1-\delta$ notice that $X$ would be an unbiased estimator of the distance of $f$ to being $(\eps, d)$-low-degree if we were getting perfect samples. Since we are using approximate Hermite sampling, $X$ has bias $b = \eps/10$. Therefore, by a Chernoff bound $m = O(\log(1/\delta)/\eps^2)$ independent samples suffice to get an estimate which is $(\eps/4+b)$-close with probability $1-\delta$. This distinguishes the two cases.
\end{proof}

\begin{corollary}
    Any classical algorithm for testing if $f$ is $(\eps, d)$-low-degree requires $\Omega_{\eps}(d)$ queries.
\end{corollary}
\begin{proof}
    We can get a lower bound by considering the special case of product of sign functions. Let $k = d$, then a product of $k$ sign functions is $(\eps, ck\log 1/\eps)$-low-degree by \cref{lem:sgn-hermite}. Also note that $k+2$ product functions are not $(\eps, ck\log 1/\eps)$-low-degree. Since we have a $\Omega(k)$ lower bound for distinguishing these two cases, we have a lower bound for testing if $f$ is $(\eps, ck\log 1/\eps)$-degree.
\end{proof}

\subsubsection{Tolerant testing Hermite polynomials}

For this subsection, we consider the problem of testing closeness to a product of Hermite polynomials.

\begin{definition}[Testing product of Hermite polynomials]
    Given black-box access to $f:\R^n \rightarrow \R$ and given $\eps_2 > \eps_1 > 0$ such that either
    \begin{itemize}
        \item $f(x)$ is $\eps_1$-close to $h_S$ for some $S\in \N^n$ such that $S$ has $k$ non-zero entries, or
        \item $f(x)$ is $\eps_2$-far from every Hermite polynomial on $k$ variables.
    \end{itemize}
\end{definition}

Let $\eps = \eps_2 - \eps_1$.
\begin{lemma}
    Hermite polynomials can be tolerantly tested using $O(1/\eps^2)$ quantum queries.
\end{lemma}
\begin{proof}
    If $f$ is $\eps$-close to a Hermite polynomial, then for some $S$ we have that $\wh{f}(S) \geq 1 - \eps$. If we perform Hermite sampling $O(\log 1/\eps)$ times then we see it with high probability. We can then estimate it to $\pm \eps/2$ and accept if it's at least $1 - \frac{\eps_1+\eps_2}{2}$.
\end{proof}

We conjecture that classically this problem requires $\Omega(k)$ queries like in the case of product sign functions. Note that using a classical Gaussian Goldreich-Levin algorithm, we can solve it classically using $O(n)$ queries so this would be tight for $k = \Omega(n)$.

\begin{conjecture}
    Testing Hermite polynomials on $k$ variables requires $\Omega(k)$ queries classically.
\end{conjecture}

\subsection{Agnostic learning of sparse concepts}
\label{sec:learning}
We define sparse concepts over real-valued functions, for which we give quantum and classical learning algorithms. First, we define cutoff weight.

\begin{definition}[Cutoff weight]
    For a parameter $\tau \in \R$, we define the cutoff weight of a function $f$ as \[\mathbf{W}_\tau(f) = \sum_{S:\wh{f}(S)\geq \tau} \wh{f}(S)^2\] 
\end{definition}

\begin{definition}[($s,\eps$)-sparse concepts]
\label{defn:sparse-concept}
    Let $\mathcal{C}$ be a concept comprising of functions $f:\R^n \rightarrow \R$ such that $\E_{x\sim \calN} \norm{\nabla f}_2^2 \leq \gamma$. We say that $\mathcal{C}$ is $(s, \eps)$-sparse if for all $f\in \mathcal{C}$, if there exists a $\tau(s,\eps) = \poly(\eps/s)$ such that there are at most $s$ Hermite coefficients larger than $\tau$ and $\mathbf{W}_\tau(f) \geq 1-\eps$.    
\end{definition}

As a consequence of \cref{alg:apx-hermite-sampling} we get learning algorithms for sparse concepts.

\begin{lemma}
    For any $(s,\eps)$-sparse concept $\mathcal{C}$ let $\kappa$ be the max distortion for any $f\in \mathcal{C}$, and let $\gamma$ be the max $L_2$-norm of $\nabla f$ for $f\in \mathcal{C}$. Then there is an agnostic learner with $\poly(s\kappa/\eps)$ quantum queries and $\poly(ns\gamma/\eps)$ classical queries.
\end{lemma}
\begin{proof}
    Since we know the concept is $(s,\eps)$-sparse we know that for every $f\in \mathcal{C}$ there is a $\tau$ such that running an algorithm for the Gaussian Goldreich-Levin problem gives us an improper agnostic learner for the function. We know that we can do this in $\poly(\kappa/\tau) = \poly(s\kappa/\eps)$ quantum queries and $\poly(n\gamma/\tau) = \poly(n\gamma\kappa/\eps)$ classical queries.
\end{proof}

\begin{corollary}[Learning juntas]
    Let $\mathcal{C}$ be the class of $f:\R^n \to \pmone$ such that $f$ is a polynomial on $k$ variables with individual degree 1. Then $\mathcal{C}$ can be agnostically learned with error $\varepsilon$ using $\poly(\frac{1}{\eps}^k)$ quantum queries and $\poly(n \frac{1}{\eps}^k)$ classical queries.
\end{corollary}

\section{Open problems and outlook}
\label{sec:open}
The most pressing question regarding our work is to find additional practical problems for which the QHT gives quantum advantage. To give more precise directions, we consider the following problems.

\begin{itemize}
    \item To realize the promise of quantum advantage for learning geometric concepts~\cite{KlivansOS08}, we ask: are there other examples of natural concept classes which are sparse concepts as in \cref{defn:sparse-concept}?
    \item Can the QHT be used to obtain learning algorithms for Linear Threshold Functions, or generally Polynomial Threshold Functions that beat the classical query lower bounds?
    \item Is there a large class of differential equations that can be simulated on a quantum computer efficiently using the QHT, while  being hard for classical computers? 
    \item Can other quantum systems be fast forwarded using the QHT?
\end{itemize}

\subsection*{Acknowledgements}
We thank Scott Aaronson, Robbie King, Daniel Rockmore, Troy Sewell, and Sophia Simon for helpful discussions. S. Jain was partially supported by a Berkeley CIQC grant and an Amazon AI Fellowship. VI is supported by an NSF Graduate Research Fellowship. NB is supported by the DOE Office of Science-ASCR, in particular the grant Novel Quantum Algorithms from Fast Classical Transforms.

\appendix

\section{Factorization of the QHO evolution operator} \label{sec:qho-factorization}
In this appendix we prove \cref{thm:factor}.

\factorization*

We will first prove a general algebraic claim, from which we will infer our required factorization.
\begin{claim}
Let $K_1, K_2, K_3$ be three operators obeying
\[
[K_1, K_2] = -2K_3, \quad
[K_3, K_1] = 2K_2, \quad
[K_3, K_2] = -2K_1.
\]
Then for all $t \in \mathbb{R}$,
\[
e^{(K_1+K_2)t} = e^{\alpha K_1} e^{\beta K_2} e^{\alpha K_1},
\]
where
\[
\alpha = \frac{\tan(t/\sqrt{2})}{\sqrt{2}}, 
\qquad 
\beta = \frac{\sin(\sqrt{2} t)}{\sqrt{2}}.
\]
\end{claim}
\begin{proof}
Let
\[
U(t) = e^{\alpha K_2} e^{\beta K_1} e^{\alpha K_2}.
\]
If we can prove
\[
U(t)^{-1} \frac{dU}{dt} = K_1 + K_2,
\]
then the claim follows.

By direct calculation,
\[
\frac{dU}{dt} 
= \dot{\alpha} K_2 e^{\alpha K_2} e^{\beta K_1} e^{\alpha K_2} 
+ e^{\alpha K_2} \dot{\beta} K_1 e^{\beta K_1} e^{\alpha K_2} 
+ e^{\alpha K_2} e^{\beta K_1} \dot{\alpha} K_2 e^{\alpha K_2}.
\]
Hence,
\[
U^{-1}\frac{dU}{dt} 
= e^{-\alpha K_2} e^{-\beta K_1} e^{-\alpha K_2} 
\left( \frac{dU}{dt} \right).
\]

Using the Campbell–Hausdorff lemma
\[
e^{X} Y e^{-X} = \sum_{j=0}^{\infty} \frac{1}{j!} [X,Y]_j,
\qquad
[X,Y]_0 = Y, \; [X,Y]_1 = [X,Y], \; [X,Y]_2 = [X,[X,Y]], \ldots,
\]
we compute the conjugations step by step:
\[
e^{-\alpha K_2} K_1 e^{\alpha K_2} = K_1 - 2\alpha K_3 + 2\alpha^2 K_2,
\]
\[
e^{-\beta K_1} K_2 e^{\beta K_1} = K_2 + 2\beta K_3 + 2\beta^2 K_1,
\]
\[
e^{-\alpha K_2} K_3 e^{\alpha K_2} = K_3 - 2\alpha K_2.
\]

Substituting these into the expression for $U^{-1} dU/dt$ and collecting terms yields
\[
U^{-1}\frac{dU}{dt} = A_1 K_1 + A_2 K_2 + A_3 K_3,
\]
with
\[
A_1 = 2\alpha \dot{\beta} + \dot{\beta}, \quad
A_2 = 2\dot{\alpha} - 4\dot{\alpha}\beta\alpha + 4\dot{\alpha}\beta\alpha^2 + 2\dot{\beta}\alpha^2, \quad
A_3 = 2\dot{\alpha}\beta - 4\dot{\alpha}\beta^2\alpha - 2\dot{\beta}\alpha.
\]

Differentiating $\alpha, \beta$ gives
\[
\dot{\alpha} = \tfrac{1}{2}\sec^2\!\left(\tfrac{t}{\sqrt{2}}\right),
\qquad
\dot{\beta} = \cos(t\sqrt{2}).
\]

Substituting back and simplifying trigonometrically shows that
\[
A_1 = 1, \quad A_2 = 1, \quad A_3 = 0.
\]
Thus
\[
U^{-1}\frac{dU}{dt} = K_1 + K_2,
\]
which implies
\[
U(t) = e^{(K_1+K_2)t}.
\]
\end{proof}

\begin{proof}[Proof of \cref{thm:factor}]
    Define $K_1 = \frac{\ri}{\sqrt{2}} \position^2, K_2 = \frac{\ri}{\sqrt{2}} \momentum^2, K_3 = \frac{\ri}{2} \anticommutator{\position,\momentum}$. Then we get the factorization
    \[
    \exp(\frac{i}{\sqrt{2}} \left( \position^2 +  \momentum^2 \right)) = \exp(\frac{\alpha \ri}{\sqrt{2}} \position^2) \exp(\frac{\beta \ri}{\sqrt{2}} \momentum^2) \exp(\frac{\alpha \ri}{\sqrt{2}} \position^2)
    \]

    This implies that for $\hamiltonian = \frac{1}{2}\left( \position^2 + \momentum^2 \right)$ we have

    \[
    \exp(- \ri \hamiltonian t) = \exp(\frac{ \ri \tan(t/2)\momentum^2}{2}) \exp(\frac{\ri \sin(t)\position^2}{2}) 
\exp(\frac{\ri \tan(t/2)\momentum^2}{2}). \qedhere
    \]
\end{proof}

We can similarly infer that in the discretized case, this factorization is accurate upto the weighted sum of the nested commutators.

\begin{lemma}
\label{lem:discrete factor}
    Let $\widetilde{U(t)} = \exp(\ri \frac{ \tan(t/2)\discretemomentum^2}{2}) \exp(\ri \frac{ \sin(t)\discreteposition^2}{2}) 
\exp(\ri \frac{\tan(t/2)\discretemomentum^2}{2})$, then we have that
    \begin{align*}
    \widetilde{U(t)}^{-1}\frac{d\widetilde{U(t)}}{dt} = -\ri \discretehamiltonian &+ \dot{\beta} \sum_{t=3}^{\infty}\frac{1}{t!2^{t/2}} \nestedcommutator{\beta' \ri \discreteposition^2}{\ri \discretemomentum^2}{t}+ \dot{\alpha}\sum_{t=3}^{\infty}\frac{1}{t!2^{t/2}} \nestedcommutator{\alpha' \ri \discretemomentum^2}{\ri \discreteposition^2}{t} \\ &+ \dot{\alpha}\sum_{t=2}^{\infty}\frac{1}{t!2^{t/2}} \nestedcommutator{\alpha' \ri \discretemomentum^2}{\ri \anticommutator{ \discreteposition,\discretemomentum}}{t} + \eta'
    \end{align*}

    where $\norm{\Pi_\lowenergy \eta'\Pi_\lowenergy} \leq \exp(-\lowenergy/4)$, $\alpha' = \tan(t/2)/2, \beta' = sin(t)/2, \dot{\alpha} = \tfrac{1}{2}\sec^2\!\left(\tfrac{t}{2}\right),
\dot{\beta} = \cos(t)$.
\end{lemma}

\begin{proof}
    We instantiate the operators as above but with the discretized versions, that is $K_1 = \frac{\ri}{\sqrt{2}} \discreteposition^2, K_2 = \frac{\ri}{\sqrt{2}} \discretemomentum^2, K_3 = \frac{\ri}{2} \anticommutator{\discreteposition,\discretemomentum}$. Now, the proof does not go through as is but we do the same substitutions but keep the terms that `went to 0' due to the commutator relations.

    \[
e^{-\frac{\alpha \ri}{\sqrt{2}} \discretemomentum^2} \frac{\ri}{\sqrt{2}} \discreteposition^2 e^{\frac{\alpha \ri}{\sqrt{2}} \discretemomentum^2} = K_1 - 2\alpha K_3 + 2\alpha^2 K_2 + \sum_{t=3}^{\infty}\frac{1}{t!2^{t/2}} \nestedcommutator{\alpha \ri \discretemomentum^2}{\ri \discreteposition^2}{t},
\]
\[
e^{-\frac{\ri \beta}{\sqrt{2}} \discreteposition^2} \frac{\ri}{\sqrt{2}} \discretemomentum^2 e^{\frac{\ri \beta}{\sqrt{2}} \discreteposition^2} = K_2 + 2\beta K_3 + 2\beta^2 K_1 +\sum_{t=3}^{\infty}\frac{1}{t!2^{t/2}} \nestedcommutator{\beta \ri \discreteposition^2}{\ri \discretemomentum^2}{t},
\]
\[
e^{-\frac{\alpha \ri}{\sqrt{2}} \discretemomentum^2} \frac{\ri}{2} \anticommutator{\discreteposition,\discretemomentum} e^{\frac{\alpha \ri}{\sqrt{2}} \discretemomentum^2} = K_3 - 2\alpha K_2 + \sum_{t=2}^{\infty}\frac{1}{t!2^{t/2}} \nestedcommutator{\alpha \ri \discretemomentum^2}{\ri \anticommutator{\discreteposition,\discretemomentum}}{t}.
\]

Substituting this into the expression
\[
\frac{d\widetilde{U}}{dt} 
= \dot{\alpha} K_2 e^{\alpha K_2} e^{\beta K_1} e^{\alpha K_2} 
+ e^{\alpha K_2} \dot{\beta} K_1 e^{\beta K_1} e^{\alpha K_2} 
+ e^{\alpha K_2} e^{\beta K_1} \dot{\alpha} K_2 e^{\alpha K_2} + \eta
\] yields that upto an additive term multiplied by $\eta$ with $\norm{\eta}$ exponentially small,
\begin{align*}
    \widetilde{U(t)}^{-1}\frac{d\widetilde{U(t)}}{dt} = -\ri \discretehamiltonian &+ \dot{\beta} \sum_{t=3}^{\infty}\frac{1}{t!2^{t/2}} \nestedcommutator{\beta' \ri \discreteposition^2}{\ri \discretemomentum^2}{t}+ \dot{\alpha}\sum_{t=3}^{\infty}\frac{1}{t!2^{t/2}} \nestedcommutator{\alpha' \ri \discretemomentum^2}{\ri \discreteposition^2}{t} \\ &+ \dot{\alpha}\sum_{t=2}^{\infty}\frac{1}{t!2^{t/2}} \nestedcommutator{\alpha' \ri \discretemomentum^2}{\ri \anticommutator{ \discreteposition,\discretemomentum}}{t} + \eta'.\qedhere
    \end{align*}
\end{proof}

\section{Centered discrete Fourier transform}
\label{app:cQFT}

The centered discrete Fourier transform $F \in \mathbb C^{M \times M}$ used in~\cite{somma2016}, where $M=2^m$, is defined via its entries:
\begin{align}
    F_{jk} = \frac 1 {\sqrt M} \exp \left(\ri \frac{2 \pi j k}M \right)\; ,
\end{align}
where $-M/2 \le j,k \le M/2-1$ label the rows and columns of the matrix. The difference with the standard QFT is only due to the shifting of these labels, where $0 \le j,k \le M-1$, and hence these transforms are related via cyclic permutations. In particular, if $m \ge 2$, we can write
\begin{align}
    F =  \sigma_z^{0}  . QFT. \sigma_z^{0} \;,
\end{align}
where $\sigma_z^0$ is the diagonal Pauli matrix acting on the last of the $m$ qubits; i.e.,  
$ \sigma_z^{0}=\one_2 \otimes \ldots \otimes \one_2 \otimes \sigma_z$
with $\one_2 = \begin{pmatrix}
    1 & 0 \cr 0 & 1
\end{pmatrix}$ and $\sigma_z = \begin{pmatrix}
    1 & 0 \cr 0 & -1
\end{pmatrix}$.

The complexity of $F$ is then that of the QFT plus two single qubit gates. For the schoolbook version of the QFT, this complexity is also $\cO(m^2)$. Using the quasilinear time algorithm for QFT~\cite{CW00} this can be implemented with $\widetilde{\cO}(m)$ gates.

\section{Some properties of the Plancherel-Rotach asymptotics}
\label{app:PRproperties}

For $n \ge 1$, consider the Plancherel-Rotach asymptotic functions
\begin{align}
     \tilde  \prvar_n(x) : =  \frac {2^{\frac 1 4}}{\pi^{\frac 1 2}n^{\frac 1 4}}\frac 1 {\sqrt{\sin \varphi(x)}}\left( \sin \left[\left ( \frac n 2 + \frac 1 4 \right) (\sin (2 \varphi(x))-2 \varphi(x)) + \frac{3 \pi}4\right] \right) \;,
\end{align}
in a domain ${\cal D}_c$ where $|x| \le c' \sqrt{2n+1}$ and $0\le c'<1$. Here, $\varphi(x):=\arccos(x/\sqrt{2n+1})$ satisfies
$\pi-c \ge \varphi(x) \ge c$. It follows that
$\sqrt{\sin \varphi(x)}$ is bounded from below by a constant and, in this domain, we can be simply bounde
the magnitude as
\begin{align}
    |\tilde \prvar_n(x)| =\cO(1/n^{1/4}) \;.
\end{align}

Next we consider the derivative. We obtain
\begin{align}
\nonumber
  &  \frac{\rD}{\rD x} \tilde \prvar_n(x) =  \frac {2^{\frac 1 4}}{\pi^{\frac 1 2}n^{\frac 1 4}} \left(-\frac {\cos(\varphi(x)) \frac{\rD}{\rD x} \varphi(x)} {2 (\sin \varphi(x))^{3/2}} \sin \left[\left ( \frac n 2 + \frac 1 4 \right) (\sin (2 \varphi(x))-2 \varphi(x)) + \frac{3 \pi}4\right] + \right . \\
   & \left . \frac 1 {\sqrt{\sin \varphi(x)}}
   \cos \left[\left ( \frac n 2 + \frac 1 4 \right) (\sin (2 \varphi(x))-2 \varphi(x)) + \frac{3 \pi}4\right] \left( \frac n 2 + \frac 1 4 \right)
   (\cos (2 \varphi(x))-2)  \frac{\rD}{\rD x}\varphi(x)\right)\;.
\end{align}
Since $\cos(\varphi(x))$ and $\sin(\varphi(x))$
are bounded by constants, and noting 
\begin{align}
  \frac{\rD}{\rD x}\varphi(x) =   \frac 1 {\sqrt{2n+1-x^2}}  \implies \left |  \frac{\rD}{\rD x}\varphi(x) \right| = \cO(1/n^{1/2}) \ {\rm for } \ x \in {\cal D}_c\;,
\end{align}
we obtain
\begin{align}
    \left | \frac{\rD}{\rD x} \tilde \prvar_n(x)\right| = \cO(n^{1/4})
\end{align}
in ${\cal D}_c$.

Last, we consider the magnitude of the second derivative. To this end we note
\begin{align}
  \frac{\rD^2}{\rD^2 x}\varphi(x)   = \frac {x}{(2n+1-x^2)^{3/2}} \implies \left |  \frac{\rD^2}{\rD^2 x}\varphi(x) \right| = \cO(1/n) \ {\rm for } \ x \in {\cal D}_c\;.
\end{align}
Using the chain rule it is possible to show that 
\begin{align}
    \left | \frac{\rD^2}{\rD^2 x} \tilde \prvar_n(x) \right| = \cO(n^{3/4})
\end{align}
in ${\cal D}_c$.
The term that determines this upper bound is a term 
that contains
\begin{align}
 \frac {2^{\frac 1 4}}{\pi^{\frac 1 2}n^{\frac 1 4}}   \left ( \frac n 2 + \frac 1 4 \right)^2 \left( \frac{\rD}{\rD x}\varphi(x)\right)^2 = \cO(n^{3/4})\;.
\end{align}

\section{Gaussian Goldreich-Levin}
\label{sec:ggl}
In this section, we study a fundamental learning problem first introduced by Goldreich-Levin~\cite{GL89} in the context of cryptography.

\begin{definition}[Goldreich-Levin problem, informal]
   Given query access to $f: \mathbb{F}_2^n \rightarrow \mathbb{F}_2$,  output a list of all the linear functions which are $\tau$-correlated with $f$. 
\end{definition}

It is known that one can solve the Goldreich-Levin problem with constant success probability with $O(n/\tau^2)$ classical queries \cite{Goldreich_2001}. The factor $1/\tau^2$ is optimal since Parseval's only gaurantees that the list $L$ of $\tau$-correlated functions is of size $O(1/\tau^2)$. In the quantum case, we can do much better and get $O(1/\tau^2)$ queries.
Here, we ask a similar question but for functions over the domain $\mathbb{R}^n$.

\begin{definition}[Gaussian Goldreich-Levin]
    Given query access to a function $f: \mathbb{R}^n \rightarrow \mathbb{R}$, return a list $L$ of indices of Hermite polynomials such that
    \begin{itemize}
        \item $|\wh{f}(S)| \geq \tau \implies S\in L$
        \item $S \in L \impliedby |\wh{f}(S)| \geq \tau/2$
    \end{itemize}
\end{definition}

The main claim of this section is that this real-analogue of the Goldreich-Levin problem is also tractable. To prove this, we will need to build some machinery analogous to the classical Goldreich-Levin algorithm, a great exposition for which can be found in the book by O'Donnell \cite{ODonnell_2014}. We now prove an analogue of the Godlreich-Levin theorem for Hermite coefficients. Below $\gamma_f = \E_{x\sim \calN} \norm{f(x)}_2$, and we drop the subscript $f$ when it's clear from context.

\begin{theorem}[Gaussian Goldreich-Levin]
    Given query access to a function $f: \mathbb{R}^n \rightarrow \mathbb{R}$, \cref{alg:goldreich-levin} runs in time $\poly(n\gamma/\tau)$. It returns a list $L$ of indices of Hermite polynomials such that
    \begin{itemize}
        \item $|\hat{f}(S)| \geq \tau \implies S\in L$
        \item $S \in L \impliedby |\hat{f}(S)| \geq \tau/2$
    \end{itemize}
\end{theorem}

\begin{algorithm}[ht]
\caption{Gaussian Goldreich-Levin learning algorithm}\label{alg:goldreich-levin}
\begin{algorithmic}[1]
\State Set $m = O(\frac{\gamma^2}{\tau})$
\State $L \gets (*,*,\dots,*)$
\For{$k = 1$ \textbf{to} $n$}
    \For{\textbf{each} $S \in L$, $S = (a_1, \dots, a_{k-1}, *, \dots, *)$}
        \For{$a_k = 0 \text{ to } m$}
            \State Let $S_{a_k} = (a_1, \dots, a_{k-1}, a_k, *, \dots, *)$
            \State Estimate $\mathbf{W}^{S_{a_k}}(f)$ to within $\pm \tau^2/4$ with probability at least $1 - \delta$\hfill \cref{lem:weight-estm}
            \If{the estimate of $\mathbf{W}^{S_{a_k}}(f)$ is at least $\tau^2/2$}
                \State Add $S_{a_k}$ to $L$
            \EndIf
        \EndFor
    \State Remove $S$ from $L$
    \EndFor
\EndFor

\State \Return $L$
\end{algorithmic}
\end{algorithm}

\begin{proof}
    Because $(\tau/2)^2 + \tau^2/4 \leq \tau^2/2$, the only Hermite coefficients in $L$ are ones which are at least $\tau/2$. Since $\tau^2 - \tau^2/4 > \tau^2/2$, all Hermite coefficients which are at least $\tau$ are in the set $L$. 
    The runtime can be analysed by multiplying the runtime of all three loops and the complexity of estimation.
    \begin{itemize}
        \item The outermost loop runs for $n$ steps,
        \item The middle loop runs in $O(1/\tau^2)$ steps since at any given point $|L| \leq O(1/\tau^2)$,
        \item The runtime of the inner loop is $m = \theta(\gamma^2/\tau)$,
        \item The query complexity of estimation is $\poly(\gamma/\tau)$ by \cref{lem:weight-estm}.
    \end{itemize}
    Thus, the overall query complexity and runtime is $\poly(n\gamma/\tau)$.
\end{proof}

Note that while our algorithm assumes knowledge of $\gamma_f$, one can always estimate $\gamma_f$ by simulating gradient queries.

\DeclareUrlCommand{\Doi}{\urlstyle{sf}}
\renewcommand{\path}[1]{\small\Doi{#1}}
\renewcommand{\url}[1]{\href{#1}{\small\Doi{#1}}}
\bibliographystyle{alphaurl}
\bibliography{refs}

\end{document}